\newif\ifarXiv
\arXivtrue % comment out to switch to Entropy

\ifarXiv

\documentclass[ aps, amsmath, amssymb, pre, superscriptaddress, showpacs, showkeys, numerical,
nofootinbib, % twocolumn,
notitlepage]{revtex4-1}

\pdfoutput=1

\usepackage{amsthm}

 \newcounter{theorem}
 \newtheorem{Theorem}[theorem]{Theorem}
 
 \newcounter{lemma}
 \newtheorem{Lemma}[lemma]{Lemma}
 
 \newcounter{corollary}
 \newcounter{proposition}
 \newtheorem{Proposition}[proposition]{Proposition}
 
 \newcounter{characterization}
 \newcounter{property}
 \newtheorem{Property}[property]{Property}
 
 \newcounter{problem}
 \newcounter{example}
 \newcounter{examplesanddefinitions}
 \newcounter{remark}
 \newtheorem{Remark}[remark]{Remark}
 
 \newcounter{definition}
 \newtheorem{Definition}[definition]{Definition}
 
 \newcounter{hypothesis}
 \newcounter{notation}
\else

\documentclass[entropy,article,submit,moreauthors,pdftex,10pt,a4paper]{mdpi}

\firstpage{1} 
\articlenumber{x}
\doinum{10.3390/------}
\pubvolume{xx}
\pubyear{2017}
\copyrightyear{2017}
\externaleditor{Academic Editor: name}
\history{Received: date; Accepted: date; Published: date}

\newcommand*\patchAmsMathEnvironmentForLineno[1]{%
  \expandafter\let\csname old#1\expandafter\endcsname\csname #1\endcsname
  \expandafter\let\csname oldend#1\expandafter\endcsname\csname end#1\endcsname
  \renewenvironment{#1}%
     {\linenomath\csname old#1\endcsname}%
     {\csname oldend#1\endcsname\endlinenomath}}% 
\newcommand*\patchBothAmsMathEnvironmentsForLineno[1]{%
  \patchAmsMathEnvironmentForLineno{#1}%
  \patchAmsMathEnvironmentForLineno{#1*}}%
\AtBeginDocument{%
\patchBothAmsMathEnvironmentsForLineno{equation}%
\patchBothAmsMathEnvironmentsForLineno{align}%
\patchBothAmsMathEnvironmentsForLineno{flalign}%
\patchBothAmsMathEnvironmentsForLineno{alignat}%
\patchBothAmsMathEnvironmentsForLineno{gather}%
\patchBothAmsMathEnvironmentsForLineno{multline}%
}

\fi

\def\iccs{I_\text{\textsc{ccs}}}
\def\imin{I_\text{min}}
\def\ired{I_\text{red}}
\def\uitilde{\widetilde{UI}}
\def\ivk{\mathcal{S}_\text{VK}}

\newcommand{\fig}[1]{Figure~\ref{fig:#1}}
\newcommand{\tb}[1]{Table~\ref{tb:#1}}
\newcommand{\secRef}[1]{Section~\ref{sec:#1}}
\newcommand{\appRef}[1]{Appendix~\ref{app:#1}}
\newcommand{\eq}[1]{(\ref{eq:#1})}
\newcommand{\cor}[1]{Corollary~\ref{cor:#1}}

\newcommand{\by}{\!\times\!}

\usepackage{multirow}
\usepackage{lipsum}
\usepackage{tikz}
\usetikzlibrary{patterns}

\usepackage[nice]{nicefrac}

\usepackage{mathtools}

\usepackage{thmtools}
\usepackage{thm-restate}

\makeatletter

\renewenvironment{thmt@restatable}[3][]{%
  \thmt@toks{}% will hold body
  \stepcounter{thmt@dummyctr}% used for data storage label.
  \long\def\thmrst@store##1{%
    \@xa\gdef\csname #3\endcsname{%
      \@ifstar{%
        \thmt@thisistheonefalse\csname thmt@stored@#3\endcsname
      }{%
        \thmt@thisistheonetrue\csname thmt@stored@#3\endcsname
      }%
    }%
    \@xa\long\@xa\gdef\csname thmt@stored@#3\@xa\endcsname\@xa{%
      \begingroup
      \ifthmt@thisistheone
      \else
        \@xa\protected@edef\csname the#2\endcsname{%
          \thmt@trivialref{thmt@@#3}{??}}%
        \ifcsname r@thmt@@#3\endcsname\else
          \G@refundefinedtrue
        \fi
        \@xa\let\csname c@#2\endcsname=\c@thmt@dummyctr
        \@xa\let\csname theH#2\endcsname=\theHthmt@dummyctr
        \let\label=\@gobble
        \let\ltx@label=\@gobble% amsmath needs this
        \def\thmt@restorecounters{}%
        \@for\thmt@ctr:=\thmt@innercounters\do{%
          \protected@edef\thmt@restorecounters{%
            \thmt@restorecounters
            \protect\setcounter{\thmt@ctr}{\arabic{\thmt@ctr}}%
          }%
        }%
        \thmt@trivialref{thmt@@#3@data}{}%
      \fi
      \ifthmt@restatethis
        \thmt@restatethisfalse
      \else
        \csname #2\@xa\endcsname\ifx\@nx#1\@nx\else[{#1}]\fi
      \fi
      \ifthmt@thisistheone
         \thmt@rst@storecounters{#3}%
        \label{thmt@@#3}%
      \fi
      ##1%
      \csname end#2\endcsname
      \ifthmt@thisistheone\else\thmt@restorecounters\fi
      \endgroup
    }% thmt@stored@#3
    \csname #3\@xa\endcsname\ifthmt@thisistheone\else*\fi
    \@xa\end\@xa{\@currenvir}
  }% thm@rst@store
  \thmt@collect@body\thmrst@store
}{%
}

\makeatother

\usepackage{mathrsfs}

\usepackage{longtable}

\usepackage[textwidth=45pt, textsize=scriptsize, backgroundcolor=yellow!20!white, linecolor=blue,
]{todonotes}

\def\nn{\nonumber}

\newcommand{\paperTitle}{Pointwise Partial Information Decomposition\\using the Specificity and
  Ambiguity Lattices}

\newcommand{\theKeywords}{mutual information; pointwise information; information decomposition;
  unique information; redundant information; complementary information; redundancy; synergy}

\newcommand{\paperAbstract}{What are the distinct ways in which a set of predictor variables can
  provide information about a target variable?  When does a variable provide unique information,
  when do variables share redundant information, and when do variables combine synergistically to
  provide complementary information?  The redundancy lattice from the partial information
  decomposition of Williams and Beer provided a promising glimpse at the answer to these questions.
  However, this structure was constructed using a much criticised measure of redundant information,
  and despite sustained research, no completely satisfactory replacement measure has been proposed.
  In this paper, we take a different approach, applying the axiomatic derivation of the redundancy
  lattice to a single realisation from a set of discrete variables.  To overcome the difficulty
  associated with signed pointwise mutual information, we apply this decomposition separately to the
  unsigned entropic components of pointwise mutual information which we refer to as the specificity
  and ambiguity.  This yields a separate redundancy lattice for each component.  Then based upon an
  operational interpretation of redundancy, we define measures of redundant specificity and
  ambiguity enabling us to evaluate the partial information atoms in each lattice.  These atoms can
  be recombined to yield the sought-after multivariate information decomposition.  We apply this
  framework to canonical examples from the literature and discuss the results and the various
  properties of the decomposition. In particular, the pointwise decomposition using specificity and
  ambiguity satisfies a chain rule over target variables, which provides new insights into the
  so-called two-bit-copy example.}

\newcommand{\thePACS}{89.70.Cf, 89.75.Fb, 05.65.+b, 87.19.lo}

\ifarXiv

\begin{document}

\title{\paperTitle}

\author{Conor Finn}
\altaffiliation[Also at ]{CSIRO Data61, Marsfield NSW 2122, Australia.}
\email{conor.finn@sydney.edu.au}

\author{Joseph T.\ Lizier}

\affiliation{ Complex Systems Research Group and Centre for Complex Systems,\\
  Faculty of Engineering \& IT, The University of Sydney, NSW 2006, Australia.}

\date{\today}

\begin{abstract}
  \paperAbstract{}
\end{abstract}

\pacs{\thePACS{}}

\keywords{\theKeywords{}}

\maketitle

\else

%=====================================================
% Document title setting for Entropy:
%=====================================================

\Title{\paperTitle}

\Author{Conor Finn $^{1,2,*}$, and Joseph T. Lizier$^{1}$}

\AuthorNames{Conor Finn, and Joseph T. Lizier}

\address{%
  $^{1}$ \quad Complex Systems Research Group and Centre for Complex Systems, Faculty of Engineering
  \& IT, The University of Sydney, NSW 2006, Australia.\\ % ; e-mail@e-mail.com\\
  $^{2}$ \quad CSIRO Data61, Marsfield NSW 2122, Australia.} %; e-mail@e-mail.com}

\corres{Correspondence: conor.finn@sydney.edu.au}

\abstract{\paperAbstract{}}

\keyword{\theKeywords{}}

% The fields PACS, MSC, and JEL may be left empty or commented out if not applicable
\PACS{\thePACS{}}
%\MSC{}
%\JEL{}

\begin{document}

\fi

\newcommand{\ob}[1]{{#1}^{\mathsf{c}}}
\newcommand{\ra}{\!\rightarrow\!}
\newcommand{\sm}{\!\setminus\!}
\newcommand{\mc}[1]{\mathcal{#1}}
\newcommand{\msc}[1]{\mathscr{#1}}
\newcommand{\mb}[1]{\boldsymbol{#1}}
\newcommand{\nf}[2]{\nicefrac{#1}{#2}}

\def\UIo{{U(S_1 \sm S_2 \ra T)}}
\def\UIop{{U^+(S_1 \sm S_2 \ra T)}}
\def\UIon{{U^-(S_1 \sm S_2 \ra T)}}
\def\uio{{u(s_1 \sm s_2 \ra t)}}
\def\uiop{{u^+(s_1 \sm s_2 \ra t)}}
\def\uion{{u^-(s_1 \sm s_2 \ra t)}}
\def\UIt{{U(S_2 \sm S_1 \ra T)}}
\def\UItp{{U^+(S_2 \sm S_1 \ra T)}}
\def\UItn{{U^-(S_2 \sm S_1 \ra T)}}
\def\uit{{u(s_2 \sm s_1 \ra t)}}
\def\uitp{{u^+(s_2 \sm s_1 \ra t)}}
\def\uitn{{u^-(s_2 \sm s_1 \ra t)}}
\def\RI{{R(S_1,S_2 \ra T)}}
\def\RIp{{R^+(S_1,S_2 \ra T)}}
\def\RIn{{R^-(S_1,S_2 \ra T)}}
\def\ri{{r(s_1,s_2 \ra t)}}
\def\rip{{r^+(s_1, s_2 \ra t)}}
\def\rin{{r^-(s_1, s_2 \ra t)}}
\def\CI{{C(S_1,S_2 \ra T)}}
\def\CIp{{C^+(S_1,S_2 \ra T)}}
\def\CIn{{C^-(S_1,S_2 \ra T)}}
\def\ci{{c(s_1,s_2 \ra t)}}
\def\cip{{c^+(s_1, s_2 \ra t)}}
\def\cin{{c^-(s_1, s_2 \ra t)}}

\def\IST{I(S;T)}
\def\ist{i(s;t)}
\def\ISoT{I(S_1;T)}
\def\isot{i(s_1;t)}
\def\IStT{I(S_2;T)}
\def\istt{i(s_2;t)}
\def\ISotT{I(S_{1,2};T)}
\def\ISoTct{I(S_1;T|S_2)}
\def\isott{i(s_{1,2};t)}
\def\ISonT{I(S_1, \ldots ,S_n;T)}

\def\istd{i(s_1 \ra t)}
\def\ipst{i^+(s_1 \ra t)}
\def\inst{i^-(s_1 \ra t)}
\def\ipsot{i^+\big(s_1 \ra t\big)}
\def\insot{i^-\big(s_1 \ra t\big)} 
\def\ipstt{i^+\big(s_2 \ra t\big)}
\def\instt{i^-\big(s_2 \ra t\big)}

\newcommand{\xor}{\textup{X}\textsc{or}}
\newcommand{\mathxor}{\text{\sffamily~XOR~}}
\newcommand{\andgate}{A\textsc{nd}}
\newcommand{\locunq}{\textup{P}\textsc{w}\textup{U}\textsc{nq}}
\newcommand{\imprdn}{\textup{R}\textsc{dn}\textup{E}\textsc{rr}}
\newcommand{\tbc}{\textup{T}\textsc{bc}}
\newcommand{\tbp}{\textup{T}\textsc{bep}}
\newcommand{\unq}{\textup{U}\textsc{nq}}
\newcommand{\and}{\textup{A}\textsc{nd}}

\newcommand{\bit}{\text{~bit}}

%%%%%%%%%%%%%%%%%%%%%%%%%%%%%%%%%%%%%%%%%%%%%%%%%%%%%%%%%%%%%%%%%%%%%%%%%%%%%%%%%%%%%%%%%%%%%%%%%%%%
%%%%%%%%%%%%%%%%%%%%%%%%%%%%%%%%%%%%%%%%%%%%%%%%%%%%%%%%%%%%%%%%%%%%%%%%%%%%%%%%%%%%%%%%%%%%%%%%%%%%

\section{Introduction}
\label{sec:intro}

The aim of information decomposition is to divide the total amount of information provided by a set
of predictor variables, about a target variable, into atoms of partial information contributed
either individually or jointly by the various subsets of the predictors.  Suppose that we are trying
to predict a target variable $T$, with discrete state space $\mc{T}$, from a pair of predictor
variables $S_1$ and $S_2$, with discrete state spaces $\mc{S}_1$ and $\mc{S}_2$.  The mutual
information $I(S_1;T)$ quantifies the information $S_1$ individually provides about $T$.  Similarly,
the mutual information $I(S_2;T)$ quantifies the information $S_2$ individually provides about $T$.
Now consider the joint variable $S_{1,2}$ with the state space~${\mc{S}_1 \by \mc{S}_2}$.  The
(joint) mutual information $I(S_{1,2};T)$ quantifies the total information $S_1$ and $S_2$ together
provide about $T$.  Although Shannon's information theory provides the prior three measures of
information, there are four possible ways $S_1$ and $S_2$ could contribute information about $T$:\
the predictor $S_1$ could uniquely provide information about $T$; or the predictor $S_2$ could
uniquely provide information about~$T$; both $S_1$ and $S_2$ could both individually, yet
redundantly, provide the same information about $T$; or the predictors $S_1$ and $S_2$ could
synergistically provide information about $T$ which is not available in either predictor
individually.  Thus we have the following underdetermined set of equations,
\begin{alignat}{3}
  \label{eq:decomp_bivar}
  \ISotT &&= \RI &+ \UIo + \UIt + \CI,  \nn\\
  \ISoT  &&= \RI &+ \UIo,               \nn\\
  \IStT  &&= \RI &+ \UIt,
\end{alignat}
where $\UIo$ and $\UIt$ are the unique information provided by $S_1$ and $S_2$ respectively, $\RI$
is the redundant information, and $\CI$ is the synergistic or complementary information.  (The
directed notation is utilise here to emphasis the privileged role of the variable $T$.) Together,
the equations in \eq{decomp_bivar} form the bivariate information decomposition.  The problem is to
define one of the unique, redundant or complementary information---something not provided by
Shannon's information theory---in order to uniquely evaluate the decomposition.

Now suppose that we are trying to predict a target variable $T$ from a set of $n$ finite state
predictor variables ${\mb{S}=\{S_1,\ldots,S_n\}}$.  In this general case, the aim of information
decomposition is to divide the total amount of information $I(S_1,\ldots,S_n;T)$ into atoms of
partial information contributed either individually or jointly by the various subsets of $\mb{S}$.
But what are the distinct ways in which these subsets of predictors might contribute information
about the target?  Multivariate information decomposition is more involved than the bivariate
information decomposition because it is not immediately obvious how many atoms of information one
needs to consider, nor is it clear how these atoms should relate to each other.  Thus the general
problem of information decomposition is to provide both a structure for multivariate information
which is consistent with the bivariate decomposition, and a way to uniquely evaluate the atoms in
this general structure.

In the remainder of \secRef{intro}, we will introduce an intriguing framework called partial
information decomposition (PID), which aims to address the general problem of information
decomposition, and highlight some of the criticisms and weaknesses of this framework.  In
\secRef{localisability}, we will consider the underappreciated pointwise nature of information and
discuss the relevance of this to the problem of information decomposition.  We will then propose a
modified pointwise partial information decomposition (PPID), but then quickly repudiate this
approach due to complications associated with decomposing the signed pointwise mutual information.
In \secRef{same_info}, we will discuss circumventing this issue by examining information on a more
fundamental level, in terms of the unsigned entropic components of pointwise mutual information
which we refer to as the specificity and the ambiguity.  Then in \secRef{local_framework}---the main
section of this paper---we will introduce the PPID using the specificity and ambiguity lattices and
the measures of redundancy in Definitions~\ref{def:specificity} and \ref{def:ambiguity}.  In
\secRef{discussion}, we will apply this framework to a number of canonical examples from the PID
literature, discuss some of the key properties of the decomposition, and compare these to existing
approaches to information decomposition.  \secRef{conclusion} will conclude the main body of the
paper.  \appRef{operational} contains discussions regarding the so-called two-bit-copy problem in
terms of Kelly gambling, \appRef{tech} contains many of the technical details and proofs, while
\appRef{tech} contains some more examples.

\ifarXiv
%\begin{widetext}
  \vspace{-10pt}
\else
% blank
\fi
\subsection{Notation}
  \ifarXiv
%  \vspace{-5pt}
  \else
  % blank
  \fi
  
  The following notational conventions are observed throughout this article:
  {
  \ifarXiv
    \setlength{\tabcolsep}{6pt}
    \def\arraystretch{1.25}
    \setlength{\LTpre}{4pt} \setlength{\LTpost}{6pt}
    \begin{longtable*}{p{60pt} p{400pt}}
  \else
    \small
    \setlength{\LTpre}{1pt} \setlength{\LTpost}{6pt}
    \begin{longtable*}{p{55pt} p{355pt}}
  \fi
    \raggedleft $T$, $\mc{T}$, $t$, $\ob{t}$,
    & denote the \emph{target} variable, event space, event and complementary event respectively; \\
    \raggedleft $S$, $\mc{S}$, $s$, $\ob{s}$,
    & denote the \emph{predictor} variable, event space, event and complementary event
      respectively; \\
    \raggedleft $\mb{S}$, $\mb{s}$,
    & represent the \emph{set} of $n$ predictor variables ${\{S_1,\ldots,S_n\}}$ and events
      ${\{s_1,\ldots,s_n\}}$ respectively;  \\
    \raggedleft $\mc{T}^t$, $\mc{S}^s$,
    & denote the \emph{two-event partition} of the event space, i.e.\ ${\mc{T}^t=\{t,\ob{t}\}}$ and
      ${\mc{S}^s=\{s,\ob{s}\}}$; \\              
    \raggedleft $H(T)$, $I(S;T)$,
    & uppercase function names be used for \emph{average} information-theoretic measures; \\
    \raggedleft $h(t)$, $i(s,t)$,
    & lowercase function names be used for \emph{pointwise} information-theoretic measures. \\ 
  \end{longtable*}
  \vspace{2pt}
  \noindent When required, the following index conventions are observed:
    \ifarXiv
    \begin{longtable*}{p{60pt} p{400pt}}  
    \else
    \begin{longtable*}{p{55pt} p{355pt}}  
    \fi
      \raggedleft $s^1$, $s^2$, $t^1$, $t^2$ 
    & superscripts distinguish between different \emph{different events} in a
      variable; \\
    \raggedleft $S_1$, $S_2$, $T_1$, $T_2$ 
    & subscripts distinguish between \emph{different variables}; \\
    \raggedleft $S_{1,2}$, $s_{1,2}$ 
    & multiple superscripts represent \emph{joint variables} and \emph{joint events}. \\
  \end{longtable*}
  \vspace{2pt}
  \noindent Finally, to be discussed in more detail when appropriate, consider the following: 
    \ifarXiv
    \begin{longtable*}{p{60pt} p{400pt}}  
    \else
    \begin{longtable*}{p{55pt} p{360pt}}
    \fi
    \raggedleft $\mb{A}_1,\ldots,\mb{A}_k$
    & \emph{sources} are sets of predictor variables, i.e.\ $\mb{A}_i \!\in\! \msc{P}_1(\mb{S})$ where
      $\msc{P}_1$ is the power set without $\emptyset$; \\
    \raggedleft $\mb{a}_1,\ldots,\mb{a}_k$
    & \emph{source events} are sets of predictor events, i.e.\ $\mb{a}_i \in \msc{P}_1(\mb{s})$. \\
  \end{longtable*}
}

\ifarXiv
%\end{widetext}
\else
% blank
\fi

\subsection{Partial Information Decomposition}
\newtheorem{wbaxiom}{W\&B Axiom}

The \emph{partial information decomposition}~(PID) of Williams and Beer~\citep{williams2010,
  williams2010priv} was introduced to address the problem of multivariate information decomposition.
The approach taken is appealing as rather than speculating about the structure of multivariate
information, Williams and Beer took a more principled, axiomatic approach.  First they consider
potentially overlapping subsets of $\mb{S}$ called sources, denoted $\mb{A}_1,\ldots,\mb{A}_k$.
Then they examine the various ways these sources might contain the same information.  Formally, they
introduce three axioms which ``any reasonable measure for redundant information [$I_\cap$] should
fulfil''~\citep[p.\ 3502]{olbrich2015}.\footnote{These axioms appear explicitly in
  \citep{williams2010priv} but are discussed in \citep{williams2010} as mere properties.  A
  published version of the axioms can be found in \citep{lizier2013}.}

\ifarXiv
  \newpage
\else
  % blank
\fi
  
\begin{wbaxiom}[Commutativity]
  \label{ax:wb_comm}
  Redundant information is invariant under any permutation $\sigma$ of sources,
  \begin{equation*}
    I_\cap\big(\mb{A}_1,\ldots,\mb{A}_k \ra T\big)
    = I_\cap\big(\sigma(\mb{A}_1),\dots,\sigma(\mb{A}_k) \ra T\big).
  \end{equation*}
\end{wbaxiom}
\begin{wbaxiom}[Monotonicity]
  \label{ax:wb_mono}
  Redundant information decreases monotonically as more sources are included,
  \begin{equation*}
    I_\cap\big(\mb{A}_1,\ldots,\mb{A}_{k-1} \ra T\big)
    \leq I_\cap\big(\mb{A}_1,\ldots,\mb{A}_k \ra T\big)
  \end{equation*}
  with equality if $\mb{A}_k \supseteq \mb{A}_i$ for any
  $\mb{A}_i \in \{\mb{A}_1,\ldots,\mb{A}_{k-1} \}$.
\end{wbaxiom}
\begin{wbaxiom}[Self-redundancy]
  \label{ax:wb_sr}
  Redundant information for a single source $\mb{A}_i$ equals the mutual information,
  \begin{equation*}
    I_\cap\big(\mb{A}_i \ra T\big) = I\big(\mb{A}_i \,; T\big).
  \end{equation*}
\end{wbaxiom}

These axioms are based upon the intuition that redundancy should be analogous to the set-theoretic
notion of intersection (which is commutative, monotonically decreasing and idempotent).  Crucially,
Axiom~\ref{ax:wb_sr} ties this notion of redundancy to Shannon's information theory.  In addition to
these three axioms, there is a (implicit) axiom being assumed here known as \emph{local
  positivity}~\citep{bertschinger2013}, which is the requirement that all atoms be non-negative.
Williams and Beer \citep{williams2010, williams2010priv} then show how these axioms reduce the
number of sources to the collection of sources such that no source is a superset of any other.
These remaining sources are called \emph{partial information atoms} (PI~atoms).  Each PI atom
corresponds to a distinct way the set of predictors $\mb{S}$ can contribute information about the
target $T$.  Furthermore, Williams and Beer show that these PI atoms are partially ordered and hence
form a lattice which they call the \emph{redundancy lattice}.  (\fig{lattices} depicts the
redundancy lattices for bivariate and trivariate cases.)  For the bivariate case, the redundancy
lattice recovers the decomposition \eq{decomp_bivar}, while in the multivariate case it provides a
meaningful structure for decomposition of the total information provided by an arbitrary number of
predictor variables.

While the redundancy lattice of PID provides a structure for multivariate information decomposition,
it does not uniquely determine the value of the PI atoms in the lattice.  To do so requires a
definition of a measure of redundant information which satisfies the above axioms.  Hence, in order
to complete the PID framework, Williams and Beer simultaneously introduced a measure of redundant
information called $\imin$ which quantifies redundancy as the minimum information that any source
provides about a target event $t$, averaged over all possible events from~$T$.  However, not long
after its introduction $\imin$ was heavily criticised. Firstly, $\imin$ does not distinguish between
``whether different random variables carry the \emph{same} information or just the \emph{same
  amount} of information''~\citep[p.\ 269]{bertschinger2013} (see also \citep{harder2013,
  griffith2014}).  Secondly, $\imin$ does not possess the target chain rule introduced by
\citet{bertschinger2013} (under the name left chain rule).  This latter point is problematic as the
target chain rule is a natural generalisation of the chain rule of mutual information---i.e.\ one of
the fundamental, and indeed characterising, properties of information in Shannon's
theory~\citep{shannon1948,fano1961}.

These issues with $\imin$ prompted much research attempting to find a suitable replacement measure
compatible with the PID framework.  Using the methods of information geometry, \citet{harder2013}
focused on a definition of redundant information called $\ired$ (see also \citep{harder2013phd}).
\citet{bertschinger2014} defined a measure of unique information $\uitilde$ based upon the notion
that if one variable contains unique information then there must be some way to exploit that
information in a decision problem.  \citet{griffith2012} used an entirely different motivation to
define a measure of synergistic information $\ivk$ whose decomposition transpired to be equivalent
to that of $\uitilde$ \citep{bertschinger2014}.  Despite this effort, none of these proposed
measures are entirely satisfactory.  Firstly, just as for $\imin$, none of these proposed measures
possess the target chain rule.  Secondly, these measures are not compatible with the PID framework
in general, but rather are only compatible with PID for the special case of bivariate predictors,
i.e.\ the decomposition~\eq{decomp_bivar}. This is because they all simultaneously satisfy the
Williams and Beers axioms, local positivity, and the \emph{identity property} introduced by
\citet{harder2013}.  In particular, \citet{rauh2014} proved that no measure satisfying the identity
property and the Williams and Beer Axioms~\ref{ax:wb_comm}--\ref{ax:wb_sr} can yield a non-negative
information decomposition beyond the bivariate case of two predictor variables.  In addition to
these proposed replacements for $\imin$, there is also a substantial body of literature discussing
either PID, similar attempts to decompose multivariate information, or the problem of information
decomposition in general \citep{bertschinger2013, rauh2014, olbrich2015, perrone2016, harder2013phd,
  griffith2014, griffith2015, rosas2016, barrett2015, lizier2013, ince2017, ince2017ped,
  chicharro2017, rauh2017a, rauh2017b, rauh2017c, faes2017, pica2017, james2017, makkeh2017,
  kay2017}.  Furthermore, the current proposals have been applied to various problems in
neuroscience \citep{angelini2010, stramaglia2016, ghazi2017, maity2017, tax2017, wibral2017}.
Nevertheless (to date), there is no generally accepted measure of redundant information that is
entirely compatible with PID framework, nor has any other well-accepted multivariate information
decomposition emerged.

To summarise the problem, we are seeking a meaningful decomposition of the information provided an
arbitrarily large set of predictor variables about a target variable, into atoms of partial
information contributed either individually or jointly by the various subsets of the predictors.
Crucially, the redundant information must capture when two predictor variables are carrying the same
information about the target, not merely the same amount of information.  Finally, any proposed
measure of redundant information should satisfy the target chain rule so that net redundant
information can be consistently computed for consistently for multiple target events.

%%%%%%%%%%%%%%%%%%%%%%%%%%%%%%%%%%%%%%%%%%%%%%%%%%%%%%%%%%%%%%%%%%%%%%%%%%%%%%%%%%%%%%%%%%%%%%%%%%%%
%%%%%%%%%%%%%%%%%%%%%%%%%%%%%%%%%%%%%%%%%%%%%%%%%%%%%%%%%%%%%%%%%%%%%%%%%%%%%%%%%%%%%%%%%%%%%%%%%%%%

\section{Pointwise Information Theory}
\label{sec:localisability}

Although underappreciated in the current reference texts on information theory \citep{cover2012,
  mackay2003}, both the entropy and mutual information can be derived from first principles as
fundamentally \emph{pointwise} quantities---that is, as quantities which measure the information
content of individual events rather than entire variables.\footnote{The term \emph{pointwise mutual
    information} has only recently become typical.  Perhaps the term event-wise would provide a more
  apt description; however, the usage is not typical.  \citet{woodward1953} and \citet{fano1961}
  both referred to it as the \emph{mutual information} and then explicitly prefix the \emph{average}
  mutual information.  Some literature, typically in the context of time-series analysis, refer to
  it as the \emph{local} mutual information, e.g.\citep{lizier2013,ince2017}.}  The pointwise
entropy $h(t)=-\log p(t)$ quantifies the information content of a single event $t$, while the
pointwise mutual information,
\begin{equation}
  i(s;t) = \log \frac{p(t|s)}{p(t)}
           = \log \frac{p(s,t)}{p(s)p(t)}
           = \log \frac{p(s|t)}{p(s)},
\end{equation}
quantifies the information provided by $s$ about $t$, or vice versa. The usual (average) entropy
and (average) mutual information can be recovered by taking the expectation over all events from the
relevant variables, i.e.\ $H(T)=\big\langle h(t) \big\rangle$ and
$I(S;T)=\big\langle i(s;t) \big\rangle$.

To our knowledge, this pointwise notion of information was first considered by Woodward and Davies
\citep{woodward1952, woodward1953} who noted that average form of Shannon's entropy ``tempts one to
enquire into other simpler methods of derivation [of the per state
entropy]''~\citep[p.~51]{woodward1953}.  Indeed, they showed that the pointwise entropy and
pointwise mutual information can both be derived from just two axioms concerning the addition of the
information provided by the occurrence of individual events~\citep{woodward1952}.  \citet{fano1961}
formalised their idea further by deriving the pointwise mutual information and pointwise entropy
from four postulates which ``should be satisfied by a useful measure of information''~\citep[p.\
31]{fano1961}.  This bottom-up approach of first deriving the pointwise quantities and then taking
the expectation over these quantities yields the same quantities as Shannon's top-down method of
directly defining the average quantities.  Although both approaches arrive at the same (average)
quantities, Shannon's treatment obfuscates the pointwise nature of the fundamental quantities---in
contrast to Fano's treatment which makes it manifestly obvious.

The relevance of this pointwise nature of information to the problem of information decomposition
will be established and discussed in detail in the next section (\secRef{local_pid}). However,
before continuing, it is important to note that---in contrast to the (average) mutual
information---the pointwise mutual information is not non-negative.  Positive pointwise information
corresponds to the predictor event $s$ raising the probability $p(t|s)$ relative to the prior
probability $p(t)$.  Hence when the event $t$ occurs it can be said that the event $s$ was
\emph{informative} about the event $t$.  Conversely, negative pointwise information corresponds to
the event $s$ lowering the posterior probability $p(t|s)$ relative to the prior probability $p(t)$.
Hence when the event $t$ occurs we can say that the event $s$ was \emph{misinformative} about the
event $t$.\footnote{The term misinformation should absolutely not be taken to mean disinformation
  (i.e.\ does not mean intentionally misleading information).  Furthermore, note that while a source
  event $s$ may be deemed to be misinformative about a particular target event $t$, a source event
  $s$ is never misinformative about the target variable $T$ \emph{on average}.  This can be seen by
  noting that the pointwise mutual information averaged over all target realisations is non-negative
  \citep{fano1961}.  In other words, the information provided by $s$ is on average helpful for
  predicting $T$; however, in certain instances this, typically helpful information is misleading in
  the sense that it lowers $p(t|s)$ relative to $p(t)$.  Typically helpful information which
  subsequently turns out to be misleading is misinformation.}

\subsection{Pointwise Information Decomposition}
\label{sec:local_pid}

Now that we are familiar with pointwise nature of information, suppose that we have a discrete
realisation from the joint event space $\mc{T} \by \mc{S}_1 \by \mc{S}_2$ consisting of the target
event $t$ and predictor events $s_1$ and $s_2$.  The pointwise mutual information $i(s_1;t)$
quantifies the information provided individually by $s_1$ about $t$, while the pointwise mutual
information $i(s_2;t)$ quantifies the information provided individually by $s_2$ about $t$. The
pointwise mutual information $i(s_{1,2};t)$ quantifies the total information provided jointly by
$s_1$ and $s_2$ about $t$.  In correspondence with the (average) bivariate decomposition
\eq{decomp_bivar}, consider the pointwise bivariate decomposition, first suggested by
\citet{lizier2013},
\begin{alignat}{3}
  \label{eq:local_decomp_bivar}
  \isott &&= \ri &+ \uio + \uit + \ci,  \nonumber\\
  \isot  &&= \ri &+ \uio,               \nonumber\\
  \istt  &&= \ri &+ \uit.
\end{alignat}
Note that the lower case quantities denote the pointwise equivalent of the corresponding upper case
quantities in \eq{decomp_bivar}.  This decomposition could be considered for every discrete
realisation on the support of the joint distribution \mbox{$P(S_1,S_2,T)$}.  Hence, consider taking
the expectation of these pointwise atoms over all discrete realisations,
\begin{align}
  \label{eq:local_decomp_bivar_avg}
  \UIo &= \big\langle \uio \big\rangle,  &  \RI  &= \big\langle \ri  \big\rangle,  \nonumber\\
  \UIt &= \big\langle \uit \big\rangle,  &  \CI  &= \big\langle \ci  \big\rangle.
\end{align}
Since the expectation is a linear operation, this will recover the (average) bivariate decomposition
\eq{decomp_bivar}.  Equations \eq{local_decomp_bivar} for every discrete realisation, together with
\eq{decomp_bivar} and \eq{local_decomp_bivar_avg} form the bivariate pointwise information
decomposition.  Just as in \eq{decomp_bivar}, these equations are underdetermined requiring a
separate definition of either the pointwise unique, redundant or complementary information for
uniqueness.  (Defining an average atom is sufficient for a unique bivariate decomposition
\eq{decomp_bivar}, but still leaves the pointwise decomposition \eq{local_decomp_bivar} within each
realisation underdetermined).

\subsection{Pointwise Unique}
\label{sec:locunq}

Now consider applying this pointwise information decomposition to the probability distribution
\emph{Pointwise Unique} (\locunq) in \tb{locunq}.  In \locunq, observing $0$ in either of $S_1$ or
$S_2$ provides zero information about the target $T$, while complete information about the outcome
of $T$ is obtained by observing $1$ or a $2$ in either predictor.  The probability distribution is
structured such that in each of the four realisations, one predictor provides complete information
while the other predictor provides zero information---the two predictors never provide the same
information about the target which is justified by noting that one of the two predictors always
provides zero pointwise information.

Given that redundancy is supposed to capture the same information, it seems reasonable to assume
there must be zero pointwise redundant information for each realisation.  This assumption is made
without any measure of pointwise redundant information; however, no other possibility seems
justifiable.  This assertion is used to determine the pointwise redundant information terms in
\tb{locunq}.  Then using the pointwise information decomposition \eq{local_decomp_bivar}, we can
then evaluate the other pointwise atoms of information in \tb{locunq}.  Finally using
\eq{local_decomp_bivar_avg}, we get that there is zero (average) redundant information, and
$\nf{1}{2}$~bit of (average) unique information from each predictor.  From the pointwise
perspective, the only reasonable conclusion seems to be that the predictors in \locunq{} must
contain only unique information about the target.

However, in contrast to the above, $\imin$, $\ired$, $\uitilde$, and $\ivk$ all say that the
predictors in \locunq{} contain no unique information, rather only $\nf{1}{2}$ bit of redundant
information plus $\nf{1}{2}$ bit of complementary information.  This problem, which will be referred
to as the \emph{pointwise unique problem}, is a consequence of the fact that these measures all
satisfy Assumption~($*$) of \citet{bertschinger2014}, which (in effect) states that the unique and
redundant information should only depend on the marginal distributions \mbox{$P(S_1,T)$} and
\mbox{$P(S_2,T)$}.  In particular, any measure which satisfies Assumption~($*$) will yield zero
unique information when \mbox{$P(S_1,T)$} is isomorphic to \mbox{$P(S_2,T)$}, as is the case for
\locunq{}. (Here isomorphic should be taken to mean isomorphic probability spaces, e.g.\
\citep[p.~27]{gray1988} or \citep[p.4]{martin1984}.)  It arises because Assumption~($*$) (and indeed
the operational interpretation the led to its introduction) does not respect the pointwise nature of
information.  This operational view does not take into account the fact that individual events $s_1$
and $s_2$ may provide different information about the event $t$, even if the probability
distributions $P(S_1,T)$ and $P(S_2,T)$ are the same. Hence, we contend that for any measure to
capture the same information (not merely the same amount), it must respect the pointwise nature of
information.

\begin{table}[t]
 \centering %% \tablesize{} %% You can specify the fontsize here, e.g.
  \ifarXiv
  \vspace{5pt}
  \else
  \tablesize{\footnotesize}
  \fi
  \setlength{\tabcolsep}{5pt}
  \begin{tabular}{c || c c | c || c c | c || c c c c}
    \hline
    $p$ & $s_1$ & $s_2$ & $t$ & $\isot$ & $\istt$ & $\isott$ & $\uio$ & $\uit$ & $\ri$ & $\ci$ \\
    \hline\hline
    \nf{1}{4} & 0 & 1 & 1 & 0 & 1 & 1 & 0 & 1 & 0 & 0 \\
    \nf{1}{4} & 1 & 0 & 1 & 1 & 0 & 1 & 1 & 0 & 0 & 0 \\
    \nf{1}{4} & 0 & 2 & 2 & 0 & 1 & 1 & 0 & 1 & 0 & 0 \\
    \nf{1}{4} & 2 & 0 & 2 & 1 & 0 & 1 & 1 & 0 & 0 & 0 \\
    \hline\hline
    \multicolumn{4}{c||}{\scriptsize Expected values}
        & \nf{1}{2} & \nf{1}{2} & 1 & \nf{1}{2} & \nf{1}{2} & 0 & 0 \\
    \hline
  \end{tabular}
  \caption{\label{tb:locunq} Example \locunq{}. For each realisation, the pointwise mutual
    information proided by each individual and joint predictor events, about the target event has
    been evaluated.  Note that one predictor event always provides full information about the target
    while the other provides zero information.  Based on the this, it is \emph{assumed} that there
    must be zero redundant information.  The PPI atoms are then calculated via
    \eq{local_decomp_bivar}.}
 
\end{table}

% UI_min is trivial, UI minimises to fully redundant table, hence so does UIvk, UI_red is zero as a
% consequence of Theorem 22 ob Bertschinger 2014

\subsection{Pointwise Partial Information Decomposition}
\label{sec:ppid}

\newtheorem{paxiom}{PPID Axiom}

With the pointwise unique problem in mind, consider constructing an information decomposition with
the pointwise nature of information as an inherent property.  Let $\mb{a}_1,\ldots,\mb{a}_k$ be
potentially intersecting subsets of the predictor events $\mb{s}=\{s_1,\ldots,s_n\}$, called source
events.  Now consider rewriting the Williams and Beer axioms in terms of a measure of pointwise
redundant information $i_\cap$ with the aim of deriving a \emph{pointwise partial information
  decomposition} (PPID).
\begin{paxiom}[Symmetry]
  Pointwise redundant information is invariant under permutations $\sigma$ of source events,
  \begin{equation*}
    i_\cap\big(\mb{a}_1,\ldots,\mb{a}_k \ra t\big)
    = i_\cap\big(\sigma(\mb{a}_1),\dots,\sigma(\mb{a}_k) \ra T\big).
  \end{equation*}
\end{paxiom}
\begin{paxiom}[Monotonicity]
  \label{ax:p_mono}
  Pointwise redundant information decreases monotonically as more source events are included,
  \begin{equation*}
    i_\cap\big(\mb{a}_1,\ldots,\mb{a}_{k-1} \ra t\big)
    \leq i_\cap\big(\mb{a}_1,\ldots,\mb{a}_k \ra t\big)
  \end{equation*}
  with equality if $\mb{a}_k \supseteq \mb{a}_i$ for any
  $\mb{a}_i \in \{\mb{a}_1,\ldots,\mb{a}_{k-1} \}$.
\end{paxiom}
\begin{paxiom}[Self-redundancy]
  Pointwise redundant information for a single source event $\mb{a}_i$ equals the pointwise mutual
  information,
  \begin{equation*}
    i_\cap\big(\mb{a}_i \ra t\big) = i\big(\mb{a}_i \,; t\big).
  \end{equation*}
\end{paxiom}

It seems that the next step should be to define some measure of pointwise redundant information
which is compatible with these PPID axioms; however, there is a problem---the pointwise mutual
information is not non-negative.  While is not an issue for the examples like \locunq{}, where none
of the source events provide negative pointwise information, it is an issue in general (e.g.\ see
\imprdn{} in \secRef{imprdn}).  The problem is that set-theoretic intuition behind
Axiom~\ref{ax:p_mono} (monotonicity) makes little sense when considering signed measures like the
pointwise mutual information.
% Indeed, it is for this very reason that local positivity is considered a desirable property for
% PID.

Given the desire to address the pointwise unique problem, there is a need to overcome this issue.
\citet{ince2017} suggested that the set-theoretic intuition is only valid when all source events
provide either positive or negative pointwise information.  Ince contends that information and
misinformation are ``fundamentally different''~\citep[p.~11]{ince2017} and that the set-theoretic
intuition should be admitted in the difficult to interpret situations where both are present.  We
however, will take a different approach---one which aims to deal with these difficult to interpret
situations whilst preserving the set-theoretic intuition that redundancy corresponds to overlapping
information.

By way of a preview, we first consider precisely how an event $s_1$ provides information about an
event $t$ by the means of two distinct types of probability mass exclusion.  We show how considering
the process in this way naturally splits the pointwise mutual information into particular entropic
components, and how one can consider redundancy on each of these components separately.  Splitting
the signed pointwise mutual information into these unsigned entropic components circumvents the
above issue with Axiom~\ref{ax:p_mono} (monotonicity).  Crucially, however, by deriving these
entropic components from the probability mass exclusions, we retain the set-theoretic intuition of
redundancy---redundant information will correspond to overlapping probability mass exclusions in the
two-event partition~${\mc{T}^t=\{t,\ob{t}\}}$.

%%%%%%%%%%%%%%%%%%%%%%%%%%%%%%%%%%%%%%%%%%%%%%%%%%%%%%%%%%%%%%%%%%%%%%%%%%%%%%%%%%%%%%%%%%%%%%%%%%%%
%%%%%%%%%%%%%%%%%%%%%%%%%%%%%%%%%%%%%%%%%%%%%%%%%%%%%%%%%%%%%%%%%%%%%%%%%%%%%%%%%%%%%%%%%%%%%%%%%%%%

\section{Probability Mass Exclusions and the Directed Components\\of Pointwise Mutual Information}
\label{sec:same_info}

By definition, the pointwise information provided by $s$ about $t$ is associated with a change from
the prior $p(t)$ to the posterior $p(t|s)$.  As we explored from first principles in
\citet{finn17a}, this change is a consequence of the \emph{exclusion} of probability mass in the
target distribution $P(T)$ induced by the occurrence of the event $s$ and inferred via the joint
distribution $P(S,T)$.  To be specific, when the event $s$ occurs, one knows that the complementary
event \mbox{$\ob{s} = \{\mc{S} \sm s\}$} did not occur.  Hence one can \emph{exclude} the
probability mass in the joint distribution $P(S,T)$ associated with the complementary event, i.e.\
exclude $P(\ob{s},T)$, leaving just the probability mass $P(s,T)$ remaining.  The new target
distribution $P(T|s)$ is evaluated by normalising this remaining probability mass.  In
\citep{finn17a} we introduced \emph{probability mass diagrams} in order to visually explore the
exclusion process.  \fig{pmass_intro} provides an example of such a diagram.  Clearly, this process
is merely a description of the definition of conditional probability.  Nevertheless, we content that
by viewing the change from the prior to the posterior in this way---by focusing explicitly on the
exclusions rather than the resultant conditional probability---the vague intuition that redundancy
corresponds to overlapping information becomes more apparent. This point will elaborated upon in
\secRef{op_interp_red}. However, in order to do so, we need to first discuss the two distinct types
of probability mass exclusion (which we do in \secRef{two_types_excl}) and then relate these to
information-theoretic quantities (which we do in \secRef{spec_amb}).

\begin{figure}[t]
  \centering \begin{tikzpicture}

  \def\originy{0} 
  \def\height{3}
  \def\width{1.5}
  \def\overdrawn{0.2}

  % -----------
  
  \def\originx{0}
  
  \draw[black, line width=1pt] (\originx,\originy)
    -- (\originx,\originy+\height);
  \draw[black, line width=1pt] (\originx-\overdrawn,\originy+\height)
    -- (\originx+\width+\overdrawn,\originy+\height);
  \draw[black, line width=1pt] (\originx+\width,\originy+\height)
    -- (\originx+\width,\originy);
  \draw[black, line width=1pt] (\originx+\width+\overdrawn,\originy)
    -- (\originx-\overdrawn,\originy);

  \draw[black, line width=0.75pt] (\originx,\originy+7/8*\height)
    -- (\originx+\width+\overdrawn,\originy+7/8*\height);
  \draw[black, line width=1pt] (\originx-\overdrawn,\originy+\height/2)
    -- (\originx+\width+\overdrawn,\originy+\height/2);
  \draw[black, line width=0.75pt] (\originx,\originy+1/4*\height)
    -- (\originx+\width+\overdrawn,\originy+1/4*\height);

  \node at (\originx,3/4*\height)[anchor=east] {$t^1$};
  \node at (\originx,1/4*\height)[anchor=east] {$t^2$};  
  \node at (\originx+\width,15/16*\height)[anchor=west] {$s^1$};
  \node at (\originx+\width,11/16*\height)[anchor=west] {$s^2$};
  \node at (\originx+\width,3/8*\height)[anchor=west] {$s^1$};
  \node at (\originx+\width,1/8*\height)[anchor=west] {$s^3$};
  \node at (\originx+\width/2,\originy+\height)[anchor=south] {$P(S,T)$};
  \node at (\originx+\width/2,\originy+15/16*\height)[anchor=center] {$\nf{1}{8}$};
  \node at (\originx+\width/2,\originy+11/16*\height)[anchor=center] {$\nf{3}{8}$};
  \node at (\originx+\width/2,\originy+3/8*\height)[anchor=center] {$\nf{1}{4}$};
  \node at (\originx+\width/2,\originy+1/8*\height)[anchor=center] {$\nf{1}{4}$};

  % -----------
  
  \def\originxtwo{3}
  
  \draw[black, line width=1pt] (\originxtwo,\originy)
    -- (\originxtwo,\originy+\height);
  \draw[black, line width=1pt] (\originxtwo-\overdrawn,\originy+\height)
    -- (\originxtwo+\width+\overdrawn,\originy+\height);
  \draw[black, line width=1pt] (\originxtwo+\width,\originy+\height)
    -- (\originxtwo+\width,\originy);
  \draw[black, line width=1pt] (\originxtwo+\width+\overdrawn,\originy)
    -- (\originxtwo-\overdrawn,\originy);

  \draw[black, line width=0.75pt] (\originxtwo,\originy+7/8*\height)
    -- (\originxtwo+\width+\overdrawn,\originy+7/8*\height);
  \draw[black, line width=1pt] (\originxtwo-\overdrawn,\originy+\height/2)
    -- (\originxtwo+\width+\overdrawn,\originy+\height/2);
  \draw[black, line width=0.75pt] (\originxtwo,\originy+1/4*\height)
    -- (\originxtwo+\width+\overdrawn,\originy+1/4*\height);

  \draw[pattern=north east lines, pattern color=red] (\originxtwo,\originy+\height/2) rectangle
    (\originxtwo+\width,\originy+7/8*\height);
  \draw[pattern=vertical lines, pattern color=red] (\originxtwo,\originy) rectangle
    (\originxtwo+\width,\originy+\height/4);

  \node at (\originxtwo,3/4*\height)[anchor=east] {$t^1$};
  \node at (\originxtwo,1/4*\height)[anchor=east] {$\ob{t^1}$};  
  \node at (\originxtwo+\width,15/16*\height)[anchor=west] {$s^1$};
  \node at (\originxtwo+\width,11/16*\height)[anchor=west] {$\ob{s^1}$};
  \node at (\originxtwo+\width,3/8*\height)[anchor=west] {$s^1$};
  \node at (\originxtwo+\width,1/8*\height)[anchor=west] {$\ob{s^1}$};
  \node at (\originxtwo+\width/2,\originy+\height)[anchor=south] {$P(s^1,T)$};
  \node at (\originxtwo+\width/2,\originy+15/16*\height)[anchor=center] {$\nf{1}{8}$};
  \node at (\originxtwo+\width/2,\originy+3/8*\height)[anchor=center] {$\nf{1}{4}$};

  % -----------

  \def\originxthree{6.5}
  \def\arrowgap{0.2}

  % ---

  \draw[->, black, line width=0.75pt] (\originxtwo+\width+\overdrawn+\arrowgap,
      \originy+\height)
    -- (\originxthree-\overdrawn-\arrowgap,\originy+\height);
  
  \draw[->, black, line width=0.75pt] (\originxtwo+\width+\overdrawn+\arrowgap,
      \originy+7/8*\height)
    -- (\originxthree-\overdrawn-\arrowgap/1.41,\originy+2/3*\height+\arrowgap/1.41);

  \draw[->, black, line width=0.75pt] (\originxtwo+\width+\overdrawn+\arrowgap,
      \originy+\height/2)
    -- (\originxthree-\overdrawn-\arrowgap/1.41,\originy+2/3*\height-\arrowgap/1.41);

  \draw[->, black, line width=0.75pt] (\originxtwo+\width+\overdrawn+\arrowgap,
      \originy+1/4*\height)
    -- (\originxthree-\overdrawn-\arrowgap/1.41,\originy+\arrowgap/1.41);

  \draw[->, black, line width=0.75pt] (\originxtwo+\width+\overdrawn+\arrowgap,
      \originy)
    -- (\originxthree-\overdrawn-\arrowgap,\originy); 

  % ---
    
  \draw[black, line width=1pt] (\originxthree,\originy)
    -- (\originxthree,\originy+\height);
  \draw[black, line width=1pt] (\originxthree-\overdrawn,\originy+\height)
    -- (\originxthree+\width+\overdrawn,\originy+\height);
  \draw[black, line width=1pt] (\originxthree+\width,\originy+\height)
    -- (\originxthree+\width,\originy);
  \draw[black, line width=1pt] (\originxthree+\width+\overdrawn,\originy)
    -- (\originxthree-\overdrawn,\originy);

  \draw[black, line width=1pt] (\originxthree-\overdrawn,\originy+2/3*\height)
    -- (\originxthree+\width+\overdrawn,\originy+2/3*\height);

  \node at (\originxthree,5/6*\height)[anchor=east] {$t^1$};
  \node at (\originxthree,2/6*\height)[anchor=east] {$\ob{t^1}$};
  \node at (\originxthree+\width,5/6*\height)[anchor=west] {$s^1$}; 
  \node at (\originxthree+\width,2/6*\height)[anchor=west] {$s^1$};
  \node at (\originxthree+\width/2,\originy+\height)[anchor=south] {$P(T|s^1)$};
  \node at (\originxthree+\width/2,\originy+5/6*\height)[anchor=center] {$\nf{1}{3}$};
  \node at (\originxthree+\width/2,\originy+2/6*\height)[anchor=center] {$\nf{2}{3}$};

\end{tikzpicture}
  \caption[]{Sample probability mass diagrams, which use length to represent the probability mass of
    each joint event from $\mc{T} \by \mc{S}$. \emph{Left}:\ the joint distribution $P(S,T)$;
    \emph{Middle}: The occurrence of the event $s^1$ leads to exclusions of the complementary event
    $\ob{s^1}$ which consists of two elementary event, i.e.  $\ob{s^1} = \{ s^2, s^3\}$.  This
    leaves the probability mass $P(s^1,T)$ remaining.  The exclusion of the probability mass
    $p(\ob{s^1},t^1)$ was misinformative since the event $t^1$ did occur. By convention,
    misinformative exclusions will be indicated with diagonal hatching.  On the other hand, the
    exclusion of the probability mass $p(\ob{t^1},\ob{s^1})$ was informative since the complementary
    event $\ob{t^1}$ did not occur. By convention, informative exclusions will be indicated with
    horizontal or vertical hatching. \emph{Right}:\ this remaining probability mass can be
    normalised yielding the conditional distribution $P(T|s^1)$.}
  \label{fig:pmass_intro}
\end{figure}
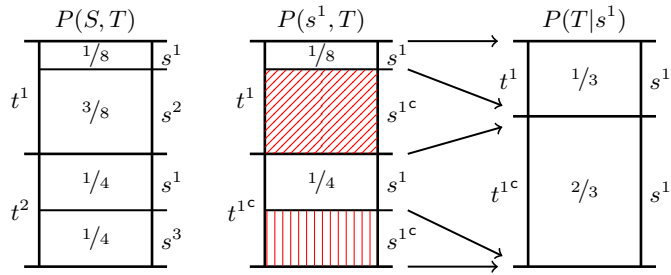

\subsection{Two Distinct Types of Probability Mass Exclusions}
\label{sec:two_types_excl}

In \citep{finn17a} we examined the two distinct types of probability mass exclusions.  The
difference between the two depends on where the exclusion occurs in the target distribution $P(T)$
and the particular target event $t$ which occurred.  \emph{Informative exclusions} are those which
are confined to the probability mass associated with the set of elementary events in the target
distribution which \emph{did not occur}, i.e.\ exclusions confined to the probability mass of the
complementary event $p(\ob{t})$.  They are called such because the pointwise mutual information
$i(s;t)$ is a monotonically increasing function of the total size of these exclusions $p(\ob{t})$.
By convention, informative exclusions are represented on the probability mass diagrams by horizontal
or vertical lines.  On the other hand, the \emph{misinformative exclusion} is confined to the
probability mass associated with the elementary event in the target distribution which \emph{did
  occur}, i.e.\ an exclusion confined to $p(t)$. It is referred to as such because the pointwise
mutual information $i(s;t)$ is a monotonically decreasing function of the size of this type of
exclusion $p(t)$.  By convention, misinformative exclusions are represented on the probability mass
diagrams by diagonal lines.

Although an event $s$ may exclusively induce either type of exclusion, in general both types of
exclusion are present simultaneously.  The distinction between the two types of exclusions leads
naturally to the following question---can one decompose the pointwise mutual information $i(s;t)$
into a positive informational component associated with the informative exclusions, and a negative
informational component associated with the misinformative exclusions?  This question is considered
in detail in \secRef{spec_amb}.  However, before moving on, there is a crucial observation to be
made about the pointwise mutual information which will have important implications for the measure
of redundant information to be introduced later. 

\begin{Remark}
  \label{re:crucial}
  The pointwise mutual information $i(s;t)$ depends only on the size of informative and
  misinformative exclusions.  In particular, it does not depend on the apportionment of the
  informative exclusions across the set of elementary events contained in the complementary event
  $\ob{t}$.
\end{Remark}

In other words, whether the event $s$ turns out to be net informative or misinformative about the
event $t$---whether $i(s;t)$ is positive or negative---depends on the size of the two types of
exclusions; but, to be explicit, does \emph{not} depend on the distribution of the informative
exclusion across the set of target events which did not occur.  This remark will be crucially
important when it comes to providing the operational interpretation of redundant information in
\secRef{op_interp_red}. (It is also further discussed in terms of Kelly gambling~\citep{kelly1956}
in \appRef{operational}).

\subsection{The Directed Components of Pointwise Information:\ Specificity and Ambiguity}
\label{sec:spec_amb}

We return now to the idea that one might be able to decompose the pointwise mutual information into
a positive and negative component associated with the informative amd misinformative exclusions
respectively.  In \citep{finn17a} we proposed four postulates for such a decomposition.  Before
stating the postulates, it is important to note that although there is a ``surprising symmetry''
\citep[p.\ 23]{ash1965} between the information provided by $s$ about $t$ and the information
provided by $t$ about $s$, there is nothing to suggest that the components of the decomposition
should be symmetric---indeed the intuition behind the decomposition only makes sense when
considering the information is considered in a directed sense. As such, directed notation will be
used to explicitly denote the information provided by $s$ about $t$.

\newtheorem{postulate}{Postulate}

\begin{postulate}[Decomposition]
  \label{post:decomp}
  The pointwise information provided by $s$ about $t$ can be decomposed into two non-negative
  components, such that
  \begin{equation*}
    i(s;t) = i_+(s \ra t) - i_-(s \ra t).
    \end{equation*}
\end{postulate}

\begin{postulate}[Monotonicity]
  \label{post:monotonicity}
  For all fixed $p(s,t)$ and $p(\ob{s},t)$, the function $i_+(s \ra t)$ is a monotonically
  increasing, continuous function of $p(\ob{t},\ob{s})$.  For all fixed $p(\ob{t},s)$ and
  $p(\ob{t},\ob{s})$, the function $i_-(s \ra t)$ is a monotonically increasing continuous function
  of $p(\ob{s},t)$.  For all fixed $p(s,t)$ and $p(\ob{t},s)$, the functions $i_+(s \ra t)$ and
  $i_-(s \ra t)$ are monotonically increasing and decreasing functions of $p(\ob{t},\ob{s})$,
  respectively.
\end{postulate}

\begin{postulate}[Self-Information]
  \label{post:self_info}
  An event cannot misinform about itself, $i_+(s \ra s) = i(s;s) = - \log p(s)$.
\end{postulate}

\begin{postulate}[Chain Rule]
  \label{post:chain_rule}
  The functions $i_+(s_{1,2} \ra t)$ and $i_-(s_{1,2} \ra t)$ satisfy a chain rule, i.e.
  \begin{align*}
    i_+(s_{1,2} \ra t) &= i_+(s_1 \ra t) + i_+(s_2 \ra t | s_1) \\
                       &= i_+(s_2 \ra t) + i_+(s_1 \ra t | s_2), \\
    i_-(s_{1,2} \ra t) &= i_-(s_1 \ra t) + i_-(s_2 \ra t | s_1) \\
                       &= i_-(s_2 \ra t) + i_-(s_1 \ra t | s_2)
  \end{align*}
\end{postulate}

% \begin{postulate}[Decomposition]
%   \label{post:decomp}
%   Given an event space $\mc{T} \by \mc{S}_1$, the pointwise mutual information provided by the
%   event $s_1$ about the event $t$ can be decomposed into two non-negative components $\ipsot$ and
%   $\insot$ such that
%   \begin{equation*}
%     i(s_1; t) = \ipsot{} - \insot{}.
%   \end{equation*}
% \end{postulate}
% 
% \begin{postulate}[Monotonicity]
%   \label{post:monotonicity}
%   The function \mbox{$\ipsot$} is a continuous, monotonically increasing function of the size of
%   the informative exclusion $p(\ob{s_1},\ob{t})$ for a fixed sized misinformative exclusion
%   $p(\ob{s_1},t)$.  On the other hand, the function \mbox{$\insot$} is a continuous, monotonically
%   increasing function of the size of the misinformative exclusion $p(\ob{s_1},t)$ for a fixed
%   sized informative exclusion $p(\ob{s_1},\ob{t})$.
% \end{postulate}
% 
% \begin{postulate}[Chain Rule]
%   \label{post:chain_rule}
%   The functions \mbox{$i^+$} and \mbox{$i^-$} satisfy the following chain rule,
%   \begin{align*}
%     i^+(s_{1,2} \ra t) &= i^+(s_1 \ra t) + i^+(s_2 \ra t | s_1), \\
%     i^-(s_{1,2} \ra t) &= i^-(s_1 \ra t) + i^-(s_2 \ra t | s_1).
%   \end{align*}
% \end{postulate}

In \citet{finn17a}, we proved that these postulates lead to the following forms which are unique up
to the choice of the base of the logarithm in the mutual information in Postulates~\ref{post:decomp}
and \ref{post:self_info},
\begin{alignat}{3}
  &i^+(s_1 \ra t)          &&= h(s_1)             &&= -\log p(s_1),   \label{eq:specificity} \\
  &i^+(s_1 \ra t | s_2)    &&= h(s_1 | s_2)       &&= -\log p(s_1 | s_2),                    \\
  &i^+(s_{1,2} \ra t)      &&= h(s_{1,2})         &&= -\log p(s_{1,2}),                      \\
  &i^-(s_1 \ra t)          &&= h(s_1 | t)         &&= -\log p(s_1 | t), \label{eq:ambiguity} \\
  &i^-(s_1 \ra t | s_2) \; &&= h(s_1 | t, s_2) \; &&= -\log p(s_1 | t, s_2),                 \\
  &i^-(s_{1,2} \ra t)      &&= h(s_{1,2} | t)     &&= -\log p(s_{1,2} | t). 
\end{alignat}
That is, the Postulates~\ref{post:decomp}--\ref{post:chain_rule} uniquely decompose the pointwise information provided by $s$ about $t$ into
the following entropic components,
\begin{align}
  \label{eq:decomp}
  i(s; t) &= i^+(s \ra t) - i^-(s \ra t)     \nn\\
          &= h(s) - h(s | t).  
\end{align}
Although the decomposition of mutual information into entropic components is well-known, it is
non-trivial that Postulates~\ref{post:decomp} and \ref{post:self_info}, based on the size of the two
distinct types of probability mass exclusions, lead to this particular form, but not
$i(s;t) = h(t)-h(t|s)$ or $ i(s;t) =h(s) + h(t) - h(s,t)$.

It is important to note that although the original motivation was to decompose the pointwise mutual
information into separate components associated with informative and misinformative exclusion, the
decomposition \eq{decomp} does not quite possess this direct correspondence:
\begin{itemize}

\item The positive informational component ${i^+(s \ra t)}$ does not depend on $t$ but rather only
  on $s$.  This can be interpreted as follows:\ the less likely $s$ is to occur, the more specific
  it is when it does occur, the greater the total amount of probability mass excluded $p(\ob{s})$,
  and the greater the potential for $s$ to inform about $t$ (or indeed any other target
  realisation).

\item The negative informational component ${i^-(s \ra t)}$ depends on both $s$ and $t$, and can
  be interpreted as follows:\ the less likely $s$ is to coincide with the event $t$, the more
  uncertainty in $s$ given $t$, the greater size of the misinformative probability mass exclusion
  $p(\ob{s},t)$, and therefore the greater the potential for $s$ to misinform about $t$.

  % the more ambiguous $t$ is about $s$,
  
\end{itemize}
In other words, although the negative informational component ${i^-(s \ra t)}$ does correspond
directly to the size of the misinformative exclusion $p(\ob{s},t)$, the positive informational
component ${i^+(s \ra t)}$ does not correspond directly to the size of the informative exclusion
$p(\ob{t},\ob{s})$.  Rather, the positive informational component ${i^+(s \ra t)}$ corresponds to
the \emph{total} size of the probability mass exclusions $p(\ob{s})$, which is the sum of the sum of
the informative and misinformative exclusions.  For the sake of brevity, the positive informational
component ${i^+(s \ra t)}$ will be referred to as the \emph{specificity}, while the negative
informational component ${i^-(s \ra t)}$ will be referred to as the \emph{ambiguity}.\footnote{The
  usage of the term ambiguity in this context is due to Shannon:\ ``equivocation measures the
  average ambiguity of the received signal''.  Specificity is an antonym of ambiguity and the usage
  here is inline with the definition since the more specific an event $s$, the more information it
  could provide about $t$ after the ambiguity is taken into account.}

\subsection{Operational Interpretation of Redundant Information}
\label{sec:op_interp_red}

Arguing about whether one piece of information differs from another piece of information is
nonsensical without some kind of unambiguous definition of what it means for two pieces of
information to be the same.  As such, \citet{bertschinger2014} advocate the need to provide an
operational interpretation of what it means for information to be unique or redundant.  This section
provides our operational definition of what it means for information to be the same. This definition
provides a concrete interpretation of what it means for information to be redundant in terms of
overlapping probability mass exclusions.

The operational interpretation of redundancy adopted here is based upon the following idea:\ since
the pointwise information is ultimately derived from probability mass exclusions, the \emph{same
  information} must induce the \emph{same exclusions}.  More formally, the information provided by a
set of predictor events $s_1,\ldots,s_k$ about a target event~$t$ must be the same information if
each source event induces the same exclusions with respect to the two-event partition
${\mc{T}^t=\{t,\ob{t}\}}$.  While this statement makes the motivational intuition clear, it is not
yet sufficient to serve as an operational interpretation of redundancy:\ there is no reference to
the two distinct types of probability mass exclusions, the specific reference to the pointwise event
space $\mc{T}^t$ has not been explained, and there is no reference to the fact the exclusions from
each source may differ in size.

Informative exclusions are fundamentally different from misinformative exclusions and hence each
type of exclusion should be compared separately:\ informative exclusions can overlap with
informative exclusions, and misinformative exclusions can overlap with misinformative exclusions.
In information-theoretic terms, this means comparing the specificity and the ambiguity of the
sources separately---i.e.\ considering a measure of redundant specificity and a separate measure of
redundant ambiguity.  Crucially, these quantities (being pointwise entropies) are unsigned meaning
that the difficulties associated with Axiom~\ref{ax:p_mono} (Monotonicity) and signed pointwise
mutual information in \secRef{ppid} will not be an issue here.

The specific reference to the two-event partition $\mc{T}^t$ in the above statement is based upon
Remark~\ref{re:crucial} and is crucially important.  The pointwise mutual information does not
depend on the apportionment of the informative exclusions across the set of events which did not
occur, hence the pointwise redundant information should not depend on this apportionment either.  In
other words, it is immaterial if two predictor events $s_1$ and $s_2$ exclude different elementary
events within the target complementary event $\ob{t}$ (assuming the probability mass excluded is
equal) since with respect to the realised target event $t$ the difference between the exclusions is
only semantic.  This has important implications for the comparison of exclusions from different
predictor events.  As the pointwise mutual information depends on, and only depends on, the size of
the exclusions, then the only sensible comparison is a comparison of size.  Hence, the common or
overlapping exclusion must be the smallest exclusion. Thus, consider the following operational
interpretation of redundancy:

\newtheorem*{op_interp}{Operational Interpretation}

\begin{op_interp}[Redundant Specificity]
  The redundant specificity between a set of predictor events $s_1,\ldots,s_n$ is the specificity
  associated with the source event which induces the smallest total exclusions.
\end{op_interp}

\begin{op_interp}[Redundant Ambiguity]
  The redundant ambiguity between a set of predictor events $s_1,\ldots,s_n$ is the ambiguity
  associated with the source event which induces the smallest misinformative exclusion.
\end{op_interp}

\subsection{Motivational Example}
\label{sec:motiv_example}

To motivate the above operational interpretation, and in particular the need to treat the
specificity separately to the ambiguity, consider \fig{pmass_spec_amb}. In this pointwise example,
two different predictor events provide the same amount of pointwise information since
\mbox{$P(T|s_1^1) = P(T|s_2^1)$}, and yet the information provided by each event is in some way
different since each excludes different sections of the target distribution $P(T)$.  In particular,
$s_1^1$ and $s_2^1$ both preclude the target event $t^2$, while $s_2^1$ additionally excludes
probability mass associated with target events $t^1$ and $t^3$.  From the perspective of the
pointwise mutual information the events $s_1^1$ and $s_2^1$ seem to be providing the same
information as
\begin{equation}
  \label{eq:local_inf}
  i(s_1^1 \ra t^1) = i(s_2^1 \ra t^1) = \log \nf{4}{3} \text{ bit}.
\end{equation}
However, from the perspective of the specificity and the ambiguity it can be seen that information
is being provided in different ways since
\begin{align}
  \label{eq:local_inf_decomp}
  i^+(s_1^1 \ra t^1) &= \log \nf{4}{3} \text{ bit}, & i^-(s_1^1 \ra t^1) &= 0 \text{ bit}, \nn\\
  i^+(s_2^1 \ra t^1) &= \log \nf{8}{3} \text{ bit}, & i^-(s_2^1 \ra t^1) &= 1 \text{ bit}.
\end{align}

Now consider the problem of decomposing information into its unique, redundant and complementary
components.  \fig{pmass_spec_amb} shows where exclusions induced by $s_1^1$ and $s_2^1$ overlap
where they both exclude the target event $t^2$ which is an informative exclusion.  This is the only
exclusion induced by $s_1^1$ and hence all of the information associated with this exclusion must be
redundantly provided by the event $s_2^1$.  Without any formal framework, consider taking the
redundant specificity and redundant ambiguity,
\begin{alignat}{4}
  &r^+(s_1^1, s_2^1 \ra t^1) &&= i^+(s_1^1 \ra t^1) =& \; \log \nf{4}{3} & \text{ bit,} \\
  &r^-(s_1^1, s_2^1 \ra t_1) &&= i^-(s_1^1 \ra t^1) =& 0 &\text{ bit.}
\end{alignat}
This would mean that the event $s_2^1$ provides the following unique specificity and unique
ambiguity,
\begin{alignat}{3}
  &u^+(s_1^1 \sm s_2^1 \ra t^1) &&= i^+(s_1^1 \ra t^1)-r^+(s_1^1,s_2^1 \ra t^1)
    &&= 1 \text{ bit,} \\
  &u^-(s_1^1 \sm s_2^1 \ra t^1) &&= i^-(s_1^1 \ra t^1)-r^-(s_1^1,s_2^1 \ra t^1)
    &&= 1 \text{ bit.}
\end{alignat}
The redundant specificity $\log \nf{4}{3}$~bit accounts for the overlapping informative exclusion of
the event~$t^2$. The unique specificity and unique ambiguity from $s_2^1$ are associated with its
non-overlapping informative and misinformative exclusions; however, both of these $1$~bit and hence,
on net, $s_2^1$ is no more informative than $s_1^1$.  Although obtained without a formal framework,
this example highlights a need to consider the specificity and ambiguity rather than merely the
pointwise mutual information.

\begin{figure}[t]
  \centering \input{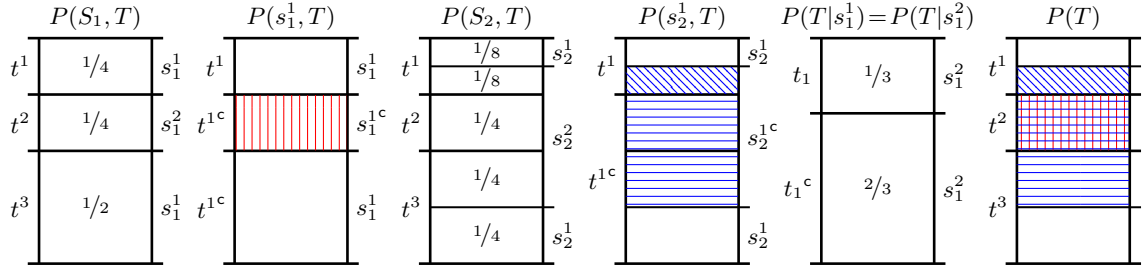}
  \caption{Sample probability mass diagrams for two predictors $S_1$ and $S_2$ to a given target
    $T$.  Here events in the two different predictor spaces provide the same amount of pointwise
    information about the target event, $\log_2 \nf{4}{3}$ bits, since
    \mbox{$P(T|s_1^1) = P(T|s_2^1)$}, although each excludes different sections of the target
    distribution $P(T)$. Since they both provide the same amount of information, is there a way to
    characterise what information the additional unique exclusions from the event $s_2^1$ are
    providing?}
  \label{fig:pmass_spec_amb}
\end{figure}

%%%%%%%%%%%%%%%%%%%%%%%%%%%%%%%%%%%%%%%%%%%%%%%%%%%%%%%%%%%%%%%%%%%%%%%%%%%%%%%%%%%%%%%%%%%%%%%%%%%%
%%%%%%%%%%%%%%%%%%%%%%%%%%%%%%%%%%%%%%%%%%%%%%%%%%%%%%%%%%%%%%%%%%%%%%%%%%%%%%%%%%%%%%%%%%%%%%%%%%%%

\section{Pointwise Partial Information Decomposition Using Specificity and Ambiguity}
\label{sec:local_framework}

\newtheorem{axiom}{Axiom}

\newcommand{\icp}{i_\cap^+}
\newcommand{\icn}{i_\cap^-}

Based upon the argumentation of \secRef{same_info}, consider the following axioms:
\begin{axiom}[Symmetry]
  \label{ax:symmetry}
  Pointwise redundant specificity $\icp$ and pointwise redundant ambiguity $\icn$ are invariant
  under any permutation $\sigma$ of source events,
  \begin{align*}
    \icp\big( \mb{a}_1, \dots, \mb{a}_k \ra t \big)
    &= \icp\big( \sigma(\mb{a}_1), \dots, \sigma(\mb{a}_k) \ra t \big), \\
    \icn\big( \mb{a}_1, \dots, \mb{a}_k \ra t \big)
    &= \icn\big( \sigma(\mb{a}_1), \dots, \sigma(\mb{a}_k) \ra t \big).
  \end{align*}
\end{axiom}

\begin{axiom}[Monotonicity]
  \label{ax:mono}
  Pointwise redundant specificity $\icp$ and pointwise redundant ambiguity $\icn$ decreases
  monotonically as more source events are included,
  \begin{align*} 
    \icp\big( \mb{a}_1, \dots, \mb{a}_{k-1}, \mb{a}_k \ra t \big)
      &\leq \icp\big(\mb{a}_1, \dots, \mb{a}_{k-1} \ra t \big),    \\
    \icn\big( \mb{a}_1, \dots, \mb{a}_{k-1}, \mb{a}_k \ra t \big)
      &\leq \icn\big(\mb{a}_1,\dots, \mb{a}_{k-1} \ra t \big).
  \end{align*}
  with equality if $\mb{a}_k \supseteq \mb{a}_i$ for any
  $\mb{a}_i \in \{\mb{a}_1, \ldots, \mb{a}_{k-1} \}$.
\end{axiom}

\begin{axiom}[Self-redundancy]
  \label{ax:sr}
  Pointwise redundant specificity $\icp$ and pointwise redundant ambiguity $\icn$ for a single
  source event $\mb{a}_i$ equals the specificity and ambiguity respectively,
  \begin{alignat*}{3}
    &\icp(\mb{a}_i \ra t) &&= i^+(\mb{a}_i \ra t) &&= h(\mb{a}_i),      \\
    &\icn(\mb{a}_i \ra t) &&= i^-(\mb{a}_i \ra t) &&= h(\mb{a}_i | t).
  \end{alignat*}
\end{axiom}

\begin{figure}[t]
 \centering %% \tablesize{} %% You can specify the fontsize here, e.g.
  \begin{tikzpicture}

  \def\ox{0}
  \def\oy{0}
  \def\gx{1.8}
  \def\gy{1.5}

  \tikzstyle{block} = [rectangle, draw=none, 
    text width=5em, text centered, minimum height=1em]

    \small
    
% -----------------------------------------------------------------
    
  \node [block] (abc)      at (\ox,       \oy+6*\gy) {$\{123\}$};  

  \node [block] (bc)       at (\ox+\gx,   \oy+5*\gy) {$\{23\}$};
  \node [block] (ac)       at (\ox,       \oy+5*\gy) {$\{13\}$};
  \node [block] (ab)       at (\ox-\gx,   \oy+5*\gy) {$\{12\}$};
  
  \node [block] (acIbc)    at (\ox+\gx,   \oy+4*\gy) {$\{13\}\{23\}$};
  \node [block] (abIbc)    at (\ox,       \oy+4*\gy) {$\{12\}\{23\}$};
  \node [block] (abIac)    at (\ox-\gx,   \oy+4*\gy) {$\{12\}\{13\}$};

  \node [block] (abIacIbc) at (\ox+2*\gx, \oy+3*\gy) {$\{12\}\{13\}\{23\}$};
  
  \node [block] (c)        at (\ox+\gx,   \oy+3*\gy) {$\{3\}$};
  \node [block] (b)        at (\ox,       \oy+3*\gy) {$\{2\}$};
  \node [block] (a)        at (\ox-\gx,   \oy+3*\gy) {$\{1\}$};   

  \node [block] (cIab)     at (\ox+\gx,   \oy+2*\gy) {$\{3\}\{12\}$};
  \node [block] (bIac)     at (\ox,       \oy+2*\gy) {$\{2\}\{13\}$};
  \node [block] (aIbc)     at (\ox-\gx,   \oy+2*\gy) {$\{1\}\{23\}$}; 

  \node [block] (bIc)      at (\ox+\gx,   \oy+1*\gy) {$\{2\}\{3\}$};
  \node [block] (aIc)      at (\ox,       \oy+1*\gy) {$\{1\}\{3\}$};
  \node [block] (aIb)      at (\ox-\gx,   \oy+1*\gy) {$\{1\}\{2\}$};

  \node [block] (aIbIc)    at (\ox,       \oy      ) {$\{1\}\{2\}\{3\}$};

  \draw[->, black, line width=0.75pt] (abc) -- (bc);
  \draw[->, black, line width=0.75pt] (abc) -- (ac);
  \draw[->, black, line width=0.75pt] (abc) -- (ab); 
  
  \draw[->, black, line width=0.75pt] (bc) -- (abIbc);
  \draw[->, black, line width=0.75pt] (ac) -- (acIbc);
  \draw[->, black, line width=0.75pt] (ab) -- (abIbc); 
  \draw[->, black, line width=0.75pt] (bc) -- (acIbc);
  \draw[->, black, line width=0.75pt] (ac) -- (abIac);
  \draw[->, black, line width=0.75pt] (ab) -- (abIac);

  \draw[->, black, line width=0.75pt] (acIbc) -- (abIacIbc);
  \draw[->, black, line width=0.75pt] (abIbc) -- (abIacIbc);
  \draw[->, black, line width=0.75pt] (abIac) -- (abIacIbc); 

  \draw[->, black, line width=0.75pt] (acIbc) -- (c);
  \draw[->, black, line width=0.75pt] (abIbc) -- (b);
  \draw[->, black, line width=0.75pt] (abIac) -- (a);

  \draw[->, black, line width=0.75pt] (abIacIbc) -- (cIab);
  \draw[->, black, line width=0.75pt] (abIacIbc) -- (bIac);
  \draw[->, black, line width=0.75pt] (abIacIbc) -- (aIbc);
  
  \draw[->, black, line width=0.75pt] (c) -- (cIab);
  \draw[->, black, line width=0.75pt] (b) -- (bIac);
  \draw[->, black, line width=0.75pt] (a) -- (aIbc);

  \draw[->, black, line width=0.75pt] (cIab) -- (bIc);
  \draw[->, black, line width=0.75pt] (bIac) -- (bIc);
  \draw[->, black, line width=0.75pt] (aIbc) -- (aIc); 
  \draw[->, black, line width=0.75pt] (cIab) -- (aIc);
  \draw[->, black, line width=0.75pt] (bIac) -- (aIb);
  \draw[->, black, line width=0.75pt] (aIbc) -- (aIb);
  
  \draw[->, black, line width=0.75pt] (bIc) -- (aIbIc);
  \draw[->, black, line width=0.75pt] (aIc) -- (aIbIc);
  \draw[->, black, line width=0.75pt] (aIb) -- (aIbIc);

  \def\obx{\ox-6}
  \def\oby{\oy+3*\gy}
  \def\gbx{1.7}
  \def\gby{1.7}
  
  \node [block] (xy)  at (\obx,     \oby+\gby) {$\{12\}$};
  \node [block] (y)   at (\obx+\gbx, \oby     ) {$\{2\}$};
  \node [block] (x)   at (\obx-\gbx, \oby     ) {$\{1\}$};
  \node [block] (xIy) at (\obx,    \oby-\gby) {$\{1\}\{2\}$};

  \draw[->, black, line width=0.75pt] (xy) -- (y);
  \draw[->, black, line width=0.75pt] (xy) -- (x);
  \draw[->, black, line width=0.75pt] (y) -- (xIy);
  \draw[->, black, line width=0.75pt] (x) -- (xIy);
  
\end{tikzpicture}
  \caption{\label{fig:lattices} The lattice induced by the partial order $\preceq$
    \eq{partial_order} over the sources $\msc{A}(\mb{s})$ \eq{nodes}. \emph{Left}:\ the lattice for
    ${\mb{s}=\{s_1,s_2\}}$. \emph{Right}:~the lattice for ${\mb{s}=\{s_1,s_2,s_3\}}$.  See
    \appRef{tech} for further details.  Each node corresponds to the self-redundancy
    (Axiom~\ref{ax:sr}) of a source event, e.g.\ $\{1\}$ corresponds to the source event
    $\big\{\{s_1\}\big\}$, while $\{12,13\}$ corresponds to the source event
    $\big\{\{s_{1,2}\},\{s_{1,3}\}\big\}$.  Note that the specificity and ambiguity lattices share
    the same structure as the redundancy lattice of PID (cf. FIG.\ 2 in \citep{williams2010}).} 
\end{figure}
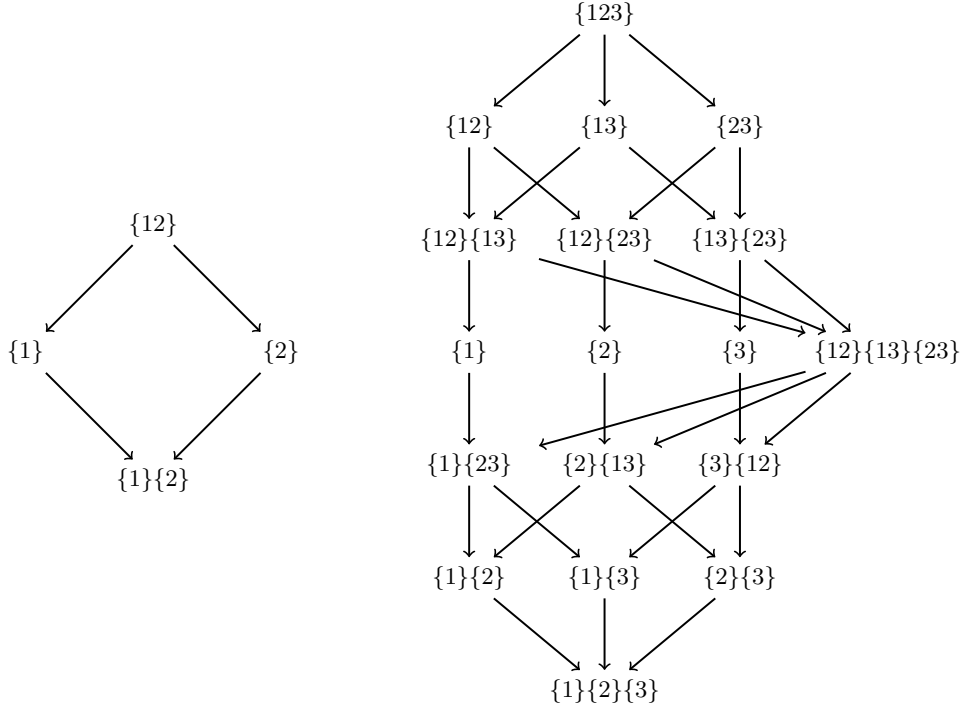

As shown in \appRef{lattices}, Axioms~\ref{ax:symmetry}--\ref{ax:sr} induce two lattices---namely
the \emph{specificity lattice} and \emph{ambiguity lattice}---which are depicted in \fig{lattices}.
Furthermore, each lattice is defined for every discrete realisation from $P(S_1, \ldots, S_n, T)$.
The redundancy measures $i_\cap^+$ or $i_\cap^-$ can be thought of as a cumulative information
functions which integrate the specificity or ambiguity uniquely contributed by each node as one
moves up each lattice.  Finally, just as in PID, performing a M\"obius inversion over each lattice
yielding the unique contributions of specificity and ambiguity from each sources event.

Similarly to PID, the specificity and ambiguity lattices provide a structure for information
decomposition, but unique evaluation requires a separate definition of redundancy.  However, unlike
PID (or even PPID), this evaluation requires both a definition of pointwise redundant specificity
and pointwise redundant ambiguity.  Before providing these definitions, it is helpful to first see
how the specificity and ambiguity lattices can be used to decompose multivariate information in the
now familiar bivariate case.

\subsection{Bivariate PPID using the Specificity and Ambiguity}
\label{sec:spec_amb_decomp}
 
Consider again the bivariate case where the aim is to decompose the information provided by $s_1$
and $s_2$ about $t$.  The specificity lattice can be used to decompose the pointwise specificity,
\begin{align}
  i^+(s_{1,2} \ra t) &= \rip + \uiop + \uitp + \cip,  \nn\\
  i^+(s_1 \ra t)      &= \rip + \uiop,                \nn\\
  i^+(s_2 \ra t)      &= \rip + \uitp; \label{eq:local_decomp_bivar_separate_spec}  \\
  \intertext{while the ambiguity lattice can be used to decompose the pointwise ambiguity,}
  i^-(s_{1,2} \ra t) &= \rin + \uion + \uitn + \cin,  \nn\\
  i^-(s_1 \ra t)      &= \rin + \uion,                \nn\\
  i^-(s_2 \ra t)      &= \rin + \uitn. \label{eq:local_decomp_bivar_separate_amb}
\end{align}
These equations share the same structural form as \eq{local_decomp_bivar} only now decompose the
specificity and the ambiguity rather than the pointwise mutual information , e.g.\ $\rip$ denotes
the redundant specificity while $\uion$ denoted the unique ambiguity from $s_1$.  Just as in for
\eq{local_decomp_bivar}, this decomposition could be considered for every discrete realisation on
the support of the joint distribution \mbox{$P(S_1,S_2,T)$}.

There are two ways one can be combine these values. Firstly, in a similar manner to
\eq{local_decomp_bivar_avg}, one could take the expectation of the atoms of specificity, or the
atoms of ambiguity, over all discrete realisations yielding the average PI atoms of specificity and
ambiguity,
\begin{align}
  \label{eq:local_decomp_bivar_avg_separate}
  \UIop &= \big\langle \uiop \big\rangle, & \UIon &= \big\langle \uion \big\rangle,  \nn\\
  \UItp &= \big\langle \uitp \big\rangle, & \UItn &= \big\langle \uitn \big\rangle,  \nn\\
  \RIp  &= \big\langle \rip  \big\rangle, & \RIn  &= \big\langle \rin  \big\rangle,  \nn\\
  \CIp  &= \big\langle \cip  \big\rangle. & \CIn  &= \big\langle \cin  \big\rangle.
\end{align}
Alternatively, one could subtract the pointwise unique, redundant and complementary ambiguity from
the pointwise unique, redundant and complementary specificity yielding the pointwise unique,
pointwise redundant and pointwise complementary information, i.e. recover the atoms from PPID,
\begin{alignat}{3}
  \label{eq:local_recomp}
  \ri  &= \rip  &&- \rin,   \nn\\
  \uio &= \uiop &&- \uion,  \nn\\
  \uit &= \uitp &&- \uitn,  \nn\\
  \ci  &= \cip  &&- \cin.
\end{alignat}
Both \eq{local_decomp_bivar_avg_separate} and \eq{local_recomp} are linear operations, hence one
could perform both of these operations (in either order) to obtain the average unique, average
redundant and average complementary information, i.e.\ recover the atoms from PID,
\begin{alignat}{3}
  \label{eq:recomp}
  \RI  &= \RIp  &&- \RIn,   \nn\\
  \UIo &= \UIop &&- \UIon,  \nn\\
  \UIt &= \UItp &&- \UItn,  \nn\\
  \CI  &= \CIp  &&- \CIn.
\end{alignat}

\subsection{Redundancy Measures on the Specificity and Ambiguity Lattices}
\label{sec:measure}

\newcommand{\rmin}{r_\text{min}}
\newcommand{\rminp}{\rmin^+}
\newcommand{\rminn}{\rmin^-}
\newcommand{\rminpn}{\rmin^\pm}
\newcommand{\Rmin}{R_\text{min}}
\newcommand{\Rminp}{\Rmin^+}
\newcommand{\Rminn}{\Rmin^-}
\newcommand{\Rminpn}{\Rmin^\pm}

\newenvironment{dudproof}
{\renewcommand{\qedsymbol}{}\proof}
{\endproof}

Now that we have a structure for our information decomposition, there is a need to provide a
definition of the pointwise redundant specificity and pointwise redundant ambiguity.  However,
before attempting to provide such a definition, there is a need to consider Remark~\ref{re:crucial}
and the operational interpretation of in \secRef{op_interp_red}.  In particular, the pointwise
redundant specificity $\icp$ and pointwise redundant ambiguity $\icn$ should only depend on the size
of informative and misinformative exclusions.  They should not depend on the apportionment of the
informative exclusions across the set of elementary events contained in the complementary event
$\ob{t}$.  Formally, this will be requirement will be enshrined via the following axiom.

\begin{axiom}[Two-event Partition]
  \label{ax:local_comp}
  The pointwise redundant specificity $\icp$ and pointwise redundant ambiguity $\icn$ are functions
  of the probability measures on the two-event partitions
  $\mc{A}_1^{\mb{a}_1} \by \mc{T}^t, \dots ,\mc{A}_k^{\mb{a}_k} \by \mc{T}^t$.
\end{axiom}

Since the pointwise redundant specificity $\icp$ is specificity associated with the source event
which induces the smallest total exclusions, and pointwise redundant ambiguity $\icn$ is the
ambiguity associated with the source event which induces the smallest misinformative exclusion,
consider the following definitions.

\begin{Definition}
  \label{def:specificity}
  The pointwise redundant specificity is given by
  \begin{equation}
    \label{eq:red_specificity}
    \rminp \big( \mb{a}_1, \dots, \mb{a}_k \ra t \big)
      = \min_{\mb{a}_i}\, i^+(\mb{a}_i \ra t)
      = \min_{\mb{a}_i}\, h(\mb{a}_i).
  \end{equation}
\end{Definition}

\begin{Definition}
  \label{def:ambiguity}
  The pointwise redundant ambiguity is given by
  \begin{equation}
    \label{eq:red_ambiguity}
    \rminn \big( \mb{a}_1, \dots, \mb{a}_k \ra t \big)
      = \min_{\mb{a}_i}\, i^-(\mb{a}_i \ra t)
      = \min_{\mb{a}_j}\, h(\mb{a}_j | t).
  \end{equation}
\end{Definition}

\ifarXiv
  \pagebreak
\else
  % blank
\fi

\begin{restatable}{restatetheorem}{theoremsatisfy}
  \label{thm:satisfy}
  The definitions of $\rminp$ and $\rminn$ satisfy Axioms~\ref{ax:symmetry}--\ref{ax:local_comp}.
\end{restatable}

\begin{restatable}{restatetheorem}{theoremmonotonically}
  \label{thm:monotonically}
  The redundancy measures $\rminp$ and $\rminn$ increase monotonically on the
  $\big\langle \msc{A}(\mb{s}),\preceq \big\rangle$.
\end{restatable}

\begin{restatable}{restatetheorem}{theoremnonnegativity}
  \label{thm:nonnegativity}
  The atoms of partial specificity $\pi^+$ and partial ambiguity $\pi^-$ evaluated using the
  measures $\rminp$ and $\rminn$ on the specificity and ambiguity lattices (respectively), are
  non-negative.
\end{restatable}

\appRef{red} contains the proof of Theorems~\ref{thm:satisfy}--\ref{thm:nonnegativity} and further
relevant consideration of Defintions~\ref{def:specificity} and \ref{def:ambiguity}.  As in
\eq{local_decomp_bivar_avg_separate}, one can take the expectation of the either the pointwise
redundant specificity $\rminp$ or the pointwise redundant ambiguity $\rminn$ to get the average
redundant specificity $\Rminp$ or the average redundant ambiguity $\Rminn$.  Alternatively, just as
in \eq{local_recomp}, one can recombine the pointwise redundant specificity $\rminp$ and the
pointwise redundant ambiguity $\rminn$ to get the pointwise redundant information $\rmin$.  Finally,
as per \eq{recomp}, one could perform both of these (linear) operations in either order to obtain
the average redundant information $\Rmin$.  Note that while Theorem~\ref{thm:nonnegativity} proves
that the atoms of partial specificity $\pi^+$ and partial ambiguity $\pi^-$ are non-negative, it is
trivial to see that $\rmin$ could be negative since when source events can redundantly provide
misinformation about a target event.  As shown in the following theorem, $\Rmin$ can also be
negative.

\begin{restatable}{restatetheorem}{theoremnonnegativityavg}
  \label{thm:nonnegativity_avg}
  The atoms of partial average information $\Pi$ evaluated by recombining and averaging $\pi^\pm$
  are not non-negative. 
\end{restatable}

This means that the measure $\Rmin$ does not satisfy local positivity.  Nonetheless the negativity
of $\Rmin$ is readily explainable in terms of the operational interpretation of
\secRef{op_interp_red}, as will be discussed further in \secRef{imprdn}.  However, failing to
satisfy local positivity does mean that $\rmin$ and $\Rmin$ do not satisfy the \emph{target
  monotonicity} property first discussed in \citet{bertschinger2013}.  Despite this, as the
following theorem shows, the measures do satisfy the target chain rule.

\begin{restatable}[Pointwise Target Chain Rule]{restatetheorem}{theoremptcr}
  \label{thm:tar_cr}
  Given the joint target realisation $t_{1,2}$, the pointwise redundant information $\rmin$
  satisfies the following chain rule,
  \begin{align}
    \label{eq:target_chain_rule_1}
    \rmin \big( \mb{a}_1,\ldots,\mb{a}_k \ra t_{1,2} \big)
      &= \rmin \big( \mb{a}_1,\ldots,\mb{a}_k \ra t_1 \big)
        + \rmin \big( \mb{a}_1,\ldots,\mb{a}_k \ra t_2 | t_1 \big),  \\
      \label{eq:target_chain_rule_2}
      &= \rmin \big( \mb{a}_1,\ldots,\mb{a}_k \ra t_2 \big)
        + \rmin \big( \mb{a}_1,\ldots,\mb{a}_k \ra t_1 | t_2 \big).
  \end{align}
\end{restatable}

The proof of the last theorem is deferred to \appRef{tcr}. Note that since the expectation is a
linear operation, Theorem~\ref{thm:tar_cr} also holds for the average redundant information $\Rmin$.
Furthermore, as these results apply to any of the source events, the target chain rule will hold for
any of the PPI atoms, e.g.\ \eq{local_recomp}, and any of the PI atoms, e.g.\ \eq{recomp}.  However,
no such rule holds for the pointwise redundant specificity or ambiguity.  The specificity depends
only on the predictor event, i.e.\ does not depend on the target events.  As such, when an
increasing number of target events are considered, the specificity remains unchanged.  Hence, a
target chain rule cannot hold for the specificity, or the ambiguity alone.

%%%%%%%%%%%%%%%%%%%%%%%%%%%%%%%%%%%%%%%%%%%%%%%%%%%%%%%%%%%%%%%%%%%%%%%%%%%%%%%%%%%%%%%%%%%%%%%%%%%%
%%%%%%%%%%%%%%%%%%%%%%%%%%%%%%%%%%%%%%%%%%%%%%%%%%%%%%%%%%%%%%%%%%%%%%%%%%%%%%%%%%%%%%%%%%%%%%%%%%%%

\section{Discussion}
\label{sec:discussion}

PPID using the specificity and ambiguity takes the ideas underpinning PID and applies them on a
pointwise scale while circumventing the monotonicity issue associated with the signed pointwise
mutual information.  This section will explore the various properties of the decomposition in an
example driven manner and compare the results to the most widely-used measures from the existing PID
literature. (Further examples can be found in \appRef{additional_examples}.)  The following
shorthand notation will be utilised in the figures throughout this section:
\begin{align*}
  i^+_1 &= i^+(s_1 \ra t), & i^+_2 &= i^+(s_2 \ra t), & i^+_{1,2} &= i^+(s_{1,2} \ra t),  \\
  i^-_1 &= i^-(s_1 \ra t), & i^-_2 &= i^-(s_2 \ra t), & i^-_{1,2} &= i^-(s_{1,2} \ra t),
\end{align*}
\vspace{-2.5em}
\begin{align*}
  u^+_1 &= \uiop, & u^+_2 &= \uitp, & r^+ &= \rip, & c^+ &= \cip,  \\
  u^-_1 &= \uion, & u^-_2 &= \uitn, & r^- &= \rin, & c^- &= \cin.  
\end{align*}

\ifarXiv
  \pagebreak
\else
  % blank
\fi

\subsection{Comparison to Existing Measures}

A similar approach to the decomposition presented in this paper is due to \citet{ince2017}, who also
sought to define a pointwise information decomposition.  Despite the similarity in this regard, the
redundancy measure $\iccs$ presented in \citep{ince2017} approaches the pointwise monotonicity
problem of \secRef{ppid} in a different way to the decomposition presented in this paper.
Specifically, $\iccs$ aims to utilise the pointwise co-information as a measure of pointwise
redundant information since it ``quantifies the set-theoretic overlap of the two univariate
[pointwise] information values''~\citep[p.~14]{ince2017}.  There are, however, difficulties with
this approach.  Firstly (unlike the average mutual information and the Shannon inequalities), there
are no inequalities which support this interpretation of pointwise co-information as the
set-theoretic overlap of the univariate pointwise information terms---indeed, both the univariate
pointwise information and the pointwise co-information are signed measures.  Secondly, the pointwise
co-information conflates the pointwise redundant information with the pointwise complementary
information, since by \eq{local_decomp_bivar} we have that
\begin{equation}
  \label{eq:coi}
  co\text{-}i(s_1;s_2;t) \vcentcolon= i(s_1;t) + i(s_2;t) - i(s_{1,2},t) = \ri - \ci.
\end{equation}
Aware of these difficulties, Ince defines $\iccs$ such that it only interprets the pointwise
co-information as a measure of set-theoretic overlap in the case where all three pointwise
information terms have the same sign, arguing that these are the only situations which admit a clear
interpretation in terms of a common change in surprisal.  In the other difficult to interpret
situations, $\iccs$ defines the pointwise redundant information to be zero.  This approach
effectively assumes that $\ci=0$ in \eq{coi} when $i(s_1;t)$, $i(s_2;t)$ and
$co\text{-}i(s_1;s_2;t)$ all have the same sign.

In a subsequent paper, \citet{ince2017ped} also presented a partial entropy decomposition which aims
to decompose multivariate entropy rather than multivariate information.  As such, this decomposition
is more similar to PPID using specificity and ambiguity than Ince's aforementioned decomposition.
Although similar in this regard, the measure of pointwise redundant entropy $H_\text{cs}$ presented
in \citep{ince2017ped} takes a different approach to the one presented in this paper.  Specifically,
$H_\text{cs}$ also uses the pointwise co-information as a measure of set-theoretic overlap and hence
as a measure of pointwise redundant entropy.  As the pointwise entropy is unsigned, the difficulties
are reduced but remain present due to the signed pointwise co-information.  In a manner similar to
$\iccs$, Ince defines $H_\text{cs}$ such that it only interprets the pointwise co-information as a
measure of set-theoretic overlap when it is positive.  As per $\iccs$, this effectively assumes that
$\ci=0$ in \eq{coi} when all information terms have the same sign.  When the pointwise
co-information is negative, $H_\text{cs}$ simply ignores the co-information by defining the
pointwise redundant information to be zero.  In contrast to both of Ince's approaches, PPID using
specificity and ambiguity does not dispose of the set-theoretic intuition in these difficult to
interpret situations.  Rather, our approach considers the notion of redundancy in terms of
overlapping exclusions---i.e.\ in terms of the underlying, unsigned measures which are amenable to a
set-theoretic interpretation.

The measures of pointwise redundant specificity $\rminp$ and pointwise redundant ambiguity $\rminn$,
from Definitions~\ref{def:specificity} and \ref{def:ambiguity} are also similar to both the minimum
mutual information $I_\text{\textsc{mmi}}$~\citep{barrett2015} and the original PID redundancy
measure $\imin$ \citep{williams2010}.  Specifically, all three of these approaches consider the
redundant information to be the minimum information provided about a target event $t$.  The
difference is that $\imin$ applies this idea to the sources $\mb{A}_1,\ldots,\mb{A}_k$, i.e.\ to
collections of entire predictor variables from $\mb{S}$, while $\rminpn$ apply this notion to the
source events $\mb{a}_1,\ldots,\mb{a}_k$, i.e.\ to collections of predictor events from $\mb{s}$.
In other words, while the measure $\imin$ can be regarded as being semi-pointwise (since it
considers the information provided by the variables $S_1,\ldots,S_n$ about an event $t$), the
measures $\rminpn$ are fully pointwise (since they consider the information provided by the events
$s_1,\ldots,s_n$ about an event $t$).  This difference in approach is most apparent in the
probability distribution \locunq{}---unlike PID, PPID using the specificity and ambiguity respects
the pointwise nature of information, as we will see in \secRef{ex_locunq}.

PPID using specificity and ambiguity also share certain similarities with the bivariate PID induced
by the measure $\uitilde$ of \citet{bertschinger2014}.  Firstly, Axiom~\ref{ax:local_comp} can be
considered to be a pointwise adaptation of their Assumption~($*$), i.e.\ the measures $\rminpn$
depend only on the marginal distributions $P(S_1,T)$ and $P(S_2,T)$ with respect to the two-event
partitions $\mc{S}_1^{s_1} \by \mc{T}^t$ and $\mc{S}_2^{s_2} \by \mc{T}^t$.  Secondly, in PPID using
specificity and ambiguity, the only way one can only decide if there is complementary information
$\ci$ is by knowing the joint distribution $P(S_1, S_2, T)$ with respect to the joint two-event
partitions $\mc{S}_1^{s_1} \by \mc{S}_2^{s_2} \by \mc{T}^t$.  This is (in effect) a pointwise form
of their Assumption~($**$).  Thirdly, by definition $\rminpn$ are given by the minimum value that
any one source event provides.  This is the largest possible value that one could take for these
quantities whilst still requiring that the unique specificity and ambiguity be non-negative.  Hence,
within each discrete realisation, $\rminpn$ minimise the unique specificity and ambiguity whilst
maximising the redundant specificity and ambiguity.  This is similar to $\uitilde$ which minimises
the (average) unique information while still satisfying Assumption~($*$).  Finally, note that since
the measure $\ivk$ produces a bivariate decomposition which is equivalent to that of $\uitilde$
\citep{bertschinger2014}, the same similarities apply between PPID using specificity and ambiguity
and the decomposition induced by $\ivk$ from \citet{griffith2012}.

\subsection{Probability Distribution X\textsc{or}}

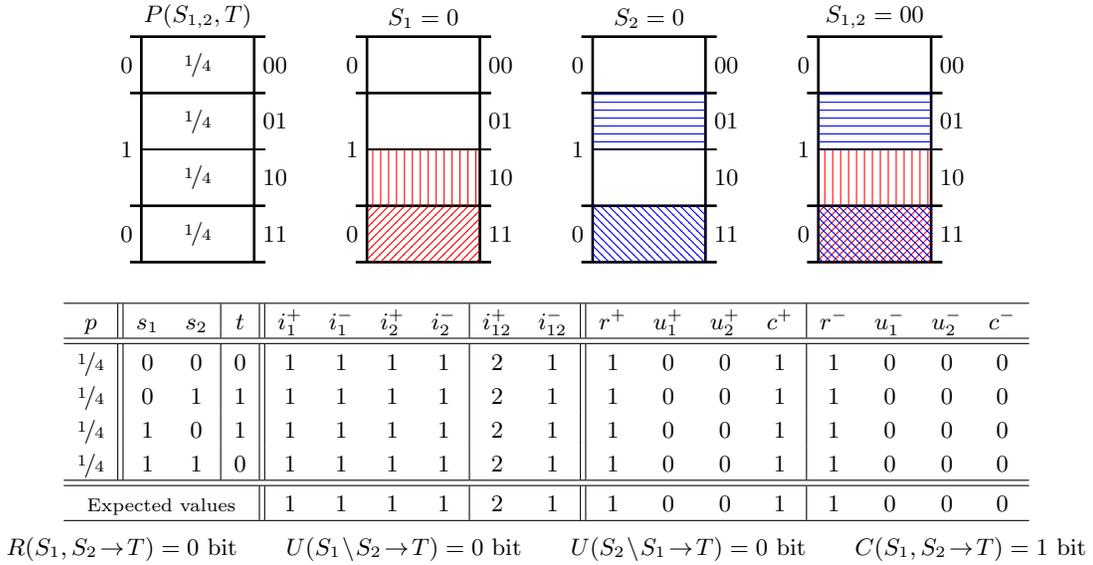
\begin{figure}[t]
 \centering %% \tablesize{} %% You can specify the fontsize here, e.g.
  \begin{tikzpicture}

  \def\originy{0} 
  \def\height{3}
  \def\width{1.5}
  \def\overdrawn{0.15}

  % -----------
  
  \def\originx{0}
  
  \draw[black, line width=1pt] (\originx,\originy)
    -- (\originx,\originy+\height);
  \draw[black, line width=1pt] (\originx-\overdrawn,\originy+\height)
    -- (\originx+\width+\overdrawn,\originy+\height);
  \draw[black, line width=1pt] (\originx+\width,\originy+\height)
    -- (\originx+\width,\originy);
  \draw[black, line width=1pt] (\originx+\width+\overdrawn,\originy)
    -- (\originx-\overdrawn,\originy);

  \draw[black, line width=1pt] (\originx-\overdrawn,\originy+3/4*\height)
    -- (\originx+\width+\overdrawn,\originy+3/4*\height);
  \draw[black, line width=0.75pt] (\originx,\originy+\height/2)
    -- (\originx+\width+\overdrawn,\originy+\height/2);
  \draw[black, line width=1pt] (\originx-\overdrawn,\originy+1/4*\height)
  -- (\originx+\width+\overdrawn,\originy+1/4*\height);
  
%  \draw[pattern=north east lines, pattern color=red] (\originx,\originy+\height/2) rectangle
%    (\originx+\width,\originy+3/4*\height);

  \node at (\originx,7/8*\height)[anchor=east] {$0$};
  \node at (\originx,1/2*\height)[anchor=east] {$1$};
  \node at (\originx,1/8*\height)[anchor=east] {$0$}; 
  \node at (\originx+\width,7/8*\height)[anchor=west] {$00$};
  \node at (\originx+\width,5/8*\height)[anchor=west] {$01$};
  \node at (\originx+\width,3/8*\height)[anchor=west] {$10$};
  \node at (\originx+\width,1/8*\height)[anchor=west] {$11$};
  \node at (\originx+\width/2,\originy+\height)[anchor=south] {$P(S_{1,2},T)$};
  \node at (\originx+\width/2,\originy+7/8*\height)[anchor=center] {$\nf{1}{4}$};
  \node at (\originx+\width/2,\originy+5/8*\height)[anchor=center] {$\nf{1}{4}$};
  \node at (\originx+\width/2,\originy+3/8*\height)[anchor=center] {$\nf{1}{4}$};
  \node at (\originx+\width/2,\originy+1/8*\height)[anchor=center] {$\nf{1}{4}$};
  
  % -----------
  
  \def\originxtwo{3}
  
  \draw[black, line width=1pt] (\originxtwo,\originy)
    -- (\originxtwo,\originy+\height);
  \draw[black, line width=1pt] (\originxtwo-\overdrawn,\originy+\height)
    -- (\originxtwo+\width+\overdrawn,\originy+\height);
  \draw[black, line width=1pt] (\originxtwo+\width,\originy+\height)
    -- (\originxtwo+\width,\originy);
  \draw[black, line width=1pt] (\originxtwo+\width+\overdrawn,\originy)
    -- (\originxtwo-\overdrawn,\originy);

  \draw[black, line width=1pt] (\originxtwo-\overdrawn,\originy+3/4*\height)
    -- (\originxtwo+\width+\overdrawn,\originy+3/4*\height);
  \draw[black, line width=0.75pt] (\originxtwo,\originy+\height/2)
    -- (\originxtwo+\width+\overdrawn,\originy+\height/2);
  \draw[black, line width=1pt] (\originxtwo-\overdrawn,\originy+1/4*\height)
  -- (\originxtwo+\width+\overdrawn,\originy+1/4*\height);

  \draw[pattern=vertical lines, pattern color=red] (\originxtwo,\originy+1/4*\height) rectangle
    (\originxtwo+\width,\originy+1/2*\height);
  \draw[pattern=north east lines, pattern color=red] (\originxtwo,\originy) rectangle
    (\originxtwo+\width,\originy+1/4*\height);
  
  \node at (\originxtwo,7/8*\height)[anchor=east] {$0$};
  \node at (\originxtwo,1/2*\height)[anchor=east] {$1$};
  \node at (\originxtwo,1/8*\height)[anchor=east] {$0$}; 
  \node at (\originxtwo+\width,7/8*\height)[anchor=west] {$00$};
  \node at (\originxtwo+\width,5/8*\height)[anchor=west] {$01$};
  \node at (\originxtwo+\width,3/8*\height)[anchor=west] {$10$};
  \node at (\originxtwo+\width,1/8*\height)[anchor=west] {$11$};
  \node at (\originxtwo+\width/2,\originy+\height)[anchor=south] {$S_1=0$};

  % -----------
  
  \def\originxthree{6}
  
  \draw[black, line width=1pt] (\originxthree,\originy)
    -- (\originxthree,\originy+\height);
  \draw[black, line width=1pt] (\originxthree-\overdrawn,\originy+\height)
    -- (\originxthree+\width+\overdrawn,\originy+\height);
  \draw[black, line width=1pt] (\originxthree+\width,\originy+\height)
    -- (\originxthree+\width,\originy);
  \draw[black, line width=1pt] (\originxthree+\width+\overdrawn,\originy)
    -- (\originxthree-\overdrawn,\originy);

  \draw[black, line width=1pt] (\originxthree-\overdrawn,\originy+3/4*\height)
    -- (\originxthree+\width+\overdrawn,\originy+3/4*\height);
  \draw[black, line width=0.75pt] (\originxthree,\originy+\height/2)
    -- (\originxthree+\width+\overdrawn,\originy+\height/2);
  \draw[black, line width=1pt] (\originxthree-\overdrawn,\originy+1/4*\height)
  -- (\originxthree+\width+\overdrawn,\originy+1/4*\height);
  
  \draw[pattern=horizontal lines, pattern color=blue] (\originxthree,\originy+\height/2) rectangle
    (\originxthree+\width,\originy+3/4*\height);
  \draw[pattern=north west lines, pattern color=blue] (\originxthree,\originy) rectangle
    (\originxthree+\width,\originy+1/4*\height);
    
  \node at (\originxthree,7/8*\height)[anchor=east] {$0$};
  \node at (\originxthree,1/2*\height)[anchor=east] {$1$};
  \node at (\originxthree,1/8*\height)[anchor=east] {$0$}; 
  \node at (\originxthree+\width,7/8*\height)[anchor=west] {$00$};
  \node at (\originxthree+\width,5/8*\height)[anchor=west] {$01$};
  \node at (\originxthree+\width,3/8*\height)[anchor=west] {$10$};
  \node at (\originxthree+\width,1/8*\height)[anchor=west] {$11$};
  \node at (\originxthree+\width/2,\originy+\height)[anchor=south] {$S_2=0$};

   % -----------
  
  \def\originxfour{9}
  
  \draw[black, line width=1pt] (\originxfour,\originy)
    -- (\originxfour,\originy+\height);
  \draw[black, line width=1pt] (\originxfour-\overdrawn,\originy+\height)
    -- (\originxfour+\width+\overdrawn,\originy+\height);
  \draw[black, line width=1pt] (\originxfour+\width,\originy+\height)
    -- (\originxfour+\width,\originy);
  \draw[black, line width=1pt] (\originxfour+\width+\overdrawn,\originy)
    -- (\originxfour-\overdrawn,\originy);

  \draw[black, line width=1pt] (\originxfour-\overdrawn,\originy+3/4*\height)
    -- (\originxfour+\width+\overdrawn,\originy+3/4*\height);
  \draw[black, line width=0.75pt] (\originxfour,\originy+\height/2)
    -- (\originxfour+\width+\overdrawn,\originy+\height/2);
  \draw[black, line width=1pt] (\originxfour-\overdrawn,\originy+1/4*\height)
  -- (\originxfour+\width+\overdrawn,\originy+1/4*\height);
  
  \draw[pattern=vertical lines, pattern color=red] (\originxfour,\originy+1/4*\height) rectangle
    (\originxfour+\width,\originy+1/2*\height);
  \draw[pattern=north east lines, pattern color=red] (\originxfour,\originy) rectangle
    (\originxfour+\width,\originy+1/4*\height);
  \draw[pattern=horizontal lines, pattern color=blue] (\originxfour,\originy+\height/2) rectangle
    (\originxfour+\width,\originy+3/4*\height);
  \draw[pattern=north west lines, pattern color=blue] (\originxfour,\originy) rectangle
    (\originxfour+\width,\originy+1/4*\height);    

  \node at (\originxfour,7/8*\height)[anchor=east] {$0$};
  \node at (\originxfour,1/2*\height)[anchor=east] {$1$};
  \node at (\originxfour,1/8*\height)[anchor=east] {$0$}; 
  \node at (\originxfour+\width,7/8*\height)[anchor=west] {$00$};
  \node at (\originxfour+\width,5/8*\height)[anchor=west] {$01$};
  \node at (\originxfour+\width,3/8*\height)[anchor=west] {$10$};
  \node at (\originxfour+\width,1/8*\height)[anchor=west] {$11$};
  \node at (\originxfour+\width/2,\originy+\height)[anchor=south] {$S_{1,2}=00$};

\end{tikzpicture} \vskip 1.5em \ifarXiv
  % blank
  \else
  \tablesize{\footnotesize}
  \fi
  \setlength{\tabcolsep}{5pt}
  \begin{tabular}{c || c c | c || c c c c | c c || c c c c | c c c c }
    \hline
    $p$ & $s_1$ & $s_2$ & $t$ & $i_1^+$ & $i_1^-$ & $i_2^+$ & $i_2^-$ & $i_{12}^+$ & $i_{12}^-$
      & $r^+$ & $u_1^+$ & $u_2^+$ & $c^+$ & $r^-$ & $u_1^-$ & $u_2^-$ & $c^-$\\
    \hline\hline
    \nf{1}{4} & 0 & 0 & 0 & 1 & 1 & 1 & 1 & 2 & 1 & 1 & 0 & 0 & 1 & 1 & 0 & 0 & 0 \\
    \nf{1}{4} & 0 & 1 & 1 & 1 & 1 & 1 & 1 & 2 & 1 & 1 & 0 & 0 & 1 & 1 & 0 & 0 & 0 \\
    \nf{1}{4} & 1 & 0 & 1 & 1 & 1 & 1 & 1 & 2 & 1 & 1 & 0 & 0 & 1 & 1 & 0 & 0 & 0 \\
    \nf{1}{4} & 1 & 1 & 0 & 1 & 1 & 1 & 1 & 2 & 1 & 1 & 0 & 0 & 1 & 1 & 0 & 0 & 0 \\
    \hline\hline
    \multicolumn{4}{c||}{\scriptsize Expected values}
      & 1 & 1 & 1 & 1 & 2 & 1 & 1 & 0 & 0 & 1 & 1 & 0 & 0 & 0 \\
    \hline
  \end{tabular}
  \vskip 5pt{\small $\RI=0\bit \qquad \UIo=0\bit \qquad \UIt=0\bit \qquad \CI=1\bit$}
  \caption{\label{fig:xor} Example \xor{}. \emph{Top}:\ probability mass diagrams for the
    realisation \mbox{$(S_1=0, S_2 = 0, T=0)$}. \emph{Middle}:\ For each realisation, the pointwise
    specificity and pointwise ambiguity has been evaluated using \eq{specificity} and \eq{ambiguity}
    respectively.  The pointwise redundant specificity and pointwise redundant ambiguity are then
    determined using \eq{red_specificity} and \eq{red_ambiguity}.  The decomposition is calculated
    using \eq{local_decomp_bivar_separate_spec} and~\eq{local_decomp_bivar_separate_amb}. The
    expected specificity and ambiguity are calculated with \eq{local_decomp_bivar_avg_separate}.
    \emph{Bottom}:\ The average information is given by \eq{recomp}.  As expected, \xor{} yields
    $1\bit$ of complementary information.}
 
\end{figure}

\fig{xor} shows the canonical example of synergy, \emph{exclusive-or} (\xor{}) which considers two
independently distributed binary predictor variables $S_1$ and $S_2$ and a target variable
${T=S_1 \mathxor S_2}$.  There are several important points to note about the decomposition of
\xor{}.  Firstly, despite providing zero pointwise information, an individual predictor event does
indeed induce exclusions.  However, the informative and misinformative exclusions are perfectly
balanced such that the posterior (conditional) distribution is equal to the prior distribution,
e.g.\ see the red coloured exclusions induced by $S_1=0$ in \fig{xor}. In information-theoretic
terms, for each realisation, the pointwise specificity equals 1~bit since half of the total
probability mass remains while the pointwise ambiguity also equals 1~bit since half of the
probability mass associated with the event which subsequently occurs (i.e.\ $T=0$), remains.  These
are perfectly balanced such that when recombined, as per \eq{decomp}, the pointwise mutual
information is equal to 0~bit, as expected.

Secondly, $S_1=0$ and $S_2=0$ both induce the same exclusions with respect to the target pointwise
event space $\mc{T}^{T=0}$. Hence, as per the operational interpretation of redundancy adopted in
\secRef{op_interp_red}, there is 1~bit of pointwise redundant specificity and 1~bit of pointwise
redundant ambiguity in each realisation.  The presence of (a form of) redundancy in \xor{} is novel
amongst the existing measures in the PID literature.  (\citet{ince2017ped} also identifies a form of
redundancy in \xor{}.) Thirdly, despite the presence of this redundancy, recombining the atoms of
pointwise specificity and ambiguity for each realisation, as per \eq{local_recomp}, leaves only one
non-zero PPI atom:\ namely the pointwise complementary information ${\ci=1 \bit}$.  Furthermore,
this is true for every pointwise realisation and hence, by \eq{recomp}, the only non-zero PI atom is
the average complementary information ${\CI=1 \bit}$.

\subsection{Probability Distribution \locunq{}}
\label{sec:ex_locunq}

\fig{locunq} shows the probability distribution \locunq{} introduced in \secRef{locunq}.
Recombining the decomposition via \eq{local_recomp} yields the pointwise information decomposition
proposed in \tb{locunq}---unsurprisingly, the explicitly pointwise approach results in a
decomposition which does not suffer from the pointwise unique problem of \secRef{locunq}.

\begin{figure}[t]
 \centering %% \tablesize{} %% You can specify the fontsize here, e.g.
  \begin{tikzpicture}

  \def\originy{0} 
  \def\height{3}
  \def\width{1.5}
  \def\overdrawn{0.15}

  % -----------
  
  \def\originx{0}
  
  \draw[black, line width=1pt] (\originx,\originy)
    -- (\originx,\originy+\height);
  \draw[black, line width=1pt] (\originx-\overdrawn,\originy+\height)
    -- (\originx+\width+\overdrawn,\originy+\height);
  \draw[black, line width=1pt] (\originx+\width,\originy+\height)
    -- (\originx+\width,\originy);
  \draw[black, line width=1pt] (\originx+\width+\overdrawn,\originy)
    -- (\originx-\overdrawn,\originy);

  \draw[black, line width=0.75pt] (\originx,\originy+3/4*\height)
    -- (\originx+\width+\overdrawn,\originy+3/4*\height);
  \draw[black, line width=1pt] (\originx-\overdrawn,\originy+\height/2)
    -- (\originx+\width+\overdrawn,\originy+\height/2);
  \draw[black, line width=0.75pt] (\originx,\originy+1/4*\height)
  -- (\originx+\width+\overdrawn,\originy+1/4*\height);
  
  \node at (\originx,3/4*\height)[anchor=east] {$1$};
  \node at (\originx,1/4*\height)[anchor=east] {$2$}; 
  \node at (\originx+\width,7/8*\height)[anchor=west] {$01$};
  \node at (\originx+\width,5/8*\height)[anchor=west] {$10$};
  \node at (\originx+\width,3/8*\height)[anchor=west] {$02$};
  \node at (\originx+\width,1/8*\height)[anchor=west] {$20$};
  \node at (\originx+\width/2,\originy+\height)[anchor=south] {$P(S_{1,2},T)$};
  \node at (\originx+\width/2,\originy+7/8*\height)[anchor=center] {$\nf{1}{4}$};
  \node at (\originx+\width/2,\originy+5/8*\height)[anchor=center] {$\nf{1}{4}$};
  \node at (\originx+\width/2,\originy+3/8*\height)[anchor=center] {$\nf{1}{4}$};
  \node at (\originx+\width/2,\originy+1/8*\height)[anchor=center] {$\nf{1}{4}$};
  
  % -----------
  
  \def\originxtwo{3}
  
  \draw[black, line width=1pt] (\originxtwo,\originy)
    -- (\originxtwo,\originy+\height);
  \draw[black, line width=1pt] (\originxtwo-\overdrawn,\originy+\height)
    -- (\originxtwo+\width+\overdrawn,\originy+\height);
  \draw[black, line width=1pt] (\originxtwo+\width,\originy+\height)
    -- (\originxtwo+\width,\originy);
  \draw[black, line width=1pt] (\originxtwo+\width+\overdrawn,\originy)
    -- (\originxtwo-\overdrawn,\originy);

  \draw[black, line width=0.75pt] (\originxtwo,\originy+3/4*\height)
    -- (\originxtwo+\width+\overdrawn,\originy+3/4*\height);
  \draw[black, line width=1pt] (\originxtwo-\overdrawn,\originy+\height/2)
    -- (\originxtwo+\width+\overdrawn,\originy+\height/2);
  \draw[black, line width=0.75pt] (\originxtwo,\originy+1/4*\height)
  -- (\originxtwo+\width+\overdrawn,\originy+1/4*\height);
  
  \draw[pattern=north east lines, pattern color=red] (\originxtwo,\originy+\height/2) rectangle
    (\originxtwo+\width,\originy+3/4*\height);
  \draw[pattern=vertical lines, pattern color=red] (\originxtwo,\originy) rectangle
    (\originxtwo+\width,\originy+1/4*\height);

  \node at (\originxtwo,3/4*\height)[anchor=east] {$1$};
  \node at (\originxtwo,1/4*\height)[anchor=east] {$2$}; 
  \node at (\originxtwo+\width,7/8*\height)[anchor=west] {$01$};
  \node at (\originxtwo+\width,5/8*\height)[anchor=west] {$10$};
  \node at (\originxtwo+\width,3/8*\height)[anchor=west] {$02$};
  \node at (\originxtwo+\width,1/8*\height)[anchor=west] {$20$};
  \node at (\originxtwo+\width/2,\originy+\height)[anchor=south] {$S_1=0$};

  % -----------
  
  \def\originxthree{6}
  
  \draw[black, line width=1pt] (\originxthree,\originy)
    -- (\originxthree,\originy+\height);
  \draw[black, line width=1pt] (\originxthree-\overdrawn,\originy+\height)
    -- (\originxthree+\width+\overdrawn,\originy+\height);
  \draw[black, line width=1pt] (\originxthree+\width,\originy+\height)
    -- (\originxthree+\width,\originy);
  \draw[black, line width=1pt] (\originxthree+\width+\overdrawn,\originy)
    -- (\originxthree-\overdrawn,\originy);

  \draw[black, line width=0.75pt] (\originxthree,\originy+3/4*\height)
    -- (\originxthree+\width+\overdrawn,\originy+3/4*\height);
  \draw[black, line width=1pt] (\originxthree-\overdrawn,\originy+\height/2)
    -- (\originxthree+\width+\overdrawn,\originy+\height/2);
  \draw[black, line width=0.75pt] (\originxthree,\originy+1/4*\height)
  -- (\originxthree+\width+\overdrawn,\originy+1/4*\height);
  
  \draw[pattern=north west lines, pattern color=blue] (\originxthree,\originy+\height/2) rectangle
    (\originxthree+\width,\originy+3/4*\height);
  \draw[pattern=horizontal lines, pattern color=blue] (\originxthree,\originy) rectangle
    (\originxthree+\width,\originy+\height/2);
    
  \node at (\originxthree,3/4*\height)[anchor=east] {$1$};
  \node at (\originxthree,1/4*\height)[anchor=east] {$2$}; 
  \node at (\originxthree+\width,7/8*\height)[anchor=west] {$01$};
  \node at (\originxthree+\width,5/8*\height)[anchor=west] {$10$};
  \node at (\originxthree+\width,3/8*\height)[anchor=west] {$02$};
  \node at (\originxthree+\width,1/8*\height)[anchor=west] {$20$};
  \node at (\originxthree+\width/2,\originy+\height)[anchor=south] {$S_2=1$};

   % -----------
  
  \def\originxfour{9}
  
  \draw[black, line width=1pt] (\originxfour,\originy)
    -- (\originxfour,\originy+\height);
  \draw[black, line width=1pt] (\originxfour-\overdrawn,\originy+\height)
    -- (\originxfour+\width+\overdrawn,\originy+\height);
  \draw[black, line width=1pt] (\originxfour+\width,\originy+\height)
    -- (\originxfour+\width,\originy);
  \draw[black, line width=1pt] (\originxfour+\width+\overdrawn,\originy)
    -- (\originxfour-\overdrawn,\originy);

  \draw[black, line width=0.75pt] (\originxfour,\originy+3/4*\height)
    -- (\originxfour+\width+\overdrawn,\originy+3/4*\height);
  \draw[black, line width=1pt] (\originxfour-\overdrawn,\originy+\height/2)
    -- (\originxfour+\width+\overdrawn,\originy+\height/2);
  \draw[black, line width=0.75pt] (\originxfour,\originy+1/4*\height)
  -- (\originxfour+\width+\overdrawn,\originy+1/4*\height);

  \draw[pattern=north east lines, pattern color=red] (\originxfour,\originy+\height/2) rectangle
    (\originxfour+\width,\originy+3/4*\height);
  \draw[pattern=vertical lines, pattern color=red] (\originxfour,\originy) rectangle
    (\originxfour+\width,\originy+1/4*\height);
  \draw[pattern=north west lines, pattern color=blue] (\originxfour,\originy+\height/2) rectangle
    (\originxfour+\width,\originy+3/4*\height);
  \draw[pattern=horizontal lines, pattern color=blue] (\originxfour,\originy) rectangle
    (\originxfour+\width,\originy+\height/2);

  \node at (\originxfour,3/4*\height)[anchor=east] {$1$};
  \node at (\originxfour,1/4*\height)[anchor=east] {$2$}; 
  \node at (\originxfour+\width,7/8*\height)[anchor=west] {$01$};
  \node at (\originxfour+\width,5/8*\height)[anchor=west] {$10$};
  \node at (\originxfour+\width,3/8*\height)[anchor=west] {$02$};
  \node at (\originxfour+\width,1/8*\height)[anchor=west] {$20$};
  \node at (\originxfour+\width/2,\originy+\height)[anchor=south] {$S_{1,2}=01$};

\end{tikzpicture} \vskip 1.5em \ifarXiv
  % blank
  \else
  \tablesize{\footnotesize}
  \fi
  \setlength{\tabcolsep}{5pt}
  \begin{tabular}{c || c c | c || c c c c | c c || c c c c | c c c c }
    \hline
    $p$ & $s_1$ & $s_2$ & $t$ & $i_1^+$ & $i_1^-$ & $i_2^+$ & $i_2^-$ & $i_{12}^+$ & $i_{12}^-$
      & $r^+$ & $u_1^+$ & $u_2^+$ & $c^+$ & $r^-$ & $u_1^-$ & $u_2^-$ & $c^-$\\
    \hline\hline
    \nf{1}{4} & 0 & 1 & 1 & 1 & 1 & 2 & 1 & 2 & 1 & 1 & 0 & 1 & 0 & 1 & 0 & 0 & 0 \\
    \nf{1}{4} & 1 & 0 & 1 & 2 & 1 & 1 & 1 & 2 & 1 & 1 & 1 & 0 & 0 & 1 & 0 & 0 & 0 \\
    \nf{1}{4} & 0 & 2 & 2 & 1 & 1 & 2 & 1 & 2 & 1 & 1 & 0 & 1 & 0 & 1 & 0 & 0 & 0 \\
    \nf{1}{4} & 2 & 0 & 2 & 2 & 1 & 1 & 1 & 2 & 1 & 1 & 1 & 0 & 0 & 1 & 0 & 0 & 0 \\
    \hline\hline
    \multicolumn{4}{c||}{\scriptsize Expected values}
      & \nf{3}{2} & 1 & \nf{3}{2} & 1 & 2 & 1 & 1 & \nf{1}{2} & \nf{1}{2} & 0 & 1 & 0 & 0 & 0 \\
    \hline
  \end{tabular}
  \vskip 5pt{\small
    $\RI=0\bit \qquad \UIo=\nf{1}{2}\bit \qquad \UIt=\nf{1}{2}\bit \qquad \CI=0\bit$}
  \caption{\label{fig:locunq} Example \locunq{}.  \emph{Top}:\ probability mass diagrams for the
    realisation ${(S_1\!=\!0, S_2\!=\!1,T\!=\!1)}$.  \emph{Middle}:\ For each realisation, the PPID
    using specificity and ambiguity is evaluated (see \fig{xor} for details).  Upon recombination as
    per \eq{local_recomp}, the PPI decomposition from \tb{locunq} is attained.  \emph{Bottom}:\ as
    does the average information---the decomposition does not have the pointwise unique problem.}
\end{figure}
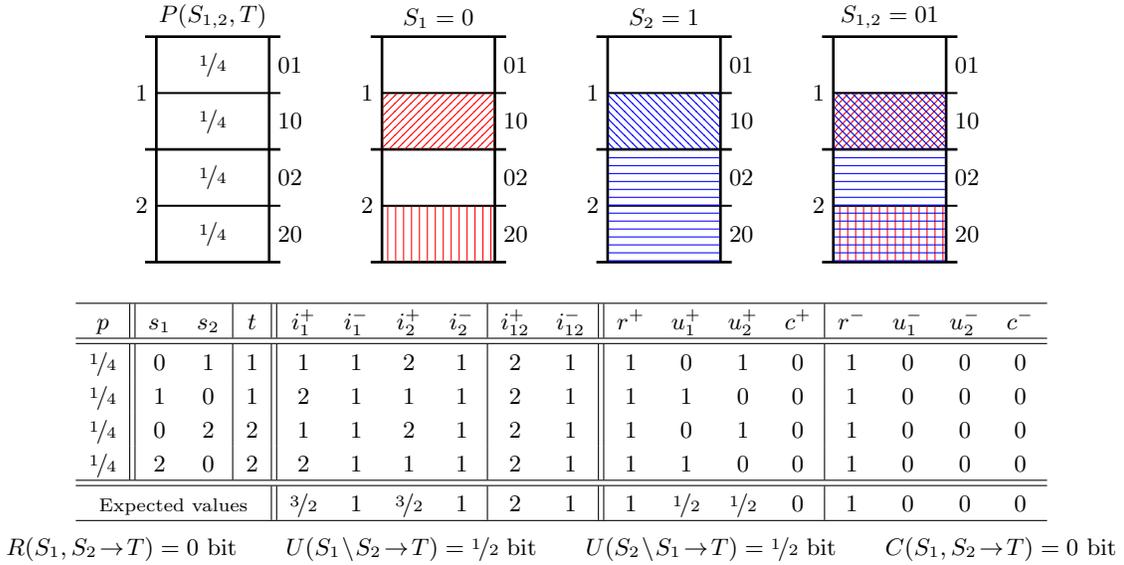

In each realisation, observing a $0$ in either source provides the same balanced informative and
misinformative exclusions as in \xor{}.  Observing either a $1$ or $2$ provides the same
misinformative exclusion as observing the $0$, but provides a larger informative exclusion than $0$.
This leaves only the probability mass associated with the event which subsequently occurs remaining
(hence why observing a $1$ and $2$ is fully informative about the target).  Information
theoretically, in each realisation the predictor events provide 1~bit of redundant pointwise
specificity and 1~bit of redundant pointwise ambiguity while the fully informative event
additionally provides 1~bit of unique specificity.

\subsection{Probability Distribution \imprdn{}}
\label{sec:imprdn}

\fig{imprdn} shows the probability distribution \emph{redundant-error} (\imprdn{}) which considers
two predictors which are nominally redundant and fully informative about the target, but where one
predictor occasionally makes an erroneous prediction.  Specifically, \fig{imprdn} shows the
decomposition of \imprdn{} where $S_2$ makes an error with a probability $\varepsilon=\nf{1}{4}$.
The important feature to note about this probability distribution is that upon recombining the
specificity and ambiguity and taking the expectation over every realisation, the resultant average
unique information from $S_2$ is $\UIt=-0.811\bit$.

On first inspection, the result that the average unique information can be negative may seem
problematic; however, it is readily explainable in terms of the operational interpretation of
\secRef{op_interp_red}.  In \imprdn{}, a source event always excludes exactly $\nf{1}{2}$ of the
total probability mass, thus every realisation contains $1\bit$ of redundant pointwise specificity.
The events of the error-free $S_1$ induce only informative exclusions and as such provide $0\bit$ of
pointwise ambiguity in each realisation.  In contrast, the events in the error-prone $S_2$ always
induce a misinformative exclusion, meaning that $S_2$ provides unique pointwise ambiguity in every
realisation.  Since $S_2$ never provides unique specificity, the average unique information is
negative on average.

\begin{figure}[t]
 \centering %% \tablesize{} %% You can specify the fontsize here, e.g.
 \vspace{-5pt}
 \input{figures/pmass_imprdn_new.tikz} \vskip 1em \ifarXiv
  % blank
  \else
  \tablesize{\footnotesize}
  \fi
  \setlength{\tabcolsep}{5pt}
  \begin{tabular}{c || c c | c || c c c c | c c || c c c c | c c c c }
    \hline
    $p$ & $s_1$ & $s_2$ & $t$ & $i_1^+$ & $i_1^-$ & $i_2^+$ & $i_2^-$ & $i_{12}^+$ & $i_{12}^-$
      & $r^+$ & $u_1^+$ & $u_2^+$ & $c^+$ & $r^-$ & $u_1^-$ & $u_2^-$ & $c^-$\\
    \hline\hline
    \nf{3}{8}&0&0&0& 1&0&1&$\lg\nf{4}{3}$&$\lg\nf{8}{3}$&$\lg\nf{4}{3}$&1&0&0&$\lg\nf{4}{3}$&0&0&$\lg\nf{4}{3}$&0 \\
    \nf{3}{8}&1&1&1& 1&0&1&$\lg\nf{4}{3}$&$\lg\nf{8}{3}$&$\lg\nf{4}{3}$&1&0&0&$\lg\nf{4}{3}$&0&0&$\lg\nf{4}{3}$&0 \\
    \nf{1}{8}&0&1&0& 1&0&1&2               &3             &2             &1&0&0&2            &0&0&2             &0 \\
    \nf{1}{8}&1&0&1& 1&0&1&2               &3             &2             &1&0&0&2            &0&0&2             &0 \\
    \hline\hline
    \multicolumn{4}{c||}{\scriptsize Expected values}
      & 1 & 0 & 1 & 0.811 & 1.811 & 0.811 & 1 & 0 & 0 & 0.811 & 0 & 0 & 0.811 & 0 \\
    \hline
  \end{tabular}
  \vskip 5pt{\small $\RI=1\bit \quad \UIo=0\bit \quad \UIt=-0.811\bit \quad \CI=0.811\bit$}
  \caption{\label{fig:imprdn} Example \imprdn{}.  \emph{Top}:\ probability mass diagrams for
    the realisations ${(S_1\!=\!0,S_2\!=\!0,T\!=\!0)}$ and ${(S_1\!=\!0,S_2\!=\!1,T\!=\!0)}$.
    \emph{Middle}:\ for each realisation, the PPID using specificity and ambiguity is evaluated (see
    \fig{xor} for details).  \emph{Bottom}:\ the average PI atoms may be negative as the
    decomposition does not satisfy local positivity.}
\end{figure}

Despite the negativity of the average unique information, in is important to observe that $S_2$
provides $0.189\bit$ of information since $S_2$ also provides $1\bit$ of average redundant
information.  It is not that $S_2$ provides negative information on average (as this is not
possible); rather it is that not all of the information provided by $S_2$ (i.e.\ the specificity) is
``useful''~\citep[p.~21]{shannon1998}.  This is in contrast to $S_1$ which only provides useful
specificity.  To summarise, it is the unique ambiguity which distinguishes the information provided
by variable $S_2$ from $S_1$, and hence why $S_2$ is deemed to provide negative average unique
information.  This form of uniqueness can only be distinguished by allowing the average unique
information to be negative.  This of course, requires abandoning the local positivity as a required
property, as per Theorem~\ref{thm:nonnegativity_avg}.  Few of the existing measures in the PID
literature consider dropping this requirement as negative information quantities are typically
regarded as being ``unfortunate''~\citep[p.\ 49]{cover2012}.  However, in the context of the
pointwise mutual information, negative information values are readily interpretable as being
misinformative values.  Despite this, the average information from each predictor must be
non-negative; however, it may be that what distinguishes one predictor from another are precisely
the misinformative predictor events, meaning that the unique information is in actual fact, unique
misinformation.  Forgoing local positivity makes the PPID using specificity and ambiguity novel (the
other exception in this regard is \citet{ince2017} who was first to consider allowing negative
average unique information.)

\subsection{Probability Distribution \tbc{}}
\label{sec:tbc}

\fig{tbc} shows the probability distribution \emph{two-bit-copy} (\tbc{}) which considers two
independently distributed binary predictor variables $S_1$ and $S_2$, and a target variable $T$
consisting of a separate elementary event for each joint event $S_{1,2}$.  There are several
important points to note about the decomposition of \tbc{}.  Firstly, due to the symmetry in the
probability distribution, each realisation will have the same pointwise decomposition.  Secondly,
due to the construction of the target, there is an isomorphism\footnote{Again, isomorphism should be
  taken to mean isomorphic probability spaces, e.g.\ \citep[p.~27]{gray1988} or
  \citep[p.4]{martin1984}.}  between $P(T)$ and $P(S_1, S_2)$, and hence the pointwise ambiguity
provided by any (individual or joint) predictor event is $0\bit$ (since given $t$, one knows $s_1$
and $s_2$).  Thirdly, the individual predictor events $s_1$ and $s_2$ each exclude $\nf{1}{2}$ of
the total probability mass in $P(T)$ and so each provide $1\bit$ of pointwise specificity; thus, by
\eq{red_specificity}, there is $1\bit$ of redundant pointwise specificity in each realisation.
Fourthly, the joint predictor event $s_{1,2}$ excludes $\nf{3}{4}$ of the total probability mass,
providing $2\bit$ of pointwise specificity; hence, by \eq{local_decomp_bivar_separate_spec}, each
joint realisation provides $1\bit$ of pointwise complementary specificity in addition to the $1\bit$
of redundant pointwise specificity.  Finally, putting this together via \eq{recomp}, \tbc{} consists
of $1\bit$ of average redundant information and $1\bit$ of average complementary information.

Although ``surprising''~\citep[p.~268]{bertschinger2013}, according to the operational
interpretation adopted in \secRef{op_interp_red}, two independently distributed predictor variables
can share redundant information.  That is, since the exclusions induced by $s_1$ and $s_2$ are the
same with respect to the two-event partition $\mc{T}^t$, the information associated with these
exclusions is regarded as being the same.  Indeed, this probability distribution highlights the
significance of specific reference to the two-event partition in \secRef{op_interp_red} and
Axiom~\ref{ax:local_comp}.  (This can be seen in the probability mass diagram in \fig{tbc}, where
the events $S_1=0$ and $S_2=0$ exclude different elementary target events within the complementary
event $\ob{0}$ and yet are considered to be the same exclusion with respect to the two-event
partition $\mc{T}^{0}$.)  That these exclusions should be regarded as being the same is discussed
further in \appRef{operational}.  Now however, there is a need to discuss \tbc{} in terms of
Theorem~\ref{thm:tar_cr} (Target Chain Rule).

\tbc{} was first considered as a ``mechanism''~\citep[p.~3]{harder2013} where ``the wires don't even
touch''~\citep[p.~167]{griffith2012}, which merely copies or concatenates $S_1$ and $S_2$ into a
composite target variable ${T_{1,2}=(T_1,T_2)}$ where ${T_1=S_1}$ and ${T_2=S_2}$.  However, using
causal mechanisms as a guiding intuition is dubious since different mechanisms can yield isomorphic
probability distributions~\citep[and references therein]{pearl1988}.  In particular, consider two
mechanisms which generate the composite target variables ${T_{1,3}=(T_1,T_3)}$ and
${T_{2,3}=(T_2,T_3)}$ where ${T_3=S_1 \mathxor S_2}$.  As can be seen in \fig{tbc}, both of these
mechanisms generate the same (isomorphic) probability distribution $P(S_1,S_2,T)$ as the mechanism
generating $T_{1,2}$.  If an information decomposition is to depend only on the probability
distribution $P(S_1,S_2,T)$, and no other semantic details such as labelling, then all three
mechanisms must yield the same information decomposition---this is not clear from the mechanistic
intuition.

\begin{figure}[t]
 \centering %% \tablesize{} %% You can specify the fontsize here, e.g.
  \input{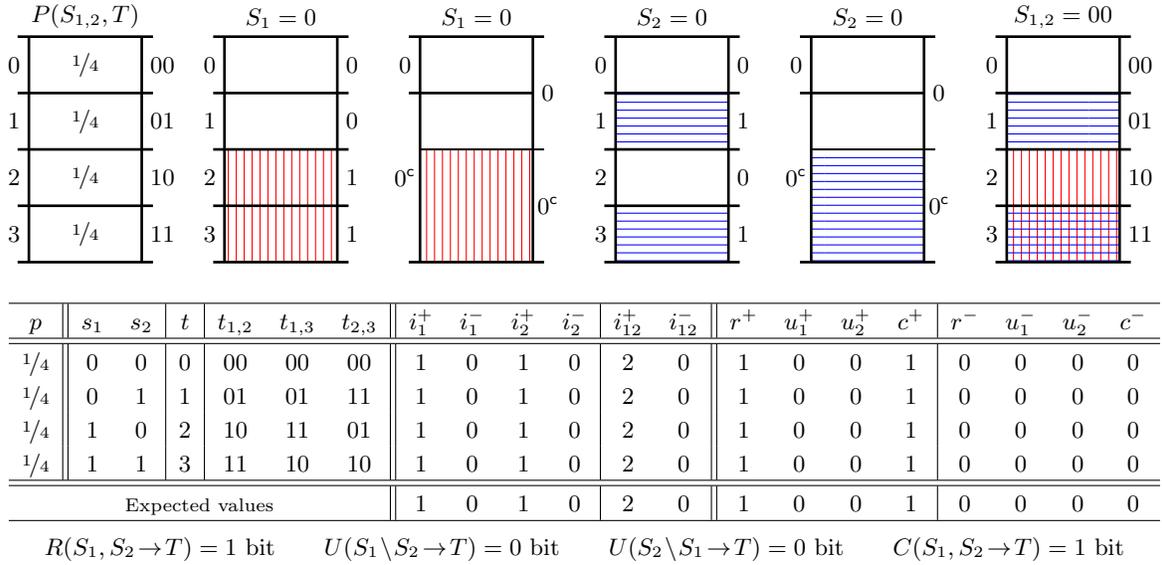} \vskip 1.5em \ifarXiv
  % blank
  \else
  \tablesize{\footnotesize}
  \fi
  \setlength{\tabcolsep}{5pt}
  \begin{tabular}{c || c c | c | c c c || c c c c | c c || c c c c | c c c c }
    \hline
    $p$ & $s_1$ & $s_2$ & $t$ & $t_{1,2}$ & $t_{1,3}$ & $t_{2,3}$
    & $i_1^+$ & $i_1^-$ & $i_2^+$ & $i_2^-$ & $i_{12}^+$ & $i_{12}^-$
      & $r^+$ & $u_1^+$ & $u_2^+$ & $c^+$ & $r^-$ & $u_1^-$ & $u_2^-$ & $c^-$\\
    \hline\hline
    \nf{1}{4} & 0 & 0 & 0 & 00 & 00 & 00 & 1 &0 &1 &0 & 2 & 0 & 1 &0 &0 &1 &0 &0 &0 &0 \\
    \nf{1}{4} & 0 & 1 & 1 & 01 & 01 & 11 & 1 &0 &1 &0 & 2 & 0 & 1 &0 &0 &1 &0 &0 &0 &0 \\
    \nf{1}{4} & 1 & 0 & 2 & 10 & 11 & 01 & 1 &0 &1 &0 & 2 & 0 & 1 &0 &0 &1 &0 &0 &0 &0 \\
    \nf{1}{4} & 1 & 1 & 3 & 11 & 10 & 10 & 1 &0 &1 &0 & 2 & 0 & 1 &0 &0 &1 &0 &0 &0 &0 \\
    \hline\hline
    \multicolumn{7}{c||}{\scriptsize Expected values}
      & 1 & 0 & 1 & 0 & 2 & 0 & 1 & 0 & 0 & 1 & 0 & 0 & 0 & 0 \\
    \hline
  \end{tabular}
  \vskip 5pt{\small $\RI=1\bit \qquad \UIo=0\bit \qquad \UIt=0\bit \qquad \CI=1\bit$}
  \caption{\label{fig:tbc} Example \tbc{}.  \emph{Top}:\ the probability mass diagrams for the
    realisation ${(S_1\!=\!0,S_2\!=\!0,T\!=\!00)}$.  \emph{Middle}:\ for each realisation, the PPID
    using specificity and ambiguity is evaluated (see \fig{xor}). \emph{Bottom}:\ the
    decomposition of \xor{} yields the same result as $\imin$.}
\end{figure}

Although the decomposition of the various composite target variables must be the same, there is no
requirement that the three systems must yield the same decomposition when analysed in terms of the
individual components of the composite target variables.  Nonetheless, there ought to be a
consistency between the decomposition of the composite target variables and the decomposition of the
component target variables---i.e.\ there should be a target chain rule.  As shown in
Theorem~\ref{thm:tar_cr}, the measures $\rmin$ and $\Rmin$ satisfy the target chain rule, whereas
$\imin$, $\uitilde$, $\ired$ and $\ivk$ do not~\citep{bertschinger2013, griffith2014}.  Failing to
satisfy the target chain rule can lead to inconsistencies between the composite and component
decompositions, depending on the order in which one considers decomposing the information (this is
discussed further in \appRef{accum_betting}).  In particular, \tb{tbc_tar_cr} shows how $\uitilde$,
$\ired$ and $\ivk$ all provide the same inconsistent decomposition for \tbc{} when considered in
terms of the composite target variable $T_{1,3}$.  In contrast, $\Rmin$ produces a consistent
decomposition of $T_{1,3}$.  Finally, based on the above isomorphism, consider the following (the
proof is deferred to \appRef{tcr}).

\begin{restatable}{restatetheorem}{theoremnopos}
  \label{thm:no_pos}
  The target chain rule, identity property and local positivity, cannot be simultaneously satisfied.
\end{restatable}

\begin{table}[b]
  \centering %% \tablesize{} %% You can specify the fontsize here, e.g.
  \ifarXiv
  % blank
  \else
  \tablesize{\footnotesize}
  \fi
  \setlength{\tabcolsep}{3pt}
  \bgroup
  \def\arraystretch{1.3}%  1 is the default, change whatever you need
  \begin{tabular}{c||c || c | c || c | c }
    \hline
    &$I(S_{1,2};T_{1,3})$
      & $I(S_{1,2};T_1)$ & $I(S_{1,2};T_3|T_1)$
      & $I(S_{1,2};T_3)$ & $I(S_{1,2};T_1|T_3)$ \\
    \hline\hline
    \multirow{2}{*}{
      $\begin{aligned}
        \raisebox{-3pt}{$\uitilde,\; \ired,$} \\
        \raisebox{1pt}{$\ivk$} \quad
      \end{aligned}$ 
    }
      & \multirow{2}{*}{
          $\begin{aligned}
            U(S_1 \sm S_2 \ra T_{1,3})&\!=\!1 \\
            U(S_2 \sm S_1 \ra T_{1,3})&\!=\!1 
          \end{aligned}$ 
        }
      & \multirow{2}{*}{$U(S_1 \sm S_2 \ra T_1)\!=\!1$}
      & \multirow{2}{*}{$U(S_2 \sm S_1 \ra T_3 | T_1)\!=\!1$}
      & \multirow{2}{*}{$C(S_1, S_2 \ra T_3)\!=\!1$}
      & \multirow{2}{*}{$R(S_1, S_2 \ra T_1 | T_3)\!=\!1$}  \\
    & & & & & \\
    \hline
    \multirow{3}{*}{$\Rmin$}
      & \multirow{3}{*}{
          $\begin{aligned}
            R(S_1,  S_2 \ra T_{1,3})&=1 \\
            C(S_1,  S_2 \ra T_{1,3})&=1
          \end{aligned}$
        }
      & \multirow{3}{*}{
          $\begin{aligned}
            \! U(S_2 \sm S_1 \ra T_1)& \!=\! -1 \! \\
            R(S_1,  S_2 \ra T_1)&\!=\!1 \\
            C(S_1, S_2 \ra T_1)&\!=\!1
          \end{aligned}$ 
        } &  &  &  \\
    & & & $U(S_2 \sm S_1 \ra T_3|T_1)\!=\!1$ & $C(S_1, S_2 \ra T_3)\!=\!1$
                         & $R(S_1, S_2 \ra T_1|T_3)\!=\!1$  \\
    & & & & & \\
    \hline
  \end{tabular}
  \egroup
  \caption{ \label{tb:tbc_tar_cr} Shows the decomposition of the quantities in the first row induced
    by the measures in the first column.  For consistency, the decomposition of $I(S_{1,2};T_{1,3})$
    should equal both the sum of the decomposition of $I(S_{1,2};T_1)$ and $I(S_{1,2};T_3|T_1)$, and
    the sum of the decomposition of $I(S_{1,2};T_3)$ and $I(S_{1,2};T_1|3)$.  Note that the
    decomposition induced by $\uitilde$, $\ired$ and $\ivk$ are not consistent.  In contrast,
    $\Rmin$ is consistent due to Theorem~\ref{thm:tar_cr}.  }

\end{table}

\ifarXiv
  \pagebreak
\else
  % blank
\fi

\subsection{Summary of Key Properties}
\label{sec:properties}

The following are the key properties of the PPID using the specificity and ambiguity.
Property~\ref{prop:zero_unq} follows directly from the Definitions~\ref{def:specificity} and
\ref{def:ambiguity}.  Property~\ref{prop:no_local_pos} follows from Theorems~\ref{thm:nonnegativity}
and \ref{thm:nonnegativity_avg}.  Property~\ref{prop:no_identity} follows from the probability
distribution \tbc{} in \secRef{tbc}.  Property~\ref{prop:no_tar_mono} was discussed in
\secRef{measure}.  Property~\ref{prop:target_chain_rule} is proved in Theorem~\ref{thm:tar_cr}.

\begin{Property}
  \label{prop:zero_unq}
  When considering the redundancy between the source events $\mb{a}_1,\ldots,\mb{a}_k$, at least one
  source event $\mb{a}_i$ will provide zero unique specificity, and at least one source event
  $\mb{a}_j$ will provide zero unique ambiguity.  The events $\mb{a}_i$ and $\mb{a}_j$ are not
  necessarily the same source event.
\end{Property}

\begin{Property}
  \label{prop:no_local_pos}
  The atoms of partial specificity and partial ambiguity satisfy local positivity, $\pi^\pm \geq 0$.
  However, upon recombination and averaging, the atoms of partial information do not satisfy local
  positivity, $\Pi \geq 0$.
\end{Property}

\begin{Property}
  \label{prop:no_identity}
  The decomposition does not satisfy the identity property.
\end{Property}

\begin{Property}
  \label{prop:no_tar_mono}
  The decomposition does not satisfy the target monotonicity property.
\end{Property}

\begin{Property}
  \label{prop:target_chain_rule}
  The decomposition satisfies the target chain rule.
\end{Property}

%%%%%%%%%%%%%%%%%%%%%%%%%%%%%%%%%%%%%%%%%%%%%%%%%%%%%%%%%%%%%%%%%%%%%%%%%%%%%%%%%%%%%%%%%%%%%%%%%%%%
%%%%%%%%%%%%%%%%%%%%%%%%%%%%%%%%%%%%%%%%%%%%%%%%%%%%%%%%%%%%%%%%%%%%%%%%%%%%%%%%%%%%%%%%%%%%%%%%%%%%

\section{Conclusion}
\label{sec:conclusion}

The partial information decomposition of \citet{williams2010, williams2010priv} provided an
intriguing framework for the decomposition of multivariate information.  However, it was not long
before ``serious flaws''~\citep[p.\ 2163]{bertschinger2014} were identified.  Firstly, the measure
of redundant information $\imin$ failed to distinguish between whether predictor variables provide
the same information or merely the same amount of information.  Secondly, $\imin$ fails to satisfy
the target chain rule, despite this addativity being one of the defining characteristics of
information.  Notwithstanding the problems, the axiomatic derivation of the redundancy lattice was
too elegant to be abandoned and hence several alternate measures were proposed, i.e.\ $\ired$,
$\uitilde$ and $\ivk$ \citep{harder2013,bertschinger2014,griffith2012}.  Nevertheless, as these
measures all satisfy the identity property, they cannot produce a non-negative decomposition for an
arbitrary number of variables~\citep{rauh2014}.  Furthermore, none of these measures satisfy the
target chain rule meaning they produce inconsistent decompositions for multiple target variables.
Finally, in spite of satisfying the identity property (which many consider to be desirable), these
measures still fail to identify when variables provide the same information, as exemplified by the
pointwise unique problem presented in \secRef{localisability}.

This paper took the axiomatic derivation of the redundancy lattice from PID and applied it to the
unsigned entropic components of the pointwise mutual information.  This yielded two separate
redundancy lattices---the specificity and the ambiguity lattices. Then based upon an operational
interpretation of redundancy, measures of pointwise redundant specificity $\rminp$ and pointwise
redundant ambiguity $\rminn$ were defined.  Together with specificity and ambiguity lattices, these
measures were used to decompose multivariate information for an arbitrary number of variables.
Crucially, upon recombination, the measure $\rmin$ satisfies the target chain rule.  Furthermore,
when applied to \locunq{}, these measures do not result in the pointwise unique problem.  In our
opinion, this demonstrates that the decomposition is indeed correctly identifying redundant
information.  However, others will likely disagree with this point given that the measure of
redundancy does not satisfy the identity property.  According to the identity property, independent
variables can never provide the same information.  In contrast, according to the operational
interpretation adopted in this paper, independent variables can provide the same information if they
happen to provide the same exclusions with respect to the two-event target distribution.  In any
case, the proof of Theorem~\ref{thm:no_pos} and the subsequent discussion in \appRef{tcr},
highlights the difficulties that the identity property introduces when considering the information
provided about events in separate target variables. (See further discussion in
\appRef{accum_betting}).

Our future work with this decomposition will be both theoretical and empirical.  Regarding future
theoretical work, given that the aim of information decomposition is to derive measures pertaining
to sets of random variables, it would be worthwhile to derive the information decomposition from
first principles in terms of measure theory.  Indeed, such an approach would surely eliminate the
semantic arguments (about what it means for information to unique, redundant or complementary),
which currently plague the problem domain.  Furthermore, this would certainly be a worthwhile
exercise before attempting to generalise the information decomposition to continuous random
variables.  Regarding future empirical work, there are many rich data sets which could be decomposed
using this decomposition including financial time-series and neural recordings,
e.g.~\citep{wibral2017,tax2017,kay2017}.

\ifarXiv

\appendix

\else

%%%%%%%%%%%%%%%%%%%%%%%%%%%%%%%%%%%%%%%%%%%%%%%%%%%%%%%%%%%%%%%%%%%%%%%%%%%%%%%%%%%%%%%%%%%%%%%%%%%%
%%%%%%%%%%%%%%%%%%%%%%%%%%%%%%%%%%%%%%%%%%%%%%%%%%%%%%%%%%%%%%%%%%%%%%%%%%%%%%%%%%%%%%%%%%%%%%%%%%%%

\appendixtitles{yes} %Leave argument "no" if all appendix headings stay EMPTY (then no dot is
                     %printed
                    %after "Appendix A"). If the appendix sections contain a heading then change the
                    %argument to "yes".
\appendixsections{multiple} %Leave argument "multiple" if there are multiple sections. Then a
                            %counter is printed ("Appendix A"). If there is only one appendix
                            %section then change the argument to "one" and no counter is printed
                            %("Appendix").

\appendix

\fi

%%%%%%%%%%%%%%%%%%%%%%%%%%%%%%%%%%%%%%%%%%%%%%%%%%%%%%%%%%%%%%%%%%%%%%%%%%%%%%%%%%%%%%%%%%%%%%%%%%%%
%%%%%%%%%%%%%%%%%%%%%%%%%%%%%%%%%%%%%%%%%%%%%%%%%%%%%%%%%%%%%%%%%%%%%%%%%%%%%%%%%%%%%%%%%%%%%%%%%%%%

\section{Kelly Gambling, Axiom~\ref{ax:local_comp}, and \tbc{}}
\label{app:operational}

In \secRef{op_interp_red}, it was argued that the information provided by a set of predictor events
$s_1,\ldots,s_k$ about a target event~$t$ is the same information if each source event induces the
same exclusions with respect to the two-event partition ${\mc{T}^t=\{t,\ob{t}\}}$.  This was based
on the fact that pointwise mutual information does not depend on the apportionment of the exclusions
across the set of events which did not occur $\ob{t}$.  It was argued that since the pointwise
mutual information is independent of these differences, the redundant mutual information should also
be independent of these differences.  This requirement was then integrated into the operational
interpretation of \secRef{op_interp_red} and was later enshrined in the form of
Axiom~\ref{ax:local_comp}.  This appendix aims to justify this operational interpretation and argue
why redundant information in \tbc{} is not ``unreasonably large''~\citep[p.~269]{bertschinger2013}.

\subsection{Pointwise Side Information and the Kelly Criterion}
\label{sec:kelly_betting_primer}

Consider a set of horses $\mc{T}$ running in a race which can be considered a random variable $T$
with distribution $P(T)$.  Say that for each $t \in \mc{T}$ a bookmaker offers odds of
$o(t)$-for-$1$, i.e. the bookmaker will pay out $o(t)$ dollars on a \$1 bet if the horse $t$
wins.  Furthermore, say that there is no track take as ${\sum_{t \in \mc{T}} \nf{1}{o(t)} =1}$, and
these odds are \emph{fair}, i.e.\ ${o(t)=\nf{1}{p(t)}}$ for all ${t\in\mc{T}}$
\citep{kelly1956}. Let $\mb{b}(T)$ be the fraction of a gambler's capital bet on each horse
${t \in \mc{T}}$ and assume that the gambler stakes all of their capital on the race, i.e.\
${\sum_{t \in \mc{T}} b(t) =1}$.

Now consider an i.i.d.\ series of these races $T_1,T_2,\ldots$ such that $P(T_k)=P(T)$ for all
$k\in\mathbb{N}$ and let $t_k\in\mc{T}$ represent the winner of the $k$-th race.  Say that the
bookmaker offers the same odds on each race and the gambler bets their entire capital on each race.
The gambler's capital after $m$ races $D_m$ is a random variable which depends on two factors per
race:\ the amount the gambler staked on each race winner $t_k$, and the odds offered on each winner
$t_k$.  That is,
\begin{equation}
  \label{eq:kelly_wealth}
  D_m = \prod_{k=1}^m b(t_k) \, o(t_k),
\end{equation}
where monetary units \$ have been chosen such that $D_0=\$1$.  The gambler's wealth grows (or
shrinks) exponentially, i.e.
\begin{equation}
  D_m = 2^{m \, W(\mb{b},T)}
\end{equation}
where
\begin{equation}
  W(\mb{b},T) = \frac{1}{m} \log D_m
    = \frac{1}{m} \sum_{k=1}^m \log b(t_k) \, o(t_k)
    = \mathrm{E} \big[ \log b(t_k) \, o(t_k) \big]
\end{equation}
is the doubling rate of the gambler's wealth using a betting strategy $\mb{b}(T)$.  Here, the last
equality is by the weak law of large numbers for large $m$.

Any reasonable gambler would aim to use an optimal strategy $\mb{b}^*(T)$ which maximises the
doubling rate $W(\mb{b},T)$.  Kelly~\citep{kelly1956,cover2012} proved that the optimal doubling
rate is given by
\begin{equation}
  W^*(T) = \max_{\mb{b}} W(\mb{b},T)
    = \mathrm{E} \big[ \log b^*(t_k) \, o(t_k) \big]
\end{equation}
and is achieved by using the proportional gambling scheme ${\mb{b}^*(T)=P(T)}$.  When the race $T_k$
occurs and the horse $t_k$ wins, the gambler will receive a payout of $b^*(t^k)\,o(t^k)=\$1$, i.e.\
the gambler receives their stake back regardless of the outcome.  In the face of fair odds, the
proportional Kelly betting scheme is the optimal strategy---non-terminating repeated betting with
any other strategy will result in losses.

% (with repeated bets, on average over the ensemble).
%% Changed the wording, although perhaps long run should be "sufficiently large number of bets". 

Now consider a gambler with access to a private wire $S$ which provides (potentially useful) side
information about the upcoming race.  Say that these messages are selected from the set $\mc{S}$,
and that the gambler receives the message $s_k$ before the race $T_k$.
Kelly~\citep{kelly1956,cover2012} showed that the optimal doubling rate in the presence of this side
information is given by
\begin{equation}
  W^*(T|S) = \max_{\mb{b}} W(\mb{b},T|S)
    = \mathrm{E} \big[ \log b^*(t_k|s_k) \, o(t_k) \big],
\end{equation}
and is achieved by using the conditional proportional gambling scheme $\mb{b}^*(T|s_k)=P(T|s_k)$.
Both the proportional gambling scheme $b^*(T)$ and the conditional proportional gambling scheme
$b^*(T|S)$ are based upon the \emph{Kelly criterion} whereby bets are apportioned according to the
best estimation of the outcome available.  The financial value of the private wire to a gambler can
be ascertained by comparing their doubling rate of the gambler with access to the side wire to that
of a gambler with no side information, i.e.
\begin{align}
  \Delta W = W^*(T|S) - W^*(T) 
    &= \mathrm{E} \big[ \log b^*(t_k|s_k) \, o(t_k) \big]
      - \mathrm{E} \big[ \log b^*(t_k) \, o(t_k) \big]     \nn\\
    &= \mathrm{E} \big[ i(s_k;t_k) \big] = I(S;T).       \label{eq:expect_loc}
\end{align}
This important result due to \citet{kelly1956} equates the increase in the doubling rate $\Delta W$
due to the presence of side information, with the mutual information between the private wire $S$
and the horse race $T$.  If on average, the gambler receives $1\bit$ of information from their
private wire, then on average the gambler can expect to double their money per race.  Furthermore,
as one would expect, independent side information does not increase the doubling rate.

With no side information, the Kelly gambler always received their original stake back from the
bookmaker.  However, this is not true for the Kelly gambler with side information.  Although their
doubling rate is greater than or equal to that of the gambler with no side information, this is only
true \emph{on average}.  Before the race $T_k$, the gambler receives the private wire message $s_k$
and then, the horse $t_k$ wins the race.  From \eq{expect_loc}, one can see that the return
$\Delta w_k$ for the $k$-th race is given by the pointwise mutual information,
\begin{equation}
  \label{eq:loc_return}
  \Delta w = i(s_k;t_k).
\end{equation}
Hence, just like the pointwise mutual information, the per race return can be positive or negative:\
if it is positive, the gambler will make a profit; if it is negative, the gambler will sustain a
loss.  Despite the potential for pointwise loses, the average return (i.e.\ the doubling rate) is,
just like the average mutual information, non-negative---and indeed, is optimal.  Furthermore, while
a Kelly gambler with side information can lose money on any single race, they can never actually go
bust.  The Kelly gambler with side information $s$ still hedges their risk by placing bets on all
horses with a non-zero probability of winning according to their side information, i.e.\ according
to $P(T|s_k)$. The only reason they would fail to place a bet on a horse is if their side
information completely precludes any possibility of that horse winning.  That is, a Kelly gambler
with side information will never fall foul of gambler's ruin.

\subsection{Justification of Axiom~\ref{ax:local_comp} and Redundant Information in \tbc{}}
\label{sec:op_just}

Consider \tbc{} semantically described in terms of a horse race.  That is, consider a four horse
race $T$ where each horse has an equiprobable chance of winning, and consider the binary variables
$T_1$, $T_2$, and $T_3$ which represent the following, respectively:\ the colour of the horse, black
$0$ or white $1$; the sex of the jockey, female $0$ or male $1$; and the colour of the jockey's
jersey, red $0$ or green $1$.  Say that the four horses have the following attributes:
\begin{description}\setlength{\itemsep}{0pt}\setlength{\parskip}{0pt}\setlength{\itemindent}{-15pt}
  \item[Horse 0] is a black horse $T_1\!=\!0$, ridden by a female jockey $T_2\!=\!0$, who is wearing
    a red jersey $T_3\!=\!0$. %\parfillskip=0pt
  \item[Horse 1] is a black horse $T_1\!=\!0$, ridden by a male jockey $T_2\!=\!1$, who is wearing a
    green jersey $T_3\!=\!1$. %\parfillskip=0pt
  \item[Horse 2] is a white horse $T_1\!=\!1$, ridden by a female jockey $T_2\!=\!0$, who is wearing
    a green jersey $T_3\!=\!1$. %\parfillskip=0pt
  \item[Horse 3] is a white horse $T_1\!=\!1$, ridden by a male jockey $T_2\!=\!1$, who is wearing a
    red jersey $T_3\!=\!0$. %\parfillskip=0pt
\end{description}
There are two important points to note. Firstly, the horses in the race $T$ could also be uniquely
described in terms of the composite binary variables $T_{1,2}$, $T_{1,3}$ or $T_{2,3}$.  Secondly,
if one knows $T_1$ and $T_2$ then one knows $T_3$ (which can be represented by the relationship
${T_3=T_1 \mathxor T_2}$).  Finally, consider private wires $S_1$ and $S_2$ which
\emph{independently} provide the colour of the horse and the colour of the jockey's jersey
(respectively) before the upcoming race, i.e.\ $S_1=T_1$ and $S_2=T_2$.

Now say a bookmaker offers fair odds of 4-for-1 on each horse in the race $T$.  Consider two
gamblers who each have access to one of $S_1$ and $S_2$.  Before each race, the two gamblers receive
their respective private wire messages and place their bets according to the Kelly strategy.  This
means that each gambler lays half of their, say $\$1$, stake on each of their two respective
non-excluded horses:\ unknowingly, both of the gamblers have placed a bet on the soon-to-be race
winner, and each gambler has placed a distinct bet on one of the two soon-to-be losers.  The only
horse neither has bet upon is also a soon-to-be loser. (See \citep{bertschinger2013} for a related
description of \tbc{} in term of the game-theoretic notions of shared and common knowledge.) After
the race, the bookmaker pays out $\$2$ dollars to each gambler:\ both have doubled their money.
This happens because both of the gamblers had one bit of $1\bit$ of information about the race,
i.e.\ pointwise mutual information.  In particular, both gamblers improved their probability of
predicting the eventual race winner.  It did not matter, in any way, that the gamblers had each laid
distinct bets on one of the three eventual race losers.  The fact that they laid different bets on
the horses which did not win, made no difference to their winnings.  The apportionment of the
exclusions across the set of events which did not occur, makes no difference to the pointwise mutual
information.  With respect to what occurred (i.e. with respect to which horse won), the fact the
that they excluded different losers is only semantic.  When it came to predicting the
would-be-winner, both gamblers had the same predictive power; they both had the same freedom of
choice with regards to selecting what would turn out to be the eventual race winner---they had the
same information.  It is for this reason that this information should be regarded as redundant
information, regardless of the independence of the information sources.  Hence, the introduction of
both the operational interpretation of redundancy in \secRef{op_interp_red} and
Axiom~\ref{ax:local_comp} in \secRef{measure}.

Now consider a third gambler who has access to both private wires $S_1$ and $S_2$, i.e.\ $S_{1,2}$.
Before the race, this gambler receives both private wire messages which, in total, precludes three
of the horses from winning.  This gambler then places the entirety of their $\$1$ stake on the
remaining horse which is sure to win.  After the race, the bookmaker pays out $\$4$:\ this gambler
has quadrupled their money as they had $2\bit$ of pointwise mutual information about the race.
Having both private wire messages simultaneously gave this gambler a $1\bit$ informational edge over
the two gamblers with access to a single side wire.  While each of the singleton gamblers had
$1\bit$ of independent information, the only way one could profit from the independence of this
information is by having both pieces of information simultaneously---this makes this $1\bit$ of
information complementary.  Although this may seem ``palpably
strange''~\citep[p.~167]{griffith2012}, it is not so strange when from the following perspective:\
the only way to exploit two pieces of independent information is by having both pieces together
simultaneously.

\subsection{Accumulator Betting and the Target Chain Rule}
\label{app:accum_betting}

Say that in addition to the 4-for-1 odds offered on the race $T$, the bookmaker also offers fair
odds of 2-for-1 on each of the binary variables $T_1$, $T_2$ and $T_3$.  Now, in addition to being
able to directly gamble on the race $T$, one could indirectly gamble on $T$ by placing a so-called
\emph{accumulator} bet on any pair of $T_1$, $T_2$ and $T_3$.  An accumulator is a series of chained
bets whereby any return from one bet is automatically staked on the next bet; if any bet in the
chain is lost then the entire chain is lost.  For example, a gambler could place 4-for-1 bet on
horse $0$ by placing the following accumulator bet:\ a 2-for1 bet on a black horse winning that
chains into a 2-for-1 bet on the winning jockey being female (or equivalently, vice versa).  In
effect, these accumulators enable a gambler to bet on $T$ by instead placing a chained bet on the
independent component variables within the (equivalent) joint variables $T_{1,2}$, $T_{1,3}$ and
$T_{2,3}$.  Now consider again the three gamblers from the prior section, i.e.\ the two gamblers who
each have a private wire $S_1$ and $S_2$, and the third gamble who has access to $S_{1,2}$.  Say
that they must each place a, say $\$1$, accumulator bet on $T_{1,3}$---what should each gambler do
according to the Kelly criterion?

For the sake of clarity, consider only the realisation where the horse $T=0$ subsequently wins (due
to the symmetry, the analysis is equivalent for all realisations).  First consider the accumulator
whereby the gamblers first bet on the colour of the winning horse $T_1$, which chains into a bet on
the colour of the winning jockey's jersey $T_3$.  Suppose that the private wire $S_1$ communicates
that the winning horse will be black, while the private wire $S_2$ communicates that the winning
horse will be ridden by a female jockey, i.e.\ $S_1=0$ and $S_2=0$.  Following to the Kelly
strategy, the gambler with access to $S_1=0$ takes out two $\$0.5$ accumulator bets.  Both of these
accumulators feature the same initial bet on the winning horse being black since $T_1=S_1=0$.  Hence
both bets return $\$1$ each which become the stake on the next bet in each accumulator.  This
gambler knows nothing about the colour of the jockey's jersey $T_3$.  As such, one accumulator
chains into a bet on the winning jersey being red $T_3=0$, while the other chains into a bet on it
being green $T_3=1$.  When the horse $T=0$ wins, the stake bet on the green jersey is lost while bet
on red jersey pays out $\$2$.  This gambler had $1\bit$ of side information and so doubled their
money.  Now consider the gambler with private wire $S_2$, who knows nothing about $T_1$ or $T_3$
individually.  Nonetheless, this gambler knows that the winner must be a female jockey $T_2=0$. As
such, this gambler knows that if a black horse $T_1=0$ wins then its jockey must be wearing a red
jersey $T_3=0$, or if a white horse $T_1=0$ wins then its jockey must be wearing a green jersey
$T_3=1$ (since ${T_3 = T_1 \mathxor T_2}$).  Thus this gambler can also utilise the Kelly strategy
to place the following two $\$0.5$ accumulator bets:\ the first accumulator bets on the winning
horse being black $T_1=0$ and then chains into a bet on the winner's jersey being red $T_3=0$, while
the second accumulator bets on the winning horse being white $T_1=1$ and then chains into a bet on
the winner's jersey being green $T_3=1$.  When the horse $T=0$ wins, the first accumulator pays out
$\$2$, while the second accumulator is be lost.  Hence, this gambler also doubles their money and so
also had $1\bit$ of side information.  Finally, consider the gambler with access to both private
wires $S_{1,3}$, who can place an accumulator on the black horse $T_1=0$ winning chaining into a bet
on the winning jockey wearing red $T_3=0$.  This gambler can quadruple their stake, and so must
possess $2\bit$ of side information.

Each of the three gamblers have the same final return regardless of whether the gamblers are betting
on the variable $T$, or placing accumulator bets on the variables $T_{1,2}$, $T_{1,3}$ or $T_{2,3}$.
However, the paths to the final result differs between the gamblers, reflecting the difference
between the information the each gambler had about the sub-variables $T_1$, $T_2$ or $T_3$.  Given
the result of \citet{kelly1956}, the proposed information decomposition should reflect these
differences, but yet still arrive at the same result---in other words, the information decomposition
should satisfy a target chain rule.  This is clear if the Kelly interpretation of information is to
remain as a ``duality''\citep[p.~159]{cover2012} in information theory.

%%%%%%%%%%%%%%%%%%%%%%%%%%%%%%%%%%%%%%%%%%%%%%%%%%%%%%%%%%%%%%%%%%%%%%%%%%%%%%%%%%%%%%%%%%%%%%%%%%%%
%%%%%%%%%%%%%%%%%%%%%%%%%%%%%%%%%%%%%%%%%%%%%%%%%%%%%%%%%%%%%%%%%%%%%%%%%%%%%%%%%%%%%%%%%%%%%%%%%%%%

\section{Supporting Proofs and Further Details}
\label{app:tech}

This appendix contains many of the important theorems and proofs relating to PPID using specificity
and ambiguity.

\subsection{Deriving the Specificity and Ambiguity Lattices from
  Axioms~\ref{ax:symmetry}--\ref{ax:local_comp}}
\label{app:lattices}

The following section is based directly on the original work of Williams and
Beer~\citep{williams2010, williams2010priv}.  The difference is that we now consider sources events
$\mb{a}_i$ rather than sources $\mb{A}_i$.

\begin{Proposition}
  \label{corol:nonneg}
  Both $\icp$ and $\icn$ are non-negative.
\end{Proposition}

\begin{proof}
  Since $\emptyset \subseteq \mb{a}_i$ for any $\mb{a}_i$, Axioms~\ref{ax:mono} and \ref{ax:sr} 
  imply
  \begin{alignat}{5}
    &\icp\big( \mb{a}_1, \dots, \mb{a}_k \ra t \big)
      &&\geq \icp\big( \mb{a}_1, \dots, \mb{a}_k, \emptyset \ra t \big)
      &&= \icp\big(\emptyset \ra t \big) &&= h(\emptyset) &&= 0,  \\
    &\icn\big( \mb{a}_1, \dots, \mb{a}_k \ra t \big)
      &&\geq \icn\big( \mb{a}_1, \dots, \mb{a}_k, \emptyset \ra t \big)
      &&= \icn\big(\emptyset \ra t \big) &&= h(\emptyset | t) &&= 0. 
  \end{alignat}
  Hence, both $\icp\big( \mb{a}_1, \dots, \mb{a}_k \ra t \big)$ and
  $\icn\big( \mb{a}_1, \dots, \mb{a}_k \ra t \big)$ are non-negative.
\end{proof}
\begin{Proposition}
  Both $\icp$ and $\icn$ are bounded from above by the specificity and the ambiguity from any single
  source event, respectively. 
\end{Proposition}
\begin{proof}
  For any single source $\mb{a}_i$, Axioms~\ref{ax:mono} and \ref{ax:sr} yield
  \begin{alignat}{4}
    h(a_i)     &&= \icp\big( \mb{a}_i \ra t \big) &&= \icp\big( \mb{a}_i, \mb{a}_i \ra t \big)
      && \geq \icp\big( \mb{a}_i, \dots \ra t \big),                                            \\
    h(a_i | t) &&= \icn\big( \mb{a}_i \ra t \big) &&= \icn\big( \mb{a}_i, \mb{a}_i \ra t \big)
      && \geq \icp\big( \mb{a}_i, \dots \ra t \big),
  \end{alignat}
  as required.
\end{proof}

In keeping with Williams and Beer's approach \citep{williams2010, williams2010priv}, consider all of
the distinct ways in which a collection of source events ${\mb{a}=\{\mb{a}_1,\dots,\mb{a}_k\}}$
could contribute redundant information.  Thus far we have assumed that the redundancy measure can be
applied to any collection of source events, i.e.\ $\msc{P}_1(\mb{a})$ where $\msc{P}_1$ denotes the
power set with the empty set removed.  Recall that the sources events are themselves collections of
predictor events, i.e.\ $\msc{P}_1(\mb{s})$.  That is, we can apply both $\icp$ and $\icn$ to
elements of $\msc{P}_1\big(\msc{P}_1(\mb{s})\big)$.  However, this can be greatly reduced using
Axiom~\ref{ax:mono} which states that if $\mb{a}_i \subseteq \mb{a}_j$, then
\begin{align*}
  \icp\big( \mb{a}_j, \mb{a}_i, \ldots \ra t \big) &= \icp\big( \mb{a}_i, \ldots \ra t \big),  \\
  \icn\big( \mb{a}_j, \mb{a}_i, \ldots \ra t \big) &= \icn\big( \mb{a}_i, \ldots \ra t \big).
\end{align*}
Hence, one need only consider the collection of source events such that no source event is a
superset of any other in order,
\begin{equation}
  \label{eq:nodes}
  \msc{A}(\mb{s}) = \Big\{ \alpha \in \msc{P}_1\big(\msc{P}_1(\mb{s})\big) \;
    \big| \; \forall \, \mb{a}_i, \, \mb{a}_j \in \alpha, \, \mb{a}_i \not\subset \mb{a}_j \Big\}.
\end{equation}
This collection $\msc{A}(\mb{s})$ captures all the distinct ways in the source events could provide
redundant information.

As per Williams and Beer's PID, this set of source events $\msc{A}(\mb{s})$ is structured. Consider
two sets of source events $\alpha,\beta \in \msc{A}(\mb{s})$.  If for every source event
$\mb{b} \in \beta$ there exists a source event $\mb{a} \in \alpha$ such that
$\mb{a} \subseteq \mb{b}$, then all of the redundant specificity and ambiguity shared by
$\mb{b} \in \beta$ must include any redundant specificity and ambiguity shared by
$\mb{a} \in \alpha$. Hence, a partial order $\preceq$ can be defined over the elements of the domain
$\msc{A}(\mb{s})$ such that any collection of predictors event coalitions precedes another if and
only if the latter provides any information the former provides,
\begin{equation}
  \label{eq:partial_order}
  \forall \alpha, \beta \in \msc{A}(\mb{s}),
    \big( \alpha \preceq \beta \iff \forall \, \mb{b} \in \beta,
    \, \exists \, \mb{a} \in \alpha \; | \; \mb{a} \subseteq \mb{b} \big).
\end{equation}
Applying this partial ordering to the elements of the domain $\msc{A}(\mb{s})$ produces a lattice
which has the same structure as the redundancy lattice from PID, i.e.\ the structure of the sources
events here is the same as the structure of the sources in PID.  (\fig{lattices} depicts this
structure for the case of 2 and 3 predictor variables.)  Applying $i_\cap^+$ to these sources events
yields a \emph{specificity lattice} while applying $i_\cap^-$ yields an \emph{ambiguity lattice}.

Similar to $I_\cap$ in PID, the redundancy measures $i_\cap^+$ or $i_\cap^-$ can be thought of as a
cumulative information functions which integrate the specificity or ambiguity uniquely contributed
by each node as one moves up each lattice.  In order in evaluate the unique contribution of
specificity and ambiguity from each node in the lattice, consider the M\"obius inverse
\citep{rota1964, stanley2012} of $i_\cap^+$ and $i_\cap^-$.  That is, the specificity and ambiguity
of a node $\alpha$ is given by
\begin{equation}
  \label{eq:icap}
  i_\cap^\pm(\alpha \ra t) = \sum_{\beta \preceq \alpha} i_\cap^\pm (\beta \ra t)
  \qquad  \forall \; \alpha, \, \beta \in \msc{A}(\mb{s}).
\end{equation}
Thus the unique contributions of \emph{partial specificity} $i_\partial^+$ and \emph{partial
  ambiguity} $i_\partial^-$ from each node can be calculated recursively from the bottom-up, i.e.
\begin{equation}
  \label{eq:ipartial}
  i_\partial^\pm(\alpha \ra t)
  = i_\cap^\pm(\alpha \ra t) - \sum_{\beta \prec \alpha} i_\partial^\pm (\beta \ra t).
\end{equation}

\begin{Theorem}
  \label{thm:closed_form_gen}
  Based on the principle of inclusion-exclusion, we have the following closed-from expression for
  the partial specificity and partial ambiguity,
  \begin{equation}
    i_\partial^\pm (\alpha \ra t) = i_\cap^\pm (\alpha \ra t)
      \, - \!\!\!\! \sum_{\emptyset \not= \gamma \subseteq \alpha^-} \!\!\!
      (-1)^{|\gamma|-1}\, i_\cap^\pm(\bigwedge \gamma \ra t)
  \end{equation}
\end{Theorem}

\begin{proof}
  For $\msc{B} \subseteq \msc{A}(\mb{s})$, define the sub-addative function
  $f^\pm(\msc{B})=\sum_{\beta \in \msc{B}} = i^\pm(\beta \ra t)$.  From \eq{icap}, we get that
  $i_\cap^\pm(\alpha \ra t) = f^\pm(\downarrow \alpha)$ and
  \begin{equation}
    i_\partial^\pm(\alpha \ra t)
      = f^\pm(\downarrow \alpha) - f^\pm(\overset{\textbf{.}}{\downarrow} \alpha)
      = f^\pm(\downarrow \alpha) - f^\pm(\bigcup_{\beta \in \alpha^-} \downarrow \beta).
  \end{equation}
  By the principle of inclusion-exclusion (e.g.\ see \citep[p.~195]{stanley2012}) we get that
  \begin{align}
    i_\partial^\pm(\alpha \ra t)
      &= f^\pm(\downarrow \alpha) \,
        - \!\!\!\! \sum_{\emptyset \not= \gamma \subseteq \alpha^-} \!\!\! (-1)^{|\gamma|-1} \;
         f^\pm(\bigcap_{\beta\in \gamma}\beta) \\
    \intertext{For any lattice $L$ and $A \subseteq L$, we have that
      ${\cap_{a \in A} \downarrow\! a = \;\downarrow\! (\bigwedge A)}$ (see \citep[p.~57]{davey2002}), thus}
    i_\partial^\pm(\alpha \ra t)
      &= f^\pm(\downarrow \alpha) \,
        - \!\!\!\! \sum_{\emptyset \not= \gamma \subseteq \alpha^-} \!\!\! (-1)^{|\gamma|-1} \,
        f^\pm(\bigwedge \gamma) \\
      &= f^\pm(\downarrow \alpha)
        - \!\!\!\! \sum_{\emptyset \not= \gamma \subseteq \alpha^-} \!\!\! (-1)^{|\gamma|-1} \,
        i^\pm(\bigwedge \gamma \ra t)
  \end{align}
  as required.
\end{proof}

Similarly to PID, the specificity and ambiguity lattices provide a structure for information
decomposition---unique evaluation requires a separate definition of redundancy.  However, unlike PID
(or even PPID), this evaluation requires both a definition of pointwise redundant specificity and
pointwise redundant ambiguity.

\subsection{Redundancy Measures on the Lattices}
\label{app:red}

In \secRef{measure}, Definitions~\ref{def:specificity} and \ref{def:ambiguity} provided the require
measures.  This section will prove some of the key properties of these measures when they are
applies to the lattices derived in the previous section.  The correspondence with the approach taken
by Williams and Beer~\citep{williams2010, williams2010priv} continues in this section.  However,
sources events $\mb{a}_i$ are used in place of sources $\mb{A}_i$ and the measures $\rminpn$ are
used in place of $\imin$.  Note that the basic concepts from lattice theory and the notion used here
are the same as found in \citep[Appendix~B]{williams2010}.

\theoremsatisfy*

\begin{proof}
  Axioms~\ref{ax:symmetry}, \ref{ax:sr} and \ref{ax:local_comp} follow trivially from the basic
  properties of the minimum.  The main statement of Axiom~\ref{ax:mono} also immediately follows
  from the properties of the minimum; however, there is a need to verify the equality condition.  As
  such, consider $\mb{a}_k$ such that $\mb{a}_k \supseteq \mb{a}_i$ for some
  $\mb{a}_i \in \{\mb{a}_1, \ldots, \mb{a}_{k-1}\}$. From Postulate~\ref{post:chain_rule}, we have
  that $h(\mb{a}_k) \geq h(\mb{a}_i)$ and hence that
  $\min_{\mb{a}_j\in\{\mb{a}_1,\ldots,\mb{a}_{k}\}}h(\mb{a}_j) =
  \min_{\mb{a}_j\in\{\mb{a}_1,\ldots,\mb{a}_{k-1}\}}h(\mb{a}_j)$, as required for $\rminp$.
  \textit{Mutatis mutandis}, similar follows for $\rminn$.
\end{proof}

\theoremmonotonically*

The proof of this theorem will require the following Lemma.

\begin{Lemma}
  \label{lem:monotonicity}
  The specificity and ambiguity $i^\pm(\mb{a} \ra t)$ are increasing functions on the lattice
  $\big\langle \msc{P}_1(\mb{s}), \subseteq \big\rangle$
\end{Lemma}

\begin{proof}
  Follows trivially from Postulate~\ref{post:chain_rule}.
\end{proof}

\begin{proof}[Proof of Theorem~\ref{thm:monotonically}]
  Assume there exists $\alpha,\,\beta \in \msc{A}(\mb{s})$ such that $\alpha \prec \beta$ and
  $\rminpn(\beta \ra t) < \rminpn(\alpha \ra t)$.  By definition, i.e.\ \eq{red_specificity} and
  \eq{red_ambiguity}, there exists $\mb{b} \in \beta$ such that
  $i^\pm(\mb{b} \ra t) < i^\pm(\mb{a} \ra t)$ for all $\mb{a} \in \alpha$.  Hence, by
  Lemma~\ref{lem:monotonicity}, there does not exist $\mb{a} \in \alpha$ such that
  $\mb{a} \subseteq \mb{b}$.  However, by assumption $\alpha \prec \beta$ and hence there exists
  $\mb{a} \in \alpha$ such that $\mb{a} \subseteq \mb{b}$, which is a contradiction.
\end{proof}

\begin{Theorem}
  \label{thm:closed_form}
  When using $\rminpn$ in place of the general redundancy measures $i_\cap^\pm$, we have the
  following closed-from expression for the partial specificity $\pi^+$ and partial ambiguity
  $\pi^-$,
  \begin{equation}
    \pi^\pm (\alpha \ra t) = \rminpn (\alpha \ra t)
    - \max_{\beta \in \alpha^-} \; \min_{\mb{b} \in \beta} \; i^\pm(\mb{b} \ra t). 
  \end{equation}
\end{Theorem}

\begin{proof}
  Let $i_\cap^+=\rminp$ and $i_\cap^-=\rminn$ in the general closed form expression for
  $i_\partial^\pm$ in Theorem~\ref{thm:closed_form_gen},
  \begin{equation}
    \pi^\pm(\alpha \ra t) = \rminpn (\alpha \ra t) \,
      - \!\!\!\! \sum_{\emptyset \not= \gamma \subseteq \alpha^-} \!\!\! (-1)^{|\gamma|-1}  
      \min_{\mb{b} \in \bigwedge \gamma} \; i^\pm( \mb{b} \ra t).
  \end{equation}
  Since $\alpha \wedge \beta = \underline{\alpha \cup \beta}$ (see \citep[Eq.~23]{williams2010}),
  and by Postulate~\ref{post:chain_rule}, we have that
  \begin{equation}
    \pi^\pm(\alpha \ra t) = \rminpn (\alpha \ra t) \,
      - \!\!\!\! \sum_{\emptyset \not= \gamma \subseteq \alpha^-} \!\!\!  (-1)^{|\gamma|-1}
      \min_{\beta \in \gamma} \; \min_{\mb{b} \in \beta} \; i^\pm( \mb{b} \ra t).
  \end{equation}
  By the maximum-minimums identity (see \citep{ross2009}), we have that,
  $\max \alpha^- = \sum_{\emptyset \not= \gamma \subseteq \alpha^-} (-1)^{|\gamma|-1} \; \min
  \gamma$, and hence
  \begin{equation}
    \pi^\pm (\alpha \ra t) = \rminpn (\alpha \ra t)
      - \max_{\beta \in \alpha^-} \; \min_{\mb{b} \in \beta} \; i^\pm(\alpha \ra t).    
  \end{equation}
  as required.
\end{proof}

\theoremnonnegativity*

\begin{proof}
  It $\alpha = \perp$, the $\pi^\pm(\alpha \ra t) = \rminpn \geq 0$ by the non-negativity of
  entropy.  If $\alpha \not= \perp$, assume there exists $\alpha \in \msc{A}(\mb{s})\sm\{\perp\}$
  such that $\pi^\pm(\alpha \ra t) < 0$.  By Theorem~\ref{thm:closed_form},
   \begin{equation}
    \pi^\pm (\alpha \ra t) = \min_{\mb{a} \in \alpha} \; i^\pm(\mb{a} \ra t)
      - \max_{\beta \in \alpha^-} \; \min_{\mb{b} \in \beta} \; i^\pm(\mb{b} \ra t). 
  \end{equation}
  From this it can be seen that there must exist $\beta \in \alpha^-$ such that for all
  $\mb{b} \in \beta$, we have that $i^\pm(\mb{a} \ra t) < i^\pm(\mb{b} \ra t)$ for some
  $\mb{a} \in \alpha$.  By Postulate~\ref{post:chain_rule} there does not exist $\mb{b} \in \beta$
  such that $\mb{b} \subset \mb{a}$.  However, since by definition, $\beta \prec \alpha$ there
  exists $\mb{b} \in \beta$ such that $\mb{b} \subset \mb{a}$, which is a contradiction.
\end{proof}

\theoremnonnegativityavg*

\begin{proof}
  The proof is by the counter-example using \imprdn{}.
\end{proof}

\subsection{Target Chain Rule}
\label{app:tcr}

By using the appropriate conditional probabilities in Definitions~\ref{def:specificity} and
\ref{def:ambiguity}, one can easily obtain the conditional pointwise redundant specificity,
\begin{equation}
  \label{eq:cond_red_specificity}
  \rminp \big( \mb{a}_1, \dots, \mb{a}_k \ra t_1 | t_2 \big) = \min_{\mb{a}_i} h(\mb{a}_i | t_2),
\end{equation}
or the conditional pointwise redundant ambiguity,
\begin{equation}
  \label{eq:cond_red_ambiguity}
  \rminn \big( \mb{a}_1, \dots, \mb{a}_k \ra t_1 | t_2 \big)
    = \min_{\mb{a}_j} h(\mb{a}_j | t_{1,2}).   
\end{equation}
As per \eq{local_recomp} these could be recombined, e.g.\ via \eq{local_recomp}, to obtain the
conditional redundant information,
\begin{equation}
  \label{eq:conditional_two_predictor_recomposition}
  \rmin \big( \mb{a}_1, \dots, \mb{a}_k \ra t_1 | t_2 \big)
  = \rminp \big( \mb{a}_1, \dots, \mb{a}_k \ra t_1 | t_2 \big)
  - \rminn \big( \mb{a}_1, \dots, \mb{a}_k \ra t_1 | t_2 \big).
\end{equation}
The relationship between the regular forms and the conditional forms of the redundant specificity
and redundant ambiguity has some important consequences. 

\begin{Proposition}
  \label{cor:cond_specificity_form}
  The conditional pointwise redundant specificity provided by $\mb{a}_1,\dots,\mb{a}_k$ about $t_1$
  given $t_2$ is equal to pointwise redundant ambiguity provided by $\mb{a}_1,\dots,\mb{a}_k$
  about $t_2$ with
  the conditioned variable,
  \begin{equation}
    \rminp \big( \mb{a}_1, \dots, \mb{a}_k \ra t_1 | t_2 \big)
      = \rminn \big( \mb{a}_1, \dots, \mb{a}_k \ra t_2 \big).
  \end{equation}
\end{Proposition}

\begin{proof}
  By \eq{red_ambiguity} and \eq{cond_red_specificity}.
\end{proof}

\begin{Proposition}
  \label{cor:red_spec_indpt_target}
  The pointwise redundant specificity provided by $\mb{a}_1,\dots,\mb{a}_k$ is independent of the
  target event and even the target variable itself,
  \begin{equation}
    \rminp \big( \mb{a}_1, \dots, \mb{a}_k \ra t_1 \big)
    = \rminp \big( \mb{a}_1, \dots, \mb{a}_k \ra t_2 \big)
    \qquad \forall\ t_1, t_2, T_1, T_2.
  \end{equation}
\end{Proposition}

\begin{proof}
  By inspection of \eq{red_specificity}.
\end{proof}

\begin{Proposition}
  \label{cor:cond_ambiguity_form}
  The conditional pointwise redundant ambiguity provided by $\mb{a}_1,\dots,\mb{a}_k$ about $t_1$
  given $t_2$ is equal to the pointwise redundant ambiguity provided by $\mb{a}_1,\dots,\mb{a}_k$
  about $t_{1,2}$,
  \begin{equation}
    \rminn \big( \mb{a}_1, \dots, \mb{a}_k \ra t_1 | t_2 \big)
      = \rminn \big( \mb{a}_1, \dots, \mb{a}_k \ra t_{1,2} \big).
  \end{equation}
\end{Proposition}

\begin{proof}
  By \eq{red_ambiguity} and \eq{cond_red_ambiguity}.
\end{proof}

Note that specificity itself is not a function of the target event or variable.  Hence, all of the
target dependency is bound up in the ambiguity. Now consider the following.

\theoremptcr*

\begin{proof}
  Starting from $\rmin$, by \cor{red_spec_indpt_target} and \cor{cond_ambiguity_form} we get that
  \begin{align}
    \rmin \big( \mb{a}_1,\ldots,\mb{a}_k \ra t_{1,2} \big)
    &= \rminp \big( \mb{a}_1,\ldots,\mb{a}_k \ra t_{1,2} \big)
      - \rminn \big( \mb{a}_1,\ldots,\mb{a}_k \ra t_{1,2}\big), \nn\\
    &= \rminp \big( \mb{a}_1,\ldots,\mb{a}_k \ra t_1 \big)
      - \rminn \big( \mb{a}_1,\ldots,\mb{a}_k \ra t_2 | t_1 \big),
    \label{eq:change_of_two_rmins}
  \end{align}
Then, by \cor{cond_specificity_form} we get that
  \begin{align}
    \begin{split}
      \rmin \big( \mb{a}_1,\ldots,\mb{a}_k \ra t_{1,2} \big)
      ={}& \rminp \big( \mb{a}_1,\ldots,\mb{a}_k \ra t_1 \big)
        - \rminn \big( \mb{a}_1,\ldots,\mb{a}_k \ra t_1 \big)  \\
      &\quad + \rminn \big( \mb{a}_1,\ldots,\mb{a}_k \ra t_1 \big)
        - \rminn \big( \mb{a}_1,\ldots,\mb{a}_k \ra t_2 | t_1 \big),
    \end{split} \nn\\
    \begin{split}
      ={}& \rminp \big( \mb{a}_1,\ldots,\mb{a}_k \ra t_1 \big)
      - \rminn \big( \mb{a}_1,\ldots,\mb{a}_k \ra t_1 \big)  \\
      &\quad + \rminp \big( \mb{a}_1,\ldots,\mb{a}_k \ra t_2 | t_1 \big)
      - \rminn \big( \mb{a}_1,\ldots,\mb{a}_k \ra t_2 | t_1 \big),
    \end{split}\nn\\
    ={}& \rmin \big( \mb{a}_1,\ldots,\mb{a}_k \ra t_1 \big)
         + \rmin \big( \mb{a}_1,\ldots,\mb{a}_k \ra t_2 | t_1 \big),
  \end{align}
  as required for \eq{target_chain_rule_1}.  \textit{Mutatis mutandis}, we obtain
  \eq{target_chain_rule_2}.
\end{proof}

\theoremnopos*

\begin{proof}
  Consider the probability distribution \tbc{}, and in particular, the isomorphic probability
  distributions $P(T_{1,2})$ and $P(T_{1,3})$.  By the identity property,
  \begin{align}
    \label{eq:tbc_result}
    U(S_1 \sm S_2 \ra T_{1,2}) &= 1\bit,  &  U(S_2 \sm S_1 \ra T_{1,2}) &= 1\bit, 
  \intertext{and hence, $R(S_1,S_2 \ra T_{1,2}) = 0\bit$. On the other hand, by local positivity,}
    C(S_1,S_2 \ra T_3) &= 1\bit,  &  R(S_1, S_2 \ra T_1 | T_3) &= 1\bit
  \intertext{Then by the target chain rule,}
    C(S_1,S_2 \ra T_{1,3}) &= 1\bit & R(S_1,S_2 \ra T_{1,3}) &= 1\bit,
  \end{align}
  Finally, since $P(T_{1,2})$ is isomorphic to $P(T_{1,3})$ we have that,
  $R(S_1,S_2 \ra T_{1,3}) = R(S_1,S_2 \ra T_{1,2})$, which is a contradiction.
\end{proof}

Theorem~\ref{thm:no_pos} can be informally generalised as follows:\ it is not possible to
simultaneously satisfy the target chain rule, the identity property, and have only $\CI=1\bit$ in
the probability distribution \xor{} without having negative (average) PI atoms in probability
distributions where there is no ambiguity from any source.  To see this, again consider decomposing
the isomorphic probability distributions $P(T_{1,2})$ and $P(T_{1,3})$.  In line with
\eq{tbc_result}, decomposing $T_{1,2}$ via the identity property yields
$C(S_1,S_2 \ra T_{1,2}) = 0\bit$.  On the other hand, decomposing $T_{1,3}$ yields
${C(S_1,S_2 \ra T_3) = 1\bit}$.  Since $P(T_{1,2})$ is isomorphic to $P(T_{1,3})$, the target chain
rule requires that,
\begin{align}
  C(S_1,S_2 \ra T_{1} | T_{3}) &= -1\bit,
  & U(S_1 \sm S_2 \ra T_{1} | T_{3}) &= 1\bit, & U(S_2 \sm S_1 \ra T_{1} | T_{3}) &= 1\bit.
\end{align}
That is, one would have to accept the negative (average) PI atom $C(S_1,S_2 \ra T_{1}|T_{3})=-1\bit$
despite the fact that there are no non-zero pointwise ambiguity terms upon splitting any of
$i(s_1;t_1|t_3)$, $i(s_2;t_1|t_3)$ and $i(s_{1,2};t_1|t_3)$ into specificity and ambiguity.
Although this does not constitute a formal proof that the identity property is incompatible with the
target chain rule, one would have to accept and find a way to justify
${C(S_1,S_2 \ra T_{1}|T_{3})=-1\bit}$.  Since there is no ambiguity in $i(s_1;t_1|t_3)$,
$i(s_2;t_1|t_3)$ and $i(s_{1,2};t_1|t_3)$, this result is not reconcilable within the framework of
specificity and ambiguity.

% \subsection{\xor{} Proof}
% 
% \note{CF: Fix this or temporarily comment out. }
% 
% The final noteworthy point, is that although apportionment of the informative exclusions across
% the full joint space is different, the redundant specificity is equal to 1~bit according to
% $\rminp$, as is required by Axiom~\ref{ax:local_comp}.  One might try to argue that because these
% exclusions are different exclusions in the full joint distribution, then the underlying
% information is not the same, and hence not redundant.  However, if one chose to follow that line
% of argument then one would not be left with zero redundancy in net and therefore only 1~bit of
% complementary information when the ambiguity is subtracted from the specificity.

%%%%%%%%%%%%%%%%%%%%%%%%%%%%%%%%%%%%%%%%%%%%%%%%%%%%%%%%%%%%%%%%%%%%%%%%%%%%%%%%%%%%%%%%%%%%%%%%%%%%
%%%%%%%%%%%%%%%%%%%%%%%%%%%%%%%%%%%%%%%%%%%%%%%%%%%%%%%%%%%%%%%%%%%%%%%%%%%%%%%%%%%%%%%%%%%%%%%%%%%%

\section{Additional Example Probability Distributions}
\label{app:additional_examples}

\subsection{Probability Distribution \tbp{}}
\label{sec:tbp}

\fig{tbp} shows the probability distribution three bit--even parity (\tbp{}) which considers binary
predictors variables $S_1$, $S_2$ and $S_3$ which are constrained such that together their parity is
even.  The target variable $T$ is simply a copy of the predictors, i.e.\ $T=T_{1,2,3}=(T_1,T_2,T_3)$
where $T_1=S_1$, $T_2=S_2$ and $T_3=S_3$.  (Equivalently, the target can be represented by any four
state variable $T$.) It was introduced by \citet{bertschinger2013} and revisited by \citet{rauh2014}
who (as mentioned in \secRef{tbc}) used it to prove the following by counter-example:\ there is no
measure of redundant average information for more than two predictor variables which simultaneously
satisfies the Williams and Beer Axioms, the identity property, and local positivity.  The measures
$\ired$, $\uitilde$ and $\ivk$ these properties.  Hence, this probability distribution which has
been used to demonstrate that these measures are not consistent with the PID framework in the
general case of an arbitrary number of predictor variables.

\begin{figure}[t]
  \vspace{-10pt}
 \centering %% \tablesize{} %% You can specify the fontsize here, e.g.
  \begin{tikzpicture}

  \def\originy{0} 
  \def\height{3}
  \def\width{1.5}
  \def\overdrawn{0.15}
  \def\gap{2.85}

  % -----------
  
  \def\originx{0}
  
  \draw[black, line width=1pt] (\originx,\originy)
    -- (\originx,\originy+\height);
  \draw[black, line width=1pt] (\originx-\overdrawn,\originy+\height)
    -- (\originx+\width+\overdrawn,\originy+\height);
  \draw[black, line width=1pt] (\originx+\width,\originy+\height)
    -- (\originx+\width,\originy);
  \draw[black, line width=1pt] (\originx+\width+\overdrawn,\originy)
    -- (\originx-\overdrawn,\originy);

  \draw[black, line width=1pt] (\originx-\overdrawn,\originy+3/4*\height)
    -- (\originx+\width+\overdrawn,\originy+3/4*\height);
  \draw[black, line width=1pt] (\originx-\overdrawn,\originy+\height/2)
    -- (\originx+\width+\overdrawn,\originy+\height/2);
  \draw[black, line width=1pt] (\originx-\overdrawn,\originy+1/4*\height)
  -- (\originx+\width+\overdrawn,\originy+1/4*\height);
  
  \node at (\originx,7/8*\height)[anchor=east] {$0$};
  \node at (\originx,5/8*\height)[anchor=east] {$1$};
  \node at (\originx,3/8*\height)[anchor=east] {$2$}; 
  \node at (\originx,1/8*\height)[anchor=east] {$3$}; 
  \node at (\originx+\width,7/8*\height)[anchor=west] {$000$};
  \node at (\originx+\width,5/8*\height)[anchor=west] {$011$};
  \node at (\originx+\width,3/8*\height)[anchor=west] {$101$};
  \node at (\originx+\width,1/8*\height)[anchor=west] {$110$};
  \node at (\originx+\width/2,\originy+\height)[anchor=south] {$P(S_{1,2,3},T)$};
  \node at (\originx+\width/2,\originy+7/8*\height)[anchor=center] {$\nf{1}{4}$};
  \node at (\originx+\width/2,\originy+5/8*\height)[anchor=center] {$\nf{1}{4}$};
  \node at (\originx+\width/2,\originy+3/8*\height)[anchor=center] {$\nf{1}{4}$};
  \node at (\originx+\width/2,\originy+1/8*\height)[anchor=center] {$\nf{1}{4}$};
  
  % -----------
  
  \def\originxtwo{\gap}
  
  \draw[pattern=vertical lines, pattern color=red] (\originxtwo,\originy+1/4*\height) rectangle
    (\originxtwo+\width,\originy+1/2*\height);
  \draw[pattern=vertical lines, pattern color=red] (\originxtwo,\originy) rectangle
    (\originxtwo+\width,\originy+1/4*\height);
  
  \draw[black, line width=1pt] (\originxtwo,\originy)
    -- (\originxtwo,\originy+\height);
  \draw[black, line width=1pt] (\originxtwo-\overdrawn,\originy+\height)
    -- (\originxtwo+\width+\overdrawn,\originy+\height);
  \draw[black, line width=1pt] (\originxtwo+\width,\originy+\height)
    -- (\originxtwo+\width,\originy);
  \draw[black, line width=1pt] (\originxtwo+\width+\overdrawn,\originy)
    -- (\originxtwo-\overdrawn,\originy);

  \draw[black, line width=1pt] (\originxtwo-\overdrawn,\originy+3/4*\height)
    -- (\originxtwo+\width+\overdrawn,\originy+3/4*\height);
  \draw[black, line width=1pt] (\originxtwo-\overdrawn,\originy+\height/2)
    -- (\originxtwo+\width+\overdrawn,\originy+\height/2);
  \draw[black, line width=1pt] (\originxtwo-\overdrawn,\originy+1/4*\height)
  -- (\originxtwo+\width+\overdrawn,\originy+1/4*\height);

  \node at (\originxtwo,7/8*\height)[anchor=east] {$0$};
  \node at (\originxtwo,5/8*\height)[anchor=east] {$1$};
  \node at (\originxtwo,3/8*\height)[anchor=east] {$2$}; 
  \node at (\originxtwo,1/8*\height)[anchor=east] {$3$}; 
  \node at (\originxtwo+\width,7/8*\height)[anchor=west] {$000$};
  \node at (\originxtwo+\width,5/8*\height)[anchor=west] {$011$};
  \node at (\originxtwo+\width,3/8*\height)[anchor=west] {$101$};
  \node at (\originxtwo+\width,1/8*\height)[anchor=west] {$110$};
  \node at (\originxtwo+\width/2,\originy+\height)[anchor=south] {$S_1=0$};

  % -----------
  
  \def\originxthree{2*\gap}
  
  \draw[pattern=horizontal lines, pattern color=blue] (\originxthree,\originy+\height/2) rectangle
    (\originxthree+\width,\originy+3/4*\height);
  \draw[pattern=horizontal lines, pattern color=blue] (\originxthree,\originy) rectangle
    (\originxthree+\width,\originy+1/4*\height);
    
  \draw[black, line width=1pt] (\originxthree,\originy)
    -- (\originxthree,\originy+\height);
  \draw[black, line width=1pt] (\originxthree-\overdrawn,\originy+\height)
    -- (\originxthree+\width+\overdrawn,\originy+\height);
  \draw[black, line width=1pt] (\originxthree+\width,\originy+\height)
    -- (\originxthree+\width,\originy);
  \draw[black, line width=1pt] (\originxthree+\width+\overdrawn,\originy)
    -- (\originxthree-\overdrawn,\originy);

  \draw[black, line width=1pt] (\originxthree-\overdrawn,\originy+3/4*\height)
    -- (\originxthree+\width+\overdrawn,\originy+3/4*\height);
  \draw[black, line width=1pt] (\originxthree-\overdrawn,\originy+\height/2)
    -- (\originxthree+\width+\overdrawn,\originy+\height/2);
  \draw[black, line width=1pt] (\originxthree-\overdrawn,\originy+1/4*\height)
  -- (\originxthree+\width+\overdrawn,\originy+1/4*\height);
  
  \node at (\originxthree,7/8*\height)[anchor=east] {$0$};
  \node at (\originxthree,5/8*\height)[anchor=east] {$1$};
  \node at (\originxthree,3/8*\height)[anchor=east] {$2$}; 
  \node at (\originxthree,1/8*\height)[anchor=east] {$3$}; 
  \node at (\originxthree+\width,7/8*\height)[anchor=west] {$000$};
  \node at (\originxthree+\width,5/8*\height)[anchor=west] {$011$};
  \node at (\originxthree+\width,3/8*\height)[anchor=west] {$101$};
  \node at (\originxthree+\width,1/8*\height)[anchor=west] {$110$};
  \node at (\originxthree+\width/2,\originy+\height)[anchor=south] {$S_2=0$};

  % -----------
  
  \def\originxfour{3*\gap}
  
  \filldraw[fill=green, opacity=0.175] (\originxfour,\originy+1/4*\height) rectangle
    (\originxfour+\width,\originy+3/4*\height);  
    
  \draw[black, line width=1pt] (\originxfour,\originy)
    -- (\originxfour,\originy+\height);
  \draw[black, line width=1pt] (\originxfour-\overdrawn,\originy+\height)
    -- (\originxfour+\width+\overdrawn,\originy+\height);
  \draw[black, line width=1pt] (\originxfour+\width,\originy+\height)
    -- (\originxfour+\width,\originy);
  \draw[black, line width=1pt] (\originxfour+\width+\overdrawn,\originy)
    -- (\originxfour-\overdrawn,\originy);

  \draw[black, line width=1pt] (\originxfour-\overdrawn,\originy+3/4*\height)
    -- (\originxfour+\width+\overdrawn,\originy+3/4*\height);
  \draw[black, line width=1pt] (\originxfour-\overdrawn,\originy+\height/2)
    -- (\originxfour+\width+\overdrawn,\originy+\height/2);
  \draw[black, line width=1pt] (\originxfour-\overdrawn,\originy+1/4*\height)
  -- (\originxfour+\width+\overdrawn,\originy+1/4*\height);
  
  \node at (\originxfour,7/8*\height)[anchor=east] {$0$};
  \node at (\originxfour,5/8*\height)[anchor=east] {$1$};
  \node at (\originxfour,3/8*\height)[anchor=east] {$2$}; 
  \node at (\originxfour,1/8*\height)[anchor=east] {$3$}; 
  \node at (\originxfour+\width,7/8*\height)[anchor=west] {$000$};
  \node at (\originxfour+\width,5/8*\height)[anchor=west] {$011$};
  \node at (\originxfour+\width,3/8*\height)[anchor=west] {$101$};
  \node at (\originxfour+\width,1/8*\height)[anchor=west] {$110$};
  \node at (\originxfour+\width/2,\originy+\height)[anchor=south] {$S_3=0$};
  
   % -----------
  
  \def\originxfive{4*\gap}
  
  \filldraw[fill=green, opacity=0.175] (\originxfive,\originy+1/4*\height) rectangle
    (\originxfive+\width,\originy+3/4*\height);  
  \draw[pattern=vertical lines, pattern color=red] (\originxfive,\originy+1/4*\height) rectangle
    (\originxfive+\width,\originy+1/2*\height);
  \draw[pattern=vertical lines, pattern color=red] (\originxfive,\originy) rectangle
    (\originxfive+\width,\originy+1/4*\height);
  \draw[pattern=horizontal lines, pattern color=blue] (\originxfive,\originy+\height/2) rectangle
    (\originxfive+\width,\originy+3/4*\height);
  \draw[pattern=horizontal lines, pattern color=blue] (\originxfive,\originy) rectangle
    (\originxfive+\width,\originy+1/4*\height);

  \draw[black, line width=1pt] (\originxfive,\originy)
    -- (\originxfive,\originy+\height);
  \draw[black, line width=1pt] (\originxfive-\overdrawn,\originy+\height)
    -- (\originxfive+\width+\overdrawn,\originy+\height);
  \draw[black, line width=1pt] (\originxfive+\width,\originy+\height)
    -- (\originxfive+\width,\originy);
  \draw[black, line width=1pt] (\originxfive+\width+\overdrawn,\originy)
    -- (\originxfive-\overdrawn,\originy);

  \draw[black, line width=1pt] (\originxfive-\overdrawn,\originy+3/4*\height)
    -- (\originxfive+\width+\overdrawn,\originy+3/4*\height);
  \draw[black, line width=1pt] (\originxfive-\overdrawn,\originy+\height/2)
    -- (\originxfive+\width+\overdrawn,\originy+\height/2);
  \draw[black, line width=1pt] (\originxfive-\overdrawn,\originy+1/4*\height)
  -- (\originxfive+\width+\overdrawn,\originy+1/4*\height);
  
  \node at (\originxfive,7/8*\height)[anchor=east] {$0$};
  \node at (\originxfive,5/8*\height)[anchor=east] {$1$};
  \node at (\originxfive,3/8*\height)[anchor=east] {$2$}; 
  \node at (\originxfive,1/8*\height)[anchor=east] {$3$}; 
  \node at (\originxfive+\width,7/8*\height)[anchor=west] {$000$};
  \node at (\originxfive+\width,5/8*\height)[anchor=west] {$011$};
  \node at (\originxfive+\width,3/8*\height)[anchor=west] {$101$};
  \node at (\originxfive+\width,1/8*\height)[anchor=west] {$110$};
  \node at (\originxfive+\width/2,\originy+\height)[anchor=south] {$S_{1,2,3}=000$};

\end{tikzpicture} \vskip 1em \begin{tikzpicture}

  \def\ox{0}
  \def\oy{0}
  \def\gx{1.7}
  \def\gy{1.5}

  \tikzstyle{block} = [rectangle, draw=none, 
    text width=5em, text centered, minimum height=1em]

    \small
    
% -----------------------------------------------------------------
    
  \node [block] (abc)      at (\ox,       \oy+6*\gy) {$\{123\}$};  

  \node [block] (bc)       at (\ox+\gx,   \oy+5*\gy) {$\{23\}$};
  \node [block] (ac)       at (\ox,       \oy+5*\gy) {$\{13\}$};
  \node [block] (ab)       at (\ox-\gx,   \oy+5*\gy) {$\{12\}$};
  
  \node [block] (acIbc)    at (\ox+\gx,   \oy+4*\gy) {$\{13\}\{23\}$};
  \node [block] (abIbc)    at (\ox,       \oy+4*\gy) {$\{12\}\{23\}$};
  \node [block] (abIac)    at (\ox-\gx,   \oy+4*\gy) {$\{12\}\{13\}$};

  \node [block] (abIacIbc) at (\ox+2*\gx, \oy+3*\gy) {$\{12\}\{13\}\{23\}$};
  
  \node [block] (c)        at (\ox+\gx,   \oy+3*\gy) {$\{3\}$};
  \node [block] (b)        at (\ox,       \oy+3*\gy) {$\{2\}$};
  \node [block] (a)        at (\ox-\gx,   \oy+3*\gy) {$\{1\}$};   

  \node [block] (cIab)     at (\ox+\gx,   \oy+2*\gy) {$\{3\}\{12\}$};
  \node [block] (bIac)     at (\ox,       \oy+2*\gy) {$\{2\}\{13\}$};
  \node [block] (aIbc)     at (\ox-\gx,   \oy+2*\gy) {$\{1\}\{23\}$}; 

  \node [block] (bIc)      at (\ox+\gx,   \oy+1*\gy) {$\{2\}\{3\}$};
  \node [block] (aIc)      at (\ox,       \oy+1*\gy) {$\{1\}\{3\}$};
  \node [block] (aIb)      at (\ox-\gx,   \oy+1*\gy) {$\{1\}\{2\}$};

  \node [block] (aIbIc)    at (\ox,       \oy      ) {$\{1\}\{2\}\{3\}$};

  \draw[->, black, line width=0.75pt] (abc) -- (bc);
  \draw[->, black, line width=0.75pt] (abc) -- (ac);
  \draw[->, black, line width=0.75pt] (abc) -- (ab); 
  
  \draw[->, black, line width=0.75pt] (bc) -- (abIbc);
  \draw[->, black, line width=0.75pt] (ac) -- (acIbc);
  \draw[->, black, line width=0.75pt] (ab) -- (abIbc); 
  \draw[->, black, line width=0.75pt] (bc) -- (acIbc);
  \draw[->, black, line width=0.75pt] (ac) -- (abIac);
  \draw[->, black, line width=0.75pt] (ab) -- (abIac);

  \draw[->, black, line width=0.75pt] (acIbc) -- (abIacIbc);
  \draw[->, black, line width=0.75pt] (abIbc) -- (abIacIbc);
  \draw[->, black, line width=0.75pt] (abIac) -- (abIacIbc); 

  \draw[->, black, line width=0.75pt] (acIbc) -- (c);
  \draw[->, black, line width=0.75pt] (abIbc) -- (b);
  \draw[->, black, line width=0.75pt] (abIac) -- (a);

  \draw[->, black, line width=0.75pt] (abIacIbc) -- (cIab);
  \draw[->, black, line width=0.75pt] (abIacIbc) -- (bIac);
  \draw[->, black, line width=0.75pt] (abIacIbc) -- (aIbc);
  
  \draw[->, black, line width=0.75pt] (c) -- (cIab);
  \draw[->, black, line width=0.75pt] (b) -- (bIac);
  \draw[->, black, line width=0.75pt] (a) -- (aIbc);

  \draw[->, black, line width=0.75pt] (cIab) -- (bIc);
  \draw[->, black, line width=0.75pt] (bIac) -- (bIc);
  \draw[->, black, line width=0.75pt] (aIbc) -- (aIc); 
  \draw[->, black, line width=0.75pt] (cIab) -- (aIc);
  \draw[->, black, line width=0.75pt] (bIac) -- (aIb);
  \draw[->, black, line width=0.75pt] (aIbc) -- (aIb);
  
  \draw[->, black, line width=0.75pt] (bIc) -- (aIbIc);
  \draw[->, black, line width=0.75pt] (aIc) -- (aIbIc);
  \draw[->, black, line width=0.75pt] (aIb) -- (aIbIc);

% -----------------------------------------------------------------

  \tikzstyle{blockb} = [rectangle, draw=none, 
    text width=2em, text centered, minimum height=0.7em]

  \def\obx{8}

  \node [blockb] (abc)      at (\obx,       \oy+6*\gy) {$2(0)$};  

  \node [blockb] (bc)       at (\obx+\gx,   \oy+5*\gy) {$2(0)$};
  \node [blockb] (ac)       at (\obx,       \oy+5*\gy) {$2(0)$};
  \node [blockb] (ab)       at (\obx-\gx,   \oy+5*\gy) {$2(0)$};
  
  \node [blockb] (acIbc)    at (\obx+\gx,   \oy+4*\gy) {$2(0)$};
  \node [blockb] (abIbc)    at (\obx,       \oy+4*\gy) {$2(0)$};
  \node [blockb] (abIac)    at (\obx-\gx,   \oy+4*\gy) {$2(0)$};

  \node [blockb] (abIacIbc) at (\obx+2*\gx, \oy+3*\gy) {$2(1)$};
  
  \node [blockb] (c)        at (\obx+\gx,   \oy+3*\gy) {$1(0)$};
  \node [blockb] (b)        at (\obx,       \oy+3*\gy) {$1(0)$};
  \node [blockb] (a)        at (\obx-\gx,   \oy+3*\gy) {$1(0)$};   

  \node [blockb] (cIab)     at (\obx+\gx,   \oy+2*\gy) {$1(0)$};
  \node [blockb] (bIac)     at (\obx,       \oy+2*\gy) {$1(0)$};
  \node [blockb] (aIbc)     at (\obx-\gx,   \oy+2*\gy) {$1(0)$}; 

  \node [blockb] (bIc)      at (\obx+\gx,   \oy+1*\gy) {$1(0)$};
  \node [blockb] (aIc)      at (\obx,       \oy+1*\gy) {$1(0)$};
  \node [blockb] (aIb)      at (\obx-\gx,   \oy+1*\gy) {$1(0)$};

  \node [blockb] (aIbIc)    at (\obx,       \oy      ) {$1(1)$};

  \draw[->, black, line width=0.75pt] (abc) -- (bc);
  \draw[->, black, line width=0.75pt] (abc) -- (ac);
  \draw[->, black, line width=0.75pt] (abc) -- (ab); 
  
  \draw[->, black, line width=0.75pt] (bc) -- (abIbc);
  \draw[->, black, line width=0.75pt] (ac) -- (acIbc);
  \draw[->, black, line width=0.75pt] (ab) -- (abIbc); 
  \draw[->, black, line width=0.75pt] (bc) -- (acIbc);
  \draw[->, black, line width=0.75pt] (ac) -- (abIac);
  \draw[->, black, line width=0.75pt] (ab) -- (abIac);

  \draw[->, black, line width=0.75pt] (acIbc) -- (abIacIbc);
  \draw[->, black, line width=0.75pt] (abIbc) -- (abIacIbc);
  \draw[->, black, line width=0.75pt] (abIac) -- (abIacIbc); 

  \draw[->, black, line width=0.75pt] (acIbc) -- (c);
  \draw[->, black, line width=0.75pt] (abIbc) -- (b);
  \draw[->, black, line width=0.75pt] (abIac) -- (a);

  \draw[->, black, line width=0.75pt] (abIacIbc) -- (cIab);
  \draw[->, black, line width=0.75pt] (abIacIbc) -- (bIac);
  \draw[->, black, line width=0.75pt] (abIacIbc) -- (aIbc);
  
  \draw[->, black, line width=0.75pt] (c) -- (cIab);
  \draw[->, black, line width=0.75pt] (b) -- (bIac);
  \draw[->, black, line width=0.75pt] (a) -- (aIbc);

  \draw[->, black, line width=0.75pt] (cIab) -- (bIc);
  \draw[->, black, line width=0.75pt] (bIac) -- (bIc);
  \draw[->, black, line width=0.75pt] (aIbc) -- (aIc); 
  \draw[->, black, line width=0.75pt] (cIab) -- (aIc);
  \draw[->, black, line width=0.75pt] (bIac) -- (aIb);
  \draw[->, black, line width=0.75pt] (aIbc) -- (aIb);
  
  \draw[->, black, line width=0.75pt] (bIc) -- (aIbIc);
  \draw[->, black, line width=0.75pt] (aIc) -- (aIbIc);
  \draw[->, black, line width=0.75pt] (aIb) -- (aIbIc);

\end{tikzpicture}
  \caption{\label{fig:tbp} Example \tbp{}. \emph{Top}:\ probability mass diagram for realisation
    ${(S_1\!=\!0,S_2\!=\!0,S_3\!=\!0,T\!=\!000)}$.  \textit{Bottom left}:\ With three predictors, it
    is convenient to represent to decomposition diagrammatically.  This is especially true \tbp{} as
    one only needs to consider the specificity lattice for one realisation.  \emph{Bottom right}:\
    The specificity lattice for the realisation ${(S_1\!=\!0,S_2\!=\!0,S_3\!=\!0,T\!=\!000)}$.  For
    each source event the left value corresponds to the value of $i_\cap^+$, evaluated using
    $\rminp$, while the right value (surrounded by parenthesis) corresponds to the partial
    information $\pi^+$.}
 \end{figure}
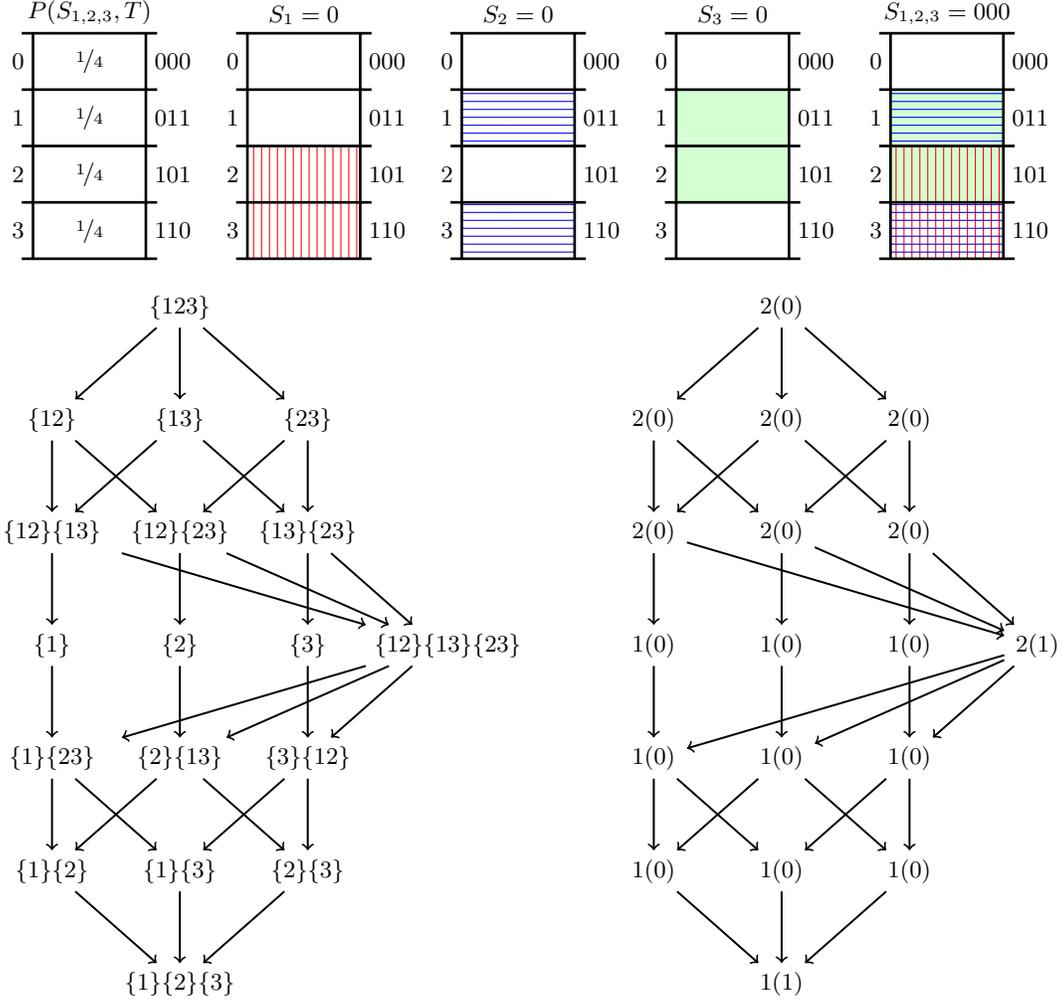

This example is similar to \tbc{} in the several ways.  Firstly, due to the symmetry in the
probability distribution, each realisation will have the same pointwise decomposition.  Secondly,
there is an isomorphism\footnote{Again, isomorphism should be taken to mean isomorphic probability
  spaces, e.g.\ \citep[p.~27]{gray1988} or \citep[p.4]{martin1984}.} between $P(T)$ and
$P(S_1, S_2, S_3)$, and hence the pointwise ambiguity provided by any (individual or joint)
predictor event is $0\bit$ (since given $t$, one knows $s_1$, $s_2$ and $s_3$).  Thirdly, the
individual predictor events $s_1$, $s_2$ and $s_3$ each exclude $\nf{1}{2}$ of the total probability
mass in $P(T)$ and so each provide $1\bit$ of pointwise specificity.  Thus, there is $1\bit$ of
three-way redundant, pointwise specificity in each realisation.  Fourthly, the joint predictor event
$s_{1,2,3}$ excludes $\nf{3}{4}$ of the total probability mass, providing $2\bit$ of pointwise
specificity (which is similar to \tbc{}).  However, unlike \tbc{}, one could consider the three
joint predictor events $s_{1,2}$, $s_{1,3}$ and $s_{2,3}$.  These joint pairs also exclude
$\nf{3}{4}$ of the total probability mass each, and hence also each provide $2\bit$ of pointwise
specificity.  As such, there is $1\bit$ of pointwise, three-way redundant, pairwise complementary
specificity between these three joint pairs of source events, in addition to the $1\bit$ of
three-way redundant, pointwise specificity.  Finally, putting this together and averaging over all
realisations, \tbp{} consists of $1\bit$ of three-way redundant information and $1\bit$ of three-way
redundant, pairwise complementary information.  The resultant average decomposition is the same as
the decomposition induced by~$\imin$~\citep{bertschinger2013}.

\subsection{Probability Distribution \unq{}}

\fig{unq} shows the decomposition of the probability distribution \emph{unique} (\unq{}).  Note that
this probability distribution corresponds to \imprdn{} where the error probability
$\varepsilon = \nf{1}{2}$, and hence the similarity in the resultant distributions.  The results may
initially seem unusual, that the predictor $S_1$ is not uniquely informative since ${\UIo=0\bit}$ as
one might intuitively expect.  Rather it is deemed to be redundantly informative ${RI=1\bit}$ with
the predictor $S_2$ which is also uniquely misinformative ${\UIt=-1\bit}$.  This is because both
$S_1$ and $S_2$ provide ${I^+(S_1 \ra T)=I^+(S_2 \ra T)=1\bit}$ of specificity; however the
information provided by $S_2$ is unique in that the $1\bit$ provided is not ``useful'' and hence
${I(S_2 \ra T) = 1\bit}$ while ${I(S_2 \ra T) = 1\bit}$~\citep[p.~21]{shannon1998}.  Finally, the
complementary information $\CI=1\bit$ is required by the decomposition in order to balance this
$1\bit$ of unique ambiguity.  The results in this example partly explain our preference for term
\emph{complementary information} as opposed to \emph{synergistic information}---while $\CI=1\bit$ is
readily explainable, it would be dubious to refer to this as \emph{synergy} given that $S_1$ enables
perfect predictions of $T$ without any knowledge of $S_2$.

\begin{figure}[t]
  \centering %% \tablesize{} %% You can specify the fontsize here, e.g.
  \begin{tikzpicture}

  \def\originy{0} 
  \def\height{3}
  \def\width{1.5}
  \def\overdrawn{0.15}

  % -----------
  
  \def\originx{0}
  
  \draw[black, line width=1pt] (\originx,\originy)
    -- (\originx,\originy+\height);
  \draw[black, line width=1pt] (\originx-\overdrawn,\originy+\height)
    -- (\originx+\width+\overdrawn,\originy+\height);
  \draw[black, line width=1pt] (\originx+\width,\originy+\height)
    -- (\originx+\width,\originy);
  \draw[black, line width=1pt] (\originx+\width+\overdrawn,\originy)
    -- (\originx-\overdrawn,\originy);

  \draw[black, line width=0.75pt] (\originx,\originy+3/4*\height)
    -- (\originx+\width+\overdrawn,\originy+3/4*\height);
  \draw[black, line width=1pt] (\originx-\overdrawn,\originy+\height/2)
    -- (\originx+\width+\overdrawn,\originy+\height/2);
  \draw[black, line width=0.75pt] (\originx,\originy+1/4*\height)
  -- (\originx+\width+\overdrawn,\originy+1/4*\height);

  \node at (\originx,3/4*\height)[anchor=east] {$0$};
  \node at (\originx,1/4*\height)[anchor=east] {$1$}; 
  \node at (\originx+\width,7/8*\height)[anchor=west] {$00$};
  \node at (\originx+\width,5/8*\height)[anchor=west] {$01$};
  \node at (\originx+\width,3/8*\height)[anchor=west] {$10$};
  \node at (\originx+\width,1/8*\height)[anchor=west] {$11$};
  \node at (\originx+\width/2,\originy+\height)[anchor=south] {$P(S_{1,2},T)$};
  \node at (\originx+\width/2,\originy+7/8*\height)[anchor=center] {$\nf{1}{4}$};
  \node at (\originx+\width/2,\originy+5/8*\height)[anchor=center] {$\nf{1}{4}$};
  \node at (\originx+\width/2,\originy+3/8*\height)[anchor=center] {$\nf{1}{4}$};
  \node at (\originx+\width/2,\originy+1/8*\height)[anchor=center] {$\nf{1}{4}$}; 

  % -----------
  
  \def\originxtwo{3}
  
  \draw[black, line width=1pt] (\originxtwo,\originy)
    -- (\originxtwo,\originy+\height);
  \draw[black, line width=1pt] (\originxtwo-\overdrawn,\originy+\height)
    -- (\originxtwo+\width+\overdrawn,\originy+\height);
  \draw[black, line width=1pt] (\originxtwo+\width,\originy+\height)
    -- (\originxtwo+\width,\originy);
  \draw[black, line width=1pt] (\originxtwo+\width+\overdrawn,\originy)
    -- (\originxtwo-\overdrawn,\originy);

  \draw[black, line width=0.75pt] (\originxtwo,\originy+3/4*\height)
    -- (\originxtwo+\width+\overdrawn,\originy+3/4*\height);
  \draw[black, line width=1pt] (\originxtwo-\overdrawn,\originy+\height/2)
    -- (\originxtwo+\width+\overdrawn,\originy+\height/2);
  \draw[black, line width=0.75pt] (\originxtwo,\originy+1/4*\height)
  -- (\originxtwo+\width+\overdrawn,\originy+1/4*\height);

  \draw[pattern=north east lines, pattern color=red] (\originxtwo,\originy+\height/2) rectangle
    (\originxtwo+\width,\originy+3/4*\height);
  \draw[pattern=vertical lines, pattern color=red] (\originxtwo,\originy) rectangle
    (\originxtwo+\width,\originy+1/4*\height);
  \draw[pattern=horizontal lines, pattern color=blue] (\originxtwo,\originy) rectangle
    (\originxtwo+\width,\originy+1/2*\height);
  
  \node at (\originxtwo,3/4*\height)[anchor=east] {$0$};
  \node at (\originxtwo,1/4*\height)[anchor=east] {$1$}; 
  \node at (\originxtwo+\width,7/8*\height)[anchor=west] {$00$};
  \node at (\originxtwo+\width,5/8*\height)[anchor=west] {$01$};
  \node at (\originxtwo+\width,3/8*\height)[anchor=west] {$10$};
  \node at (\originxtwo+\width,1/8*\height)[anchor=west] {$11$};
  \node at (\originxtwo+\width/2,\originy+\height)[anchor=south] {$S_{1,2}=00$};
  
  % -----------
  
  \def\originxthree{6}
  
  \draw[black, line width=1pt] (\originxthree,\originy)
    -- (\originxthree,\originy+\height);
  \draw[black, line width=1pt] (\originxthree-\overdrawn,\originy+\height)
    -- (\originxthree+\width+\overdrawn,\originy+\height);
  \draw[black, line width=1pt] (\originxthree+\width,\originy+\height)
    -- (\originxthree+\width,\originy);
  \draw[black, line width=1pt] (\originxthree+\width+\overdrawn,\originy)
    -- (\originxthree-\overdrawn,\originy);

  \draw[black, line width=0.75pt] (\originxthree,\originy+3/4*\height)
    -- (\originxthree+\width+\overdrawn,\originy+3/4*\height);
  \draw[black, line width=1pt] (\originxthree-\overdrawn,\originy+\height/2)
    -- (\originxthree+\width+\overdrawn,\originy+\height/2);
  \draw[black, line width=0.75pt] (\originxthree,\originy+1/4*\height)
  -- (\originxthree+\width+\overdrawn,\originy+1/4*\height);

  \draw[pattern=north east lines, pattern color=red] (\originxthree,\originy+3/4*\height) rectangle
    (\originxthree+\width,\originy+\height);
  \draw[pattern=vertical lines, pattern color=red] (\originxthree,\originy+1/4*\height) rectangle
    (\originxthree+\width,\originy+\height/2);
  \draw[pattern=horizontal lines, pattern color=blue] (\originxthree,\originy) rectangle
    (\originxthree+\width,\originy+1/2*\height);
  
  \node at (\originxthree,3/4*\height)[anchor=east] {$0$};
  \node at (\originxthree,1/4*\height)[anchor=east] {$1$}; 
  \node at (\originxthree+\width,7/8*\height)[anchor=west] {$00$};
  \node at (\originxthree+\width,5/8*\height)[anchor=west] {$01$};
  \node at (\originxthree+\width,3/8*\height)[anchor=west] {$10$};
  \node at (\originxthree+\width,1/8*\height)[anchor=west] {$11$};
  \node at (\originxthree+\width/2,\originy+\height)[anchor=south] {$S_{1,2}=01$};

  % -----------
  
  \def\originxfour{9}
  
  \draw[black, line width=1pt] (\originxfour,\originy)
    -- (\originxfour,\originy+\height);
  \draw[black, line width=1pt] (\originxfour-\overdrawn,\originy+\height)
    -- (\originxfour+\width+\overdrawn,\originy+\height);
  \draw[black, line width=1pt] (\originxfour+\width,\originy+\height)
    -- (\originxfour+\width,\originy);
  \draw[black, line width=1pt] (\originxfour+\width+\overdrawn,\originy)
    -- (\originxfour-\overdrawn,\originy);

  \draw[black, line width=0.75pt] (\originxfour,\originy+3/4*\height)
    -- (\originxfour+\width+\overdrawn,\originy+3/4*\height);
  \draw[black, line width=1pt] (\originxfour-\overdrawn,\originy+\height/2)
    -- (\originxfour+\width+\overdrawn,\originy+\height/2);
  \draw[black, line width=0.75pt] (\originxfour,\originy+1/4*\height)
  -- (\originxfour+\width+\overdrawn,\originy+1/4*\height);

  \draw[pattern=vertical lines, pattern color=red] (\originxfour,\originy+\height/2) rectangle
    (\originxfour+\width,\originy+3/4*\height);
  \draw[pattern=north east lines, pattern color=red] (\originxfour,\originy) rectangle
    (\originxfour+\width,\originy+1/4*\height);
  \draw[pattern=horizontal lines, pattern color=blue] (\originxfour,\originy+\height) rectangle
    (\originxfour+\width,\originy+\height/2);

  \node at (\originxfour,3/4*\height)[anchor=east] {$0$};
  \node at (\originxfour,1/4*\height)[anchor=east] {$1$}; 
  \node at (\originxfour+\width,7/8*\height)[anchor=west] {$00$};
  \node at (\originxfour+\width,5/8*\height)[anchor=west] {$01$};
  \node at (\originxfour+\width,3/8*\height)[anchor=west] {$10$};
  \node at (\originxfour+\width,1/8*\height)[anchor=west] {$11$};
  \node at (\originxfour+\width/2,\originy+\height)[anchor=south] {$S_{1,2}=10$};
  
  % -----------
  
  \def\originxfive{12}
  
  \draw[black, line width=1pt] (\originxfive,\originy)
    -- (\originxfive,\originy+\height);
  \draw[black, line width=1pt] (\originxfive-\overdrawn,\originy+\height)
    -- (\originxfive+\width+\overdrawn,\originy+\height);
  \draw[black, line width=1pt] (\originxfive+\width,\originy+\height)
    -- (\originxfive+\width,\originy);
  \draw[black, line width=1pt] (\originxfive+\width+\overdrawn,\originy)
    -- (\originxfive-\overdrawn,\originy);

  \draw[black, line width=0.75pt] (\originxfive,\originy+3/4*\height)
    -- (\originxfive+\width+\overdrawn,\originy+3/4*\height);
  \draw[black, line width=1pt] (\originxfive-\overdrawn,\originy+\height/2)
    -- (\originxfive+\width+\overdrawn,\originy+\height/2);
  \draw[black, line width=0.75pt] (\originxfive,\originy+1/4*\height)
  -- (\originxfive+\width+\overdrawn,\originy+1/4*\height);
 
  \draw[pattern=vertical lines, pattern color=red] (\originxfive,\originy+3/4*\height) rectangle
    (\originxfive+\width,\originy+\height);
  \draw[pattern=north east lines, pattern color=red] (\originxfive,\originy+1/4*\height) rectangle
    (\originxfive+\width,\originy+\height/2);
  \draw[pattern=horizontal lines, pattern color=blue] (\originxfive,\originy+\height) rectangle
    (\originxfive+\width,\originy+\height/2);
  
  \node at (\originxfive,3/4*\height)[anchor=east] {$0$};
  \node at (\originxfive,1/4*\height)[anchor=east] {$1$}; 
  \node at (\originxfive+\width,7/8*\height)[anchor=west] {$00$};
  \node at (\originxfive+\width,5/8*\height)[anchor=west] {$01$};
  \node at (\originxfive+\width,3/8*\height)[anchor=west] {$10$};
  \node at (\originxfive+\width,1/8*\height)[anchor=west] {$11$};
  \node at (\originxfive+\width/2,\originy+\height)[anchor=south] {$S_{1,2}=11$};

\end{tikzpicture}
  \vskip 1.5em
  \ifarXiv
  % blank
  \else
  \tablesize{\footnotesize}
  \fi
  \setlength{\tabcolsep}{5pt}
  \begin{tabular}{c || c c | c || c c c c | c c || c c c c | c c c c }
    \hline
    $p$ & $s_1$ & $s_2$ & $t$ & $i_1^+$ & $i_1^-$ & $i_2^+$ & $i_2^-$ & $i_{12}^+$ & $i_{12}^-$
      & $r^+$ & $u_1^+$ & $u_2^+$ & $c^+$ & $r^-$ & $u_1^-$ & $u_2^-$ & $c^-$\\
    \hline\hline
    \nf{1}{4}&0&0&0&1&0&1&1&2&1&1&0&0&1&0&0&1&0 \\
    \nf{1}{4}&0&1&0&1&0&1&1&2&1&1&0&0&1&0&0&1&0 \\
    \nf{1}{4}&1&0&1&1&0&1&1&2&1&1&0&0&1&0&0&1&0 \\
    \nf{1}{4}&1&1&1&1&0&1&1&2&1&1&0&0&1&0&0&1&0 \\
    \hline\hline
    \multicolumn{4}{c||}{\scriptsize Expected values}
      & 1 &0 & 1 & 1 & 2 & 1 & 1 & 0 & 0 & 1 & 0 & 0 & 1 & 0 \\
    \hline
  \end{tabular}
  \vskip 5pt{\small
    $\RI=1\bit \qquad \UIo=0\bit \qquad \UIt=-1\bit \qquad \CI=1\bit$}
  \caption{\label{fig:unq} Example \unq{}.  \emph{Top}:\ the probability mass diagrams for every
    single possible realisation.  \emph{Middle}:\ for each realisation, the PPID using specificity
    and ambiguity is evaluated (see \fig{xor}).  \emph{Bottom}:\ the atoms of (average) partial
    infromation obtained through recombination of the averages.}
\end{figure}
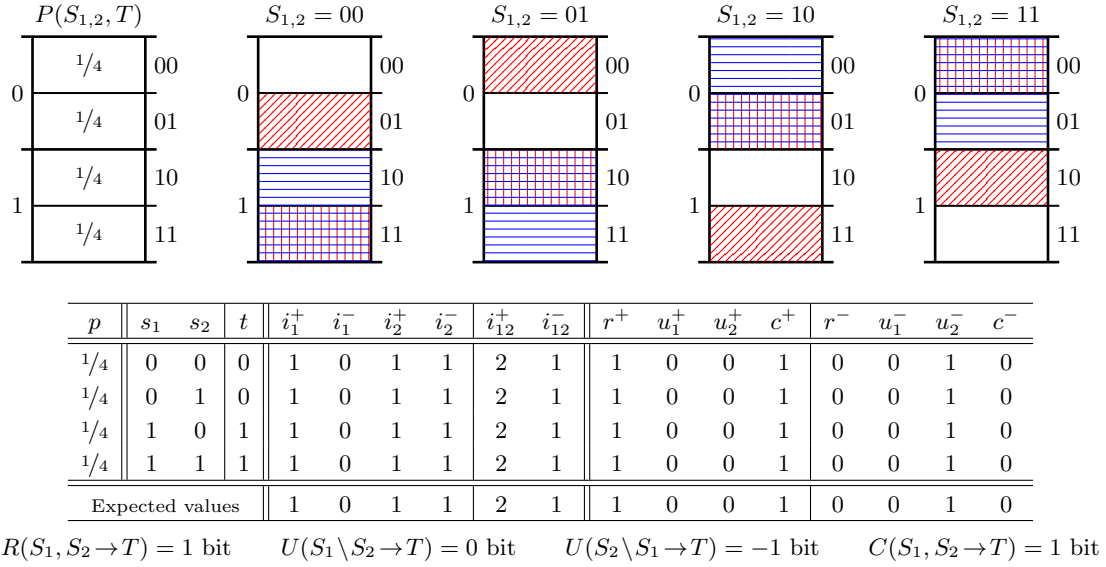

\subsection{Probability Distribution \and{}}

\fig{and} shows the decomposition of the probability distribution \emph{and} (\and{}).  Note that
the probability distribution \emph{or} (\textup{O}\textsc{r}) has the same decomposition as the
target distributions are isomorphic.

\begin{figure}[!b]
  \centering %% \tablesize{} %% You can specify the fontsize here, e.g.
  \begin{tikzpicture}

  \def\originy{0} 
  \def\height{3}
  \def\width{1.5}
  \def\overdrawn{0.15}

  % -----------
  
  \def\originx{0}
  
  \draw[black, line width=1pt] (\originx,\originy)
    -- (\originx,\originy+\height);
  \draw[black, line width=1pt] (\originx-\overdrawn,\originy+\height)
    -- (\originx+\width+\overdrawn,\originy+\height);
  \draw[black, line width=1pt] (\originx+\width,\originy+\height)
    -- (\originx+\width,\originy);
  \draw[black, line width=1pt] (\originx+\width+\overdrawn,\originy)
    -- (\originx-\overdrawn,\originy);

  \draw[black, line width=0.75pt] (\originx,\originy+3/4*\height)
    -- (\originx+\width+\overdrawn,\originy+3/4*\height);
  \draw[black, line width=0.75pt] (\originx,\originy+\height/2)
    -- (\originx+\width+\overdrawn,\originy+\height/2);
  \draw[black, line width=1pt] (\originx-\overdrawn,\originy+1/4*\height)
  -- (\originx+\width+\overdrawn,\originy+1/4*\height);

  \node at (\originx,5/8*\height)[anchor=east] {$0$};
  \node at (\originx,1/8*\height)[anchor=east] {$1$}; 
  \node at (\originx+\width,7/8*\height)[anchor=west] {$00$};
  \node at (\originx+\width,5/8*\height)[anchor=west] {$01$};
  \node at (\originx+\width,3/8*\height)[anchor=west] {$10$};
  \node at (\originx+\width,1/8*\height)[anchor=west] {$11$};
  \node at (\originx+\width/2,\originy+\height)[anchor=south] {$P(S_{1,2},T)$};
  \node at (\originx+\width/2,\originy+7/8*\height)[anchor=center] {$\nf{1}{4}$};
  \node at (\originx+\width/2,\originy+5/8*\height)[anchor=center] {$\nf{1}{4}$};
  \node at (\originx+\width/2,\originy+3/8*\height)[anchor=center] {$\nf{1}{4}$};
  \node at (\originx+\width/2,\originy+1/8*\height)[anchor=center] {$\nf{1}{4}$}; 

  % -----------
  
  \def\originxtwo{3}
  
  \draw[black, line width=1pt] (\originxtwo,\originy)
    -- (\originxtwo,\originy+\height);
  \draw[black, line width=1pt] (\originxtwo-\overdrawn,\originy+\height)
    -- (\originxtwo+\width+\overdrawn,\originy+\height);
  \draw[black, line width=1pt] (\originxtwo+\width,\originy+\height)
    -- (\originxtwo+\width,\originy);
  \draw[black, line width=1pt] (\originxtwo+\width+\overdrawn,\originy)
    -- (\originxtwo-\overdrawn,\originy);

  \draw[black, line width=0.75pt] (\originxtwo,\originy+3/4*\height)
    -- (\originxtwo+\width+\overdrawn,\originy+3/4*\height);
  \draw[black, line width=0.75pt] (\originxtwo,\originy+\height/2)
    -- (\originxtwo+\width+\overdrawn,\originy+\height/2);
  \draw[black, line width=1pt] (\originxtwo-\overdrawn,\originy+1/4*\height)
  -- (\originxtwo+\width+\overdrawn,\originy+1/4*\height);

  \draw[pattern=north east lines, pattern color=red] (\originxtwo,\originy+\height/2) rectangle
    (\originxtwo+\width,\originy+3/4*\height);
  \draw[pattern=vertical lines, pattern color=red] (\originxtwo,\originy) rectangle
    (\originxtwo+\width,\originy+1/4*\height);
  \draw[pattern=north west lines, pattern color=blue] (\originxtwo,\originy+1/4*\height) rectangle
    (\originxtwo+\width,\originy+\height/2);
  \draw[pattern=horizontal lines, pattern color=blue] (\originxtwo,\originy) rectangle
    (\originxtwo+\width,\originy+1/4*\height);
  
  \node at (\originxtwo,5/8*\height)[anchor=east] {$0$};
  \node at (\originxtwo,1/8*\height)[anchor=east] {$1$}; 
  \node at (\originxtwo+\width,7/8*\height)[anchor=west] {$00$};
  \node at (\originxtwo+\width,5/8*\height)[anchor=west] {$01$};
  \node at (\originxtwo+\width,3/8*\height)[anchor=west] {$10$};
  \node at (\originxtwo+\width,1/8*\height)[anchor=west] {$11$};
  \node at (\originxtwo+\width/2,\originy+\height)[anchor=south] {$S_{1,2}=00$};
  
  % -----------
  
  \def\originxthree{6}
  
  \draw[black, line width=1pt] (\originxthree,\originy)
    -- (\originxthree,\originy+\height);
  \draw[black, line width=1pt] (\originxthree-\overdrawn,\originy+\height)
    -- (\originxthree+\width+\overdrawn,\originy+\height);
  \draw[black, line width=1pt] (\originxthree+\width,\originy+\height)
    -- (\originxthree+\width,\originy);
  \draw[black, line width=1pt] (\originxthree+\width+\overdrawn,\originy)
    -- (\originxthree-\overdrawn,\originy);

  \draw[black, line width=0.75pt] (\originxthree,\originy+3/4*\height)
    -- (\originxthree+\width+\overdrawn,\originy+3/4*\height);
  \draw[black, line width=0.75pt] (\originxthree,\originy+\height/2)
    -- (\originxthree+\width+\overdrawn,\originy+\height/2);
  \draw[black, line width=1pt] (\originxthree-\overdrawn,\originy+1/4*\height)
  -- (\originxthree+\width+\overdrawn,\originy+1/4*\height);

  \draw[pattern=north east lines, pattern color=red] (\originxthree,\originy+3/4*\height) rectangle
    (\originxthree+\width,\originy+\height);
  \draw[pattern=north east lines, pattern color=red] (\originxthree,\originy+1/4*\height) rectangle
    (\originxthree+\width,\originy+\height/2);
  \draw[pattern=north west lines, pattern color=blue] (\originxthree,\originy+1/4*\height) rectangle
    (\originxthree+\width,\originy+\height/2);
  \draw[pattern=horizontal lines, pattern color=blue] (\originxthree,\originy) rectangle
    (\originxthree+\width,\originy+1/4*\height);
  
  \node at (\originxthree,5/8*\height)[anchor=east] {$0$};
  \node at (\originxthree,1/8*\height)[anchor=east] {$1$}; 
  \node at (\originxthree+\width,7/8*\height)[anchor=west] {$00$};
  \node at (\originxthree+\width,5/8*\height)[anchor=west] {$01$};
  \node at (\originxthree+\width,3/8*\height)[anchor=west] {$10$};
  \node at (\originxthree+\width,1/8*\height)[anchor=west] {$11$};
  \node at (\originxthree+\width/2,\originy+\height)[anchor=south] {$S_{1,2}=01$};

  % -----------
  
  \def\originxfour{9}
  
  \draw[black, line width=1pt] (\originxfour,\originy)
    -- (\originxfour,\originy+\height);
  \draw[black, line width=1pt] (\originxfour-\overdrawn,\originy+\height)
    -- (\originxfour+\width+\overdrawn,\originy+\height);
  \draw[black, line width=1pt] (\originxfour+\width,\originy+\height)
    -- (\originxfour+\width,\originy);
  \draw[black, line width=1pt] (\originxfour+\width+\overdrawn,\originy)
    -- (\originxfour-\overdrawn,\originy);

  \draw[black, line width=0.75pt] (\originxfour,\originy+3/4*\height)
    -- (\originxfour+\width+\overdrawn,\originy+3/4*\height);
  \draw[black, line width=0.75pt] (\originxfour,\originy+\height/2)
    -- (\originxfour+\width+\overdrawn,\originy+\height/2);
  \draw[black, line width=1pt] (\originxfour-\overdrawn,\originy+1/4*\height)
  -- (\originxfour+\width+\overdrawn,\originy+1/4*\height);

  \draw[pattern=north east lines, pattern color=red] (\originxfour,\originy+\height/2) rectangle
    (\originxfour+\width,\originy+3/4*\height);
  \draw[pattern=vertical lines, pattern color=red] (\originxfour,\originy) rectangle
    (\originxfour+\width,\originy+1/4*\height);
  \draw[pattern=north west lines, pattern color=blue] (\originxfour,\originy+\height) rectangle
    (\originxfour+\width,\originy+\height/2);

  \node at (\originxfour,5/8*\height)[anchor=east] {$0$};
  \node at (\originxfour,1/8*\height)[anchor=east] {$1$}; 
  \node at (\originxfour+\width,7/8*\height)[anchor=west] {$00$};
  \node at (\originxfour+\width,5/8*\height)[anchor=west] {$01$};
  \node at (\originxfour+\width,3/8*\height)[anchor=west] {$10$};
  \node at (\originxfour+\width,1/8*\height)[anchor=west] {$11$};
  \node at (\originxfour+\width/2,\originy+\height)[anchor=south] {$S_{1,2}=10$};
  
  % -----------
  
  \def\originxfive{12}
  
  \draw[black, line width=1pt] (\originxfive,\originy)
    -- (\originxfive,\originy+\height);
  \draw[black, line width=1pt] (\originxfive-\overdrawn,\originy+\height)
    -- (\originxfive+\width+\overdrawn,\originy+\height);
  \draw[black, line width=1pt] (\originxfive+\width,\originy+\height)
    -- (\originxfive+\width,\originy);
  \draw[black, line width=1pt] (\originxfive+\width+\overdrawn,\originy)
    -- (\originxfive-\overdrawn,\originy);

  \draw[black, line width=0.75pt] (\originxfive,\originy+3/4*\height)
    -- (\originxfive+\width+\overdrawn,\originy+3/4*\height);
  \draw[black, line width=0.75pt] (\originxfive,\originy+\height/2)
    -- (\originxfive+\width+\overdrawn,\originy+\height/2);
  \draw[black, line width=1pt] (\originxfive-\overdrawn,\originy+1/4*\height)
  -- (\originxfive+\width+\overdrawn,\originy+1/4*\height);
 
  \draw[pattern=vertical lines, pattern color=red] (\originxfive,\originy+3/4*\height) rectangle
    (\originxfive+\width,\originy+\height);
  \draw[pattern=vertical lines, pattern color=red] (\originxfive,\originy+1/4*\height) rectangle
    (\originxfive+\width,\originy+\height/2);
  \draw[pattern=horizontal lines, pattern color=blue] (\originxfive,\originy+\height) rectangle
    (\originxfive+\width,\originy+\height/2);
  
  \node at (\originxfive,5/8*\height)[anchor=east] {$0$};
  \node at (\originxfive,1/8*\height)[anchor=east] {$1$}; 
  \node at (\originxfive+\width,7/8*\height)[anchor=west] {$00$};
  \node at (\originxfive+\width,5/8*\height)[anchor=west] {$01$};
  \node at (\originxfive+\width,3/8*\height)[anchor=west] {$10$};
  \node at (\originxfive+\width,1/8*\height)[anchor=west] {$11$};
  \node at (\originxfive+\width/2,\originy+\height)[anchor=south] {$S_{1,2}=11$};

\end{tikzpicture}
  \vskip 1.5em
  \ifarXiv
  % blank
  \else
  \tablesize{\footnotesize}
  \fi
  \setlength{\tabcolsep}{5pt}
  \begin{tabular}{c || c c | c || c c c c | c c || c c c c | c c c c }
    \hline
    $p$ & $s_1$ & $s_2$ & $t$ & $i_1^+$ & $i_1^-$ & $i_2^+$ & $i_2^-$ & $i_{12}^+$ & $i_{12}^-$
      & $r^+$ & $u_1^+$ & $u_2^+$ & $c^+$ & $r^-$ & $u_1^-$ & $u_2^-$ & $c^-$\\
    \hline\hline
    \nf{1}{4}&0&0&0&1&$\lg\nf{3}{2}$&1&$\lg\nf{3}{2}$&2&$\lg 3$&1&0&0&1&$\lg\nf{3}{2}$&0&0&1 \\
    \nf{1}{4}&0&1&0&1&$\lg\nf{3}{2}$&1&$\lg 3$       &2&$\lg 3$&1&0&0&1&$\lg\nf{3}{2}$&0&1&0 \\
    \nf{1}{4}&1&0&0&1&$\lg 3$       &1&$\lg\nf{3}{2}$&2&$\lg 3$&1&0&0&1&$\lg\nf{3}{2}$&1&0&0 \\
    \nf{1}{4}&1&1&1&1&0             &1&0             &2&0      &1&0&0&1&0             &0&0&0 \\
    \hline\hline
    \multicolumn{4}{c||}{\scriptsize Expected values}
      & 1 &0.689 & 1 & 0.689 & 2 & 1.189 & 1 & 0 & 0 & 1 & 0.439 & 0.250 & 0.250 & 0.25 \\
    \hline
  \end{tabular}
  \vskip 5pt{\small
    $\RI\!=\!0.561\bit \quad \UIo\!=\!-0.25\bit \quad \UIt\!=\!-0.25\bit \quad \CI\!=\!0.75\bit$}
  \caption{\label{fig:and} Example \and{}.  \emph{Top}:\ the probability mass diagrams for every
    single possible realisation.  \emph{Middle}:\ for each realisation, the PPID using specificity
    and ambiguity is evaluated (see \fig{xor}).  \emph{Bottom}:\ the atoms of (average) partial
    infromation obtained through recombination of the averages.}
\end{figure}
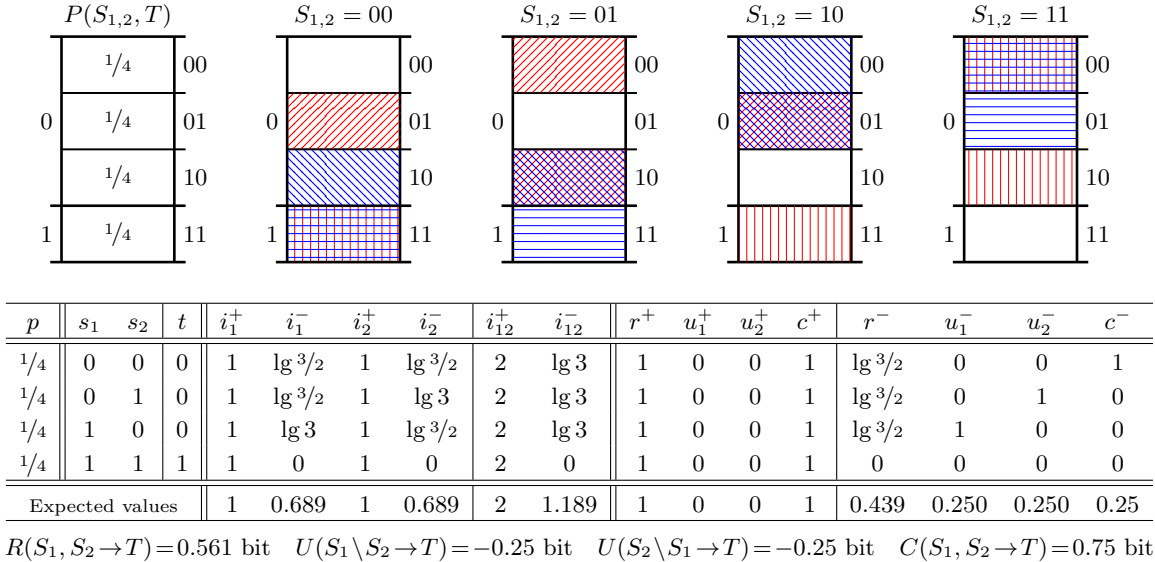

%%%%%%%%%%%%%%%%%%%%%%%%%%%%%%%%%%%%%%%%%%%%%%%%%%%%%%%%%%%%%%%%%%%%%%%%%%%%%%%%%%%%%%%%%%%%%%%%%%%%
%%%%%%%%%%%%%%%%%%%%%%%%%%%%%%%%%%%%%%%%%%%%%%%%%%%%%%%%%%%%%%%%%%%%%%%%%%%%%%%%%%%%%%%%%%%%%%%%%%%%

\ifarXiv

\begin{acknowledgments}
  JL was supported through the Australian Research Council DECRA grant DE160100630.  We thank
  Mikhail Prokopenko, Richard Spinney, Michael Wibral, Nathan Harding, Robin Ince, Nils
  Bertschinger, and Nihat Ay for helpful discussions relating to this manuscript.  We also thank the
  anonymous reviewers for their particularly detailed and helpful feedback.
\end{acknowledgments}

\else

\vspace{6pt} 

\acknowledgments{ JL was supported through the Australian Research Council DECRA grant DE160100630.
  We thank Mikhail Prokopenko, Richard Spinney, Michael Wibral, Nathan Harding, Robin Ince, Nils
  Bertschinger, and Nihat Ay for helpful discussions relating to this manuscript. We also thank the
  anonymous reviewers for their particularly detailed and helpful feedback. }

\authorcontributions{ CF and JL conceived the idea; CF designed, wrote and analyzed the
  computational examples; CF and JL wrote the manuscript.  }

\conflictsofinterest{The authors declare no conflict of interest. The founding sponsors had no role
  in the design of the study; in the collection, analyses, or interpretation of data; in the writing
  of the manuscript, and in the decision to publish the results.}

\fi

%%%%%%%%%%%%%%%%%%%%%%%%%%%%%%%%%%%%%%%%%%%%%%%%%%%%%%%%%%%%%%%%%%%%%%%%%%%%%%%%%%%%%%%%%%%%%%%%%%%%
%%%%%%%%%%%%%%%%%%%%%%%%%%%%%%%%%%%%%%%%%%%%%%%%%%%%%%%%%%%%%%%%%%%%%%%%%%%%%%%%%%%%%%%%%%%%%%%%%%%%

% Citations and References in Supplementary files are permitted provided that they also appear in
% the reference list here.

\bibliographystyle{mdpi}

%=====================================
% References, variant A: internal bibliography
%=====================================
\renewcommand\bibname{References}
\bibliography{local_info_decomp_spec_ambig}

\end{document}